\documentclass{article}

\usepackage{algpseudocode}

\usepackage{subfig}
\usepackage{multirow}
\usepackage{amssymb}
\usepackage{url}

\usepackage{graphicx}

\newcommand{\algo}{{\sc mcts4dm}}

\usepackage{ltxtable} 
\usepackage{capt-of}
\usepackage{booktabs}
\usepackage{hyperref}

\usepackage{amsmath,amsfonts,amssymb,amsthm,enumitem,amsopn,authblk}

\theoremstyle{plain}
\newtheorem{thm}{Theorem}[section]

\newtheorem{proposition}[thm]{Proposition}
\newtheorem{algorithm}{Algorithm}[section]

\theoremstyle{definition}
\newtheorem{definition}{Definition}[]
\newtheorem{problem}{Problem}[]

\theoremstyle{remark}
\newtheorem{example}{Example}[]

\DeclareMathOperator*{\argmax}{arg\,max}

\begin{document}

\title{Any-time Diverse Subgroup Discovery with Monte Carlo Tree Search}

\author[1]{Guillaume Bosc}
\author[1]{Jean-Fran\c{c}ois Boulicaut}
\author[2]{Chedy Ra\"{i}ssi}
\author[1]{Mehdi Kaytoue}
\affil[1]{Universit\'e de Lyon, CNRS, INSA-Lyon, LIRIS, UMR5205, F-69621, France, email: firstname.lastname@insa-lyon.fr}
\affil[2]{INRIA Nancy Grand Est, France, email: chedy.raissi@inria.fr}
\date{}
\maketitle

\begin{abstract}
The discovery of patterns that accurately discriminate one class label from another remains a challenging data mining task. Subgroup discovery (SD) is one of the frameworks that enables to elicit such interesting patterns from labeled data. A question remains fairly open: How to select an accurate heuristic search technique when exhaustive enumeration of the pattern space is infeasible? 
Existing approaches make use of beam-search, sampling, and genetic algorithms for discovering a pattern set that is non-redundant and of high quality w.r.t.  a  pattern quality measure. 
We argue that such approaches produce pattern sets that lack of diversity: Only few patterns of high quality, and different enough, are discovered. 
Our main contribution is then to formally define pattern mining as a game and to solve it with Monte Carlo tree search (MCTS). It can be seen as an exhaustive search guided by random simulations which can be stopped early (limited budget) by virtue of its \textit{best-first search} property.
We show through a comprehensive set of experiments how MCTS enables the
anytime discovery of a diverse pattern set of high quality. It outperforms other approaches when dealing with a large pattern search space and for different quality measures.
Thanks to its genericity, our MCTS approach can be used for SD but also for many other pattern mining tasks.
\end{abstract}

\section{Introduction\label{sec:introduction}}
The discovery of patterns, or descriptions, which discriminate a group of objects given a target (class label) has been widely studied as overviewed by \cite{DBLP:journals/jmlr/NovakLW09}. 
Discovering such descriptive rules can be formalized as the so-called subgroup discovery task (SD introduced by  \cite{DBLP:conf/pkdd/Wrobel97}). 
Given a set of objects, each being associated to a description and  a class label, a subgroup is a description generalization whose discriminating ability is evaluated by a quality measure (F1-score, accuracy, etc). 
In the last two decades, different aspects of SD have been studied: The description and target languages (itemset, sequences, graphs on one side, quantitative and qualitative targets on the other), the algorithms that enable the discovery of the best subgroups, and  the definition of measures that express  pattern interestingness. 
These directions of work are closely related and many of the pioneer approaches were \emph{ad hoc} solutions lacking from easy implementable generalizations (see for examples the surveys of  \cite{DBLP:journals/jmlr/NovakLW09} and \cite{DBLP:journals/datamine/DuivesteijnFK16}). 
SD still faces two important challenges: 
First, how to characterize the interest of a pattern? Secondly, how to design an accurate heuristic search technique when exhaustive enumeration of the pattern space is unfeasible?

\cite{DBLP:conf/pkdd/LemanFK08} introduced a more general framework than SD called exceptional model mining (EMM). It tackles the first issue. 
EMM aims to find patterns that cover tuples that locally induce a model that substantially differs from the model of the whole dataset, this difference being measured with a quality measure. This rich framework extends the classical SD settings to multi-labeled data and it leads to a large class of models, quality measures, and applications [\cite{DBLP:journals/datamine/LeeuwenK12,DBLP:journals/datamine/DuivesteijnFK16,KaytoueEtAl/Mach16}].
In a similar fashion to other pattern mining approaches, SD and EMM have to perform a heuristic search when exhaustive search fails. The most widely used techniques are \textit{beam search} [\cite{DBLP:journals/datamine/LeeuwenK12,DBLP:conf/sdm/MeengDK14}], \textit{genetic algorithms} [\cite{DBLP:journals/tfs/JesusGHM07,Lucas2017}], and \textit{pattern sampling} [\cite{DBLP:conf/ida/MoensB14,BendimeradEtAl/ICDM2016}]. 
	
The main goal of these heuristics is to drive the search towards the most interesting parts, i.e., the regions of the search space where patterns maximize a given quality measure. However, it often happens that the best patterns are \textit{redundant}: 
They tend to represent the same description, almost the same set of objects, and consequently slightly differ on their pattern quality measures. Several solutions have been proposed to filter out redundant subgroups,  e.g. as did \cite{DBLP:journals/kais/BringmannZ09,DBLP:journals/datamine/LeeuwenK12,DBLP:conf/sdm/MeengDK14,DBLP:conf/dis/BoscGBRPBK16}.
Basically, a neighboring function enables to keep only local optima. 
However, one may end up with a pattern set of small cardinality: This is the problem of \textit{diversity}, that is, many local optima have been missed. 

Let us illustrate this problem on Figure \ref{fig:exploration-methods}. 
The search space of patterns, which can be represented as a lattice, hides several local optima (patterns maximizing a pattern quality measure in a neighborhood). Figure \ref{fig:exploration-methods}(a) presents such optima with red dots, surrounded with redundant patterns in their neighborhood. 
Given the minimal number of objects a pattern must cover, exhaustive search algorithms, such as SD-Map [\cite{DBLP:conf/pkdd/AtzmullerP06,DBLP:conf/ismis/AtzmullerL09}], are able to traverse this search space efficiently: The monotonocity of the minimum support and upper bounds on some quality measures such as the \textit{weighted relative accuracy} ($WRAcc$) enable efficient and safe pruning of the search space. However, when the search space of patterns becomes tremendously large, either the number of patterns explodes or the search is intractable. Figure \ref{fig:exploration-methods}(b) presents beam-search, probably the most popular technique within the SD and EMM recent literature. It operates a top-down level-wise greedy exploration of the patterns with a controlled level width that penalizes diversity (although several enhancements to favor diversity have been devised [\cite{DBLP:journals/datamine/LeeuwenK12,DBLP:conf/pkdd/LeeuwenU13,DBLP:conf/sdm/MeengDK14}]). %
Genetic algorithms have been proposed as well [\cite{DBLP:journals/isci/RodriguezRRA12,DBLP:conf/hais/PachonVDL11,DBLP:journals/tfs/CarmonaGJH10}]. They give however no guarantees that all local optima will be found and they have been designed for specific pattern languages and quality measures [\cite{Lucas2017}]. Finally, pattern sampling  is attractive as it enables direct interactions with the user for using his/her preferences to drive the search [\cite{DBLP:conf/kdd/BoleyLPG11,DBLP:conf/ida/MoensB14}]. Besides, with sampling methods, a result is available anytime. However, traditional sampling methods used for pattern mining need a given probability distribution over the pattern space which depends on both the data and the measure and may be costly to compute [\cite{DBLP:conf/kdd/BoleyLPG11,DBLP:conf/ida/MoensB14}].  Each iteration is independent and draws a pattern given this probability distribution (Figure \ref{fig:exploration-methods}(c)).

In this article, we propose to support subgroup discovery with a novel search method, Monte Carlo tree search (MCTS). It has been mainly used in AI for domains such as games and planning problems, that can be represented as trees of sequential decisions [\cite{DBLP:journals/tciaig/BrownePWLCRTPSC12}]. It has been popularized as definitively successful for the game of Go in \cite{DBLP:journals/nature/SilverHMGSDSAPL16}. MCTS explores a search space by building a game tree in an incremental and asymmetric manner: The tree construction is driven by random simulations and an exploration/exploitation trade-off provided by the so called upper confidence bounds (UCB) [\cite{DBLP:conf/ecml/KocsisS06}]. The construction can be stopped anytime, e.g., when a maximal budget is reached.
As illustrated on Figure \ref{fig:exploration-methods}(d), our intuition for pattern mining is that MCTS searches for some local optima, and once found, the search can be redirected towards other local optima. This principle enables \textit{per se} a diversity of the result set: Several high quality patterns covering different parts of the data set  can be extracted. More importantly, the power of random search leads to \emph{anytime mining}: A solution is always available, it improves with time and it converges to the optimal one if given enough time and memory budget. This is  a \textit{best-first search}. Given a reasonable time and memory budget, MCTS quickly drives the search towards a diverse pattern set of high quality. Interestingly, it can consider, in theory, any pattern quality measure and pattern language (in contrast to current sampling techniques  as developped by \cite{DBLP:conf/kdd/BoleyLPG11,DBLP:conf/ida/MoensB14}).

Our main contribution is to a complete characterization of MCTS for subgroup discovery and pattern mining in general. Revisiting MCTS in such a setting is not simple and the definition of it requires  smart new policies. We show through an extensive set of experiments that MCTS is a compelling  solution for a pattern mining task and that it outperforms the state-of-the-art approaches (exhaustive search, beam search, genetic algorithm, pattern sampling) when dealing with large search space of numerical and nominal attributes and for different quality measures.

The rest of this article is organized as follows. Section \ref{Sec:Def} formally introduces the pattern set discovery problem. 
Section \ref{Sec:Mcts} then recalls the basic definitions of MCTS. 
We present our MCTS method, called \algo{}, in Section \ref{Sec:Method}.
After discussing the related work in Section \ref{sec:RW}, we report on experiments for understanding how to configure a MCTS for pattern mining (Section \ref{sec:quanti-exp}) and how does MCTS compare to competitors (Section \ref{sec:xp-comparisons}).

{\begin{figure}[t!] 
	\subfloat[Redundancy problem\label{fig:exploration-exhaustive}.]{\includegraphics[width=.49\columnwidth]{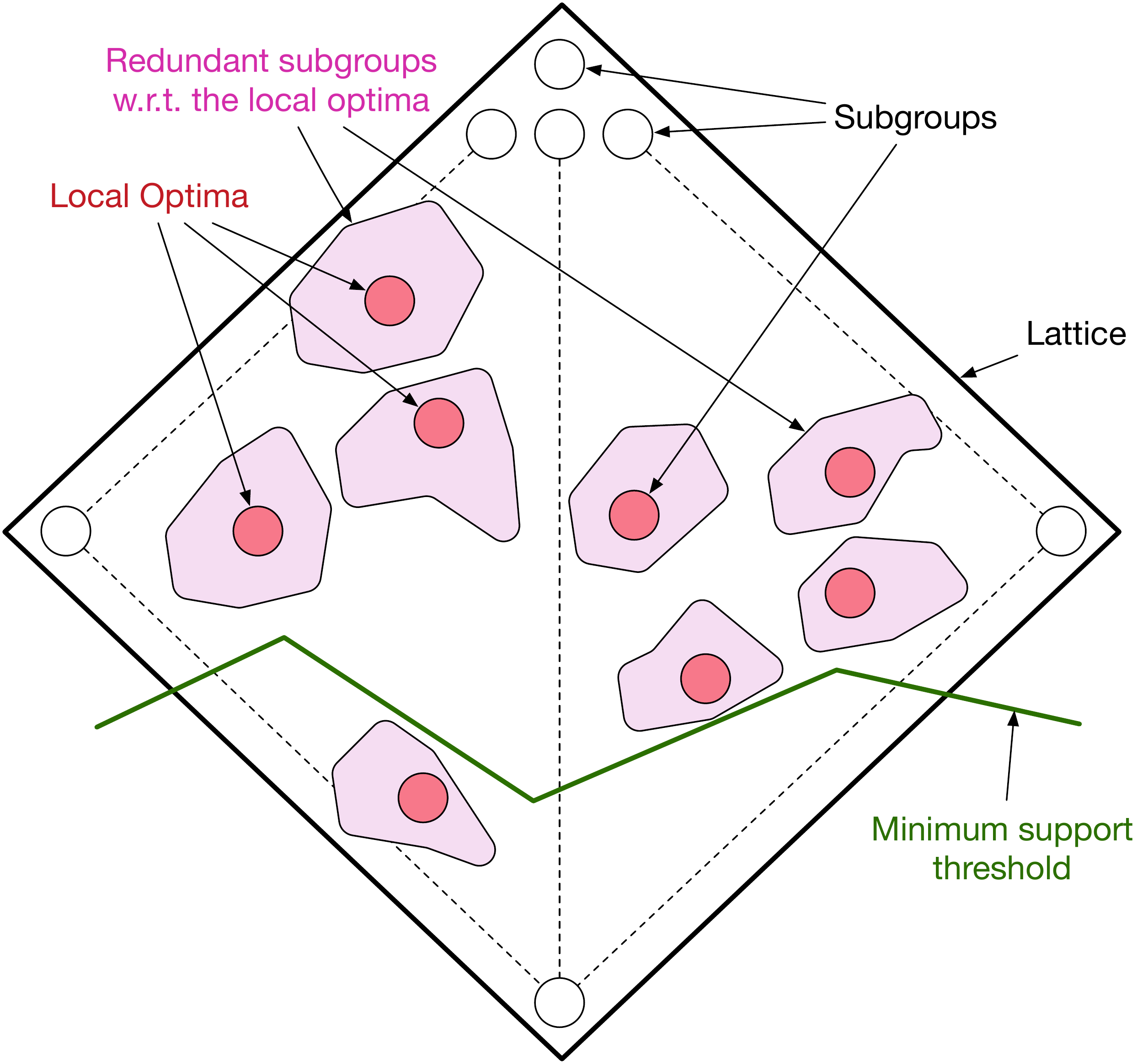}}\centering	
	\subfloat[Beam search\label{fig:exploration-beamsearch}.]{\includegraphics[width=.49\columnwidth]{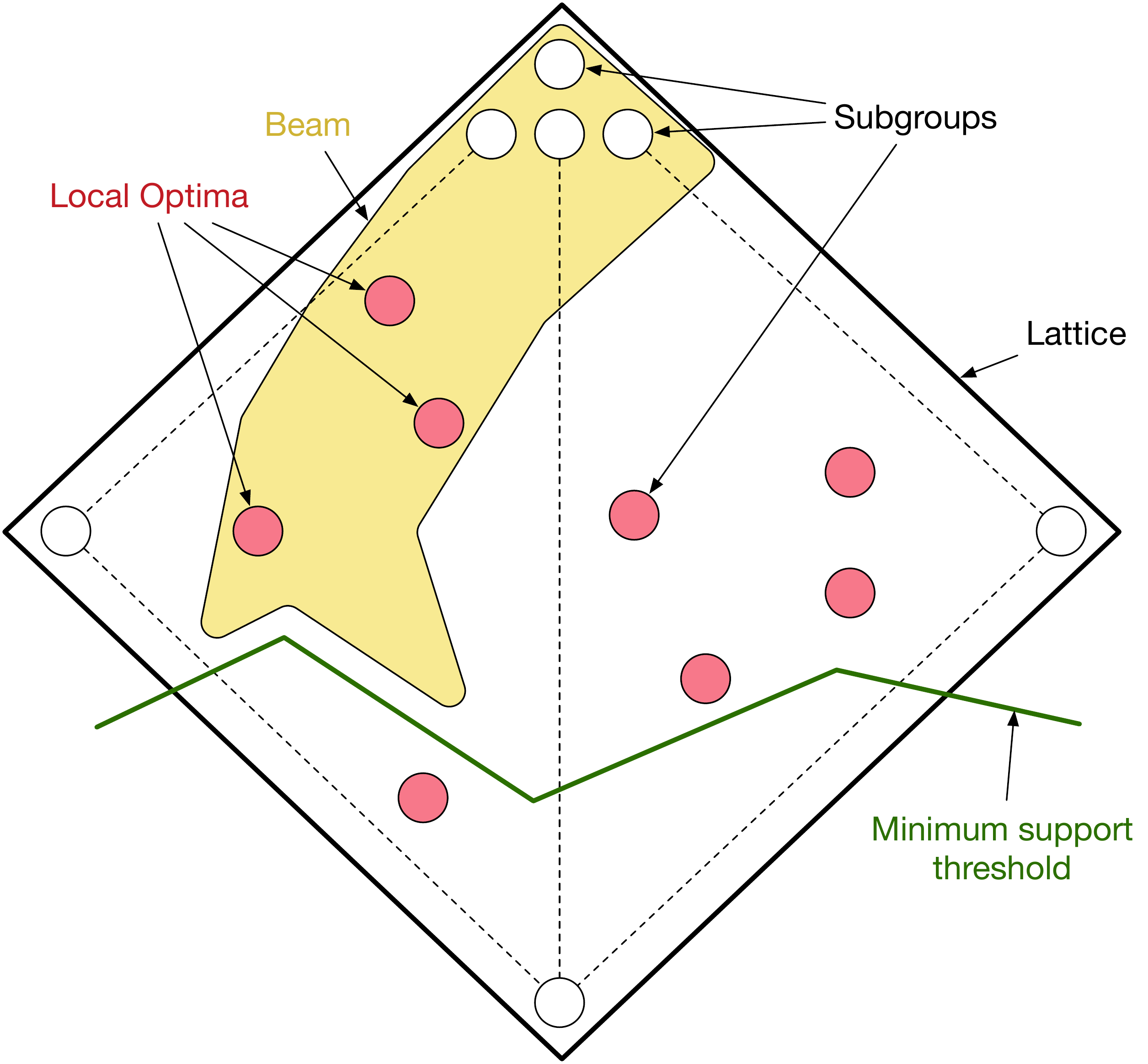}}\centering \\
	\subfloat[Sampling exploration\label{fig:exploration-sampling}.]{\includegraphics[width=.49\columnwidth]{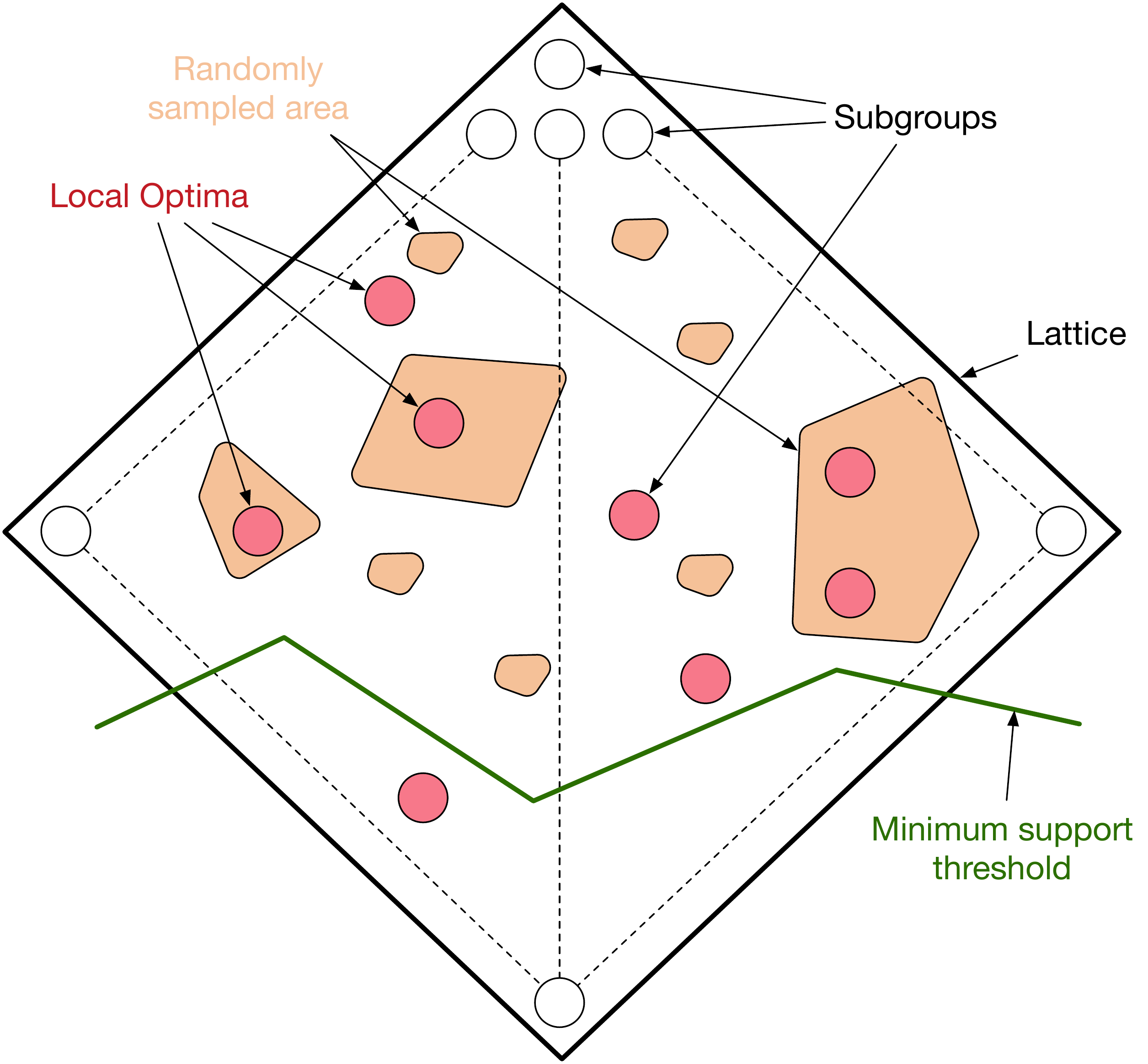}}\centering	
	\subfloat[MCTS-based exploration\label{fig:exploration-mcts}.]{\includegraphics[width=.49\columnwidth]{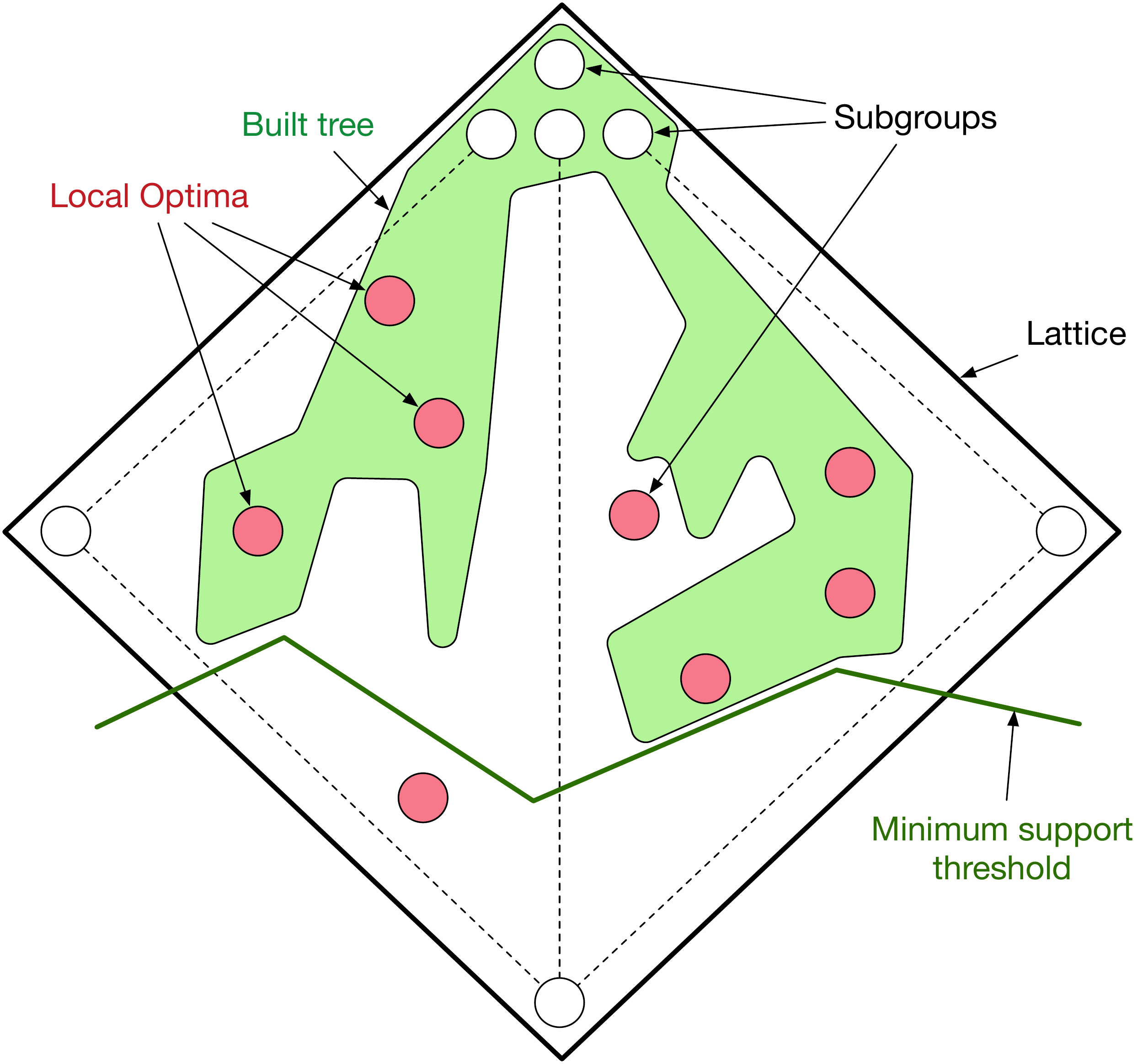}}\centering
	\caption{Illustration of different SD search algorithms.\label{fig:exploration-methods}}
\end{figure}
}

\section{Pattern set discovery\label{Sec:Def}} 

There exists several formal pattern mining frameworks  and we choose here subgroup discovery to illustrate our purpose. We provide some basic definitions and then formally define pattern set discovery.

\begin{definition}[Dataset $\mathcal{D} (\mathcal{O}, \mathcal{A}, \mathcal{C}, \mathit{class})$] 
Let $\mathcal{O}$, $\mathcal{A}$ and $\mathcal{C}$ be respectively a set of objects, a set of attributes, and a set of class labels. The domain of an attribute $a \in \mathcal{A}$ is $Dom(a)$ where $a$ is either nominal or numerical. The mapping $\mathit{class} : \mathcal{O}\mapsto {\mathcal{C}}$ associates 
each object to a unique class label. 
\label{def:data}
\end{definition}
A subgroup can be represented either by a description (the pattern) or by its coverage, also called its extent. 

\begin{table}\centering
\caption{Toy dataset\label{tab:toyData}}
\begin{tabular}{|c||ccc||c|}
\hline
ID	&	$a$		&	$b$	&	$c$	&	$\mathit{class}(.)$	\\
\hline
1	&	150	&	21	&	11	&$l_1$	\\
2	&	128	&	29	&	9	&	$l_2$	\\
3	&	136	&	24	&	10	&	$l_2$	\\
4	&	152	&	23	&	11	&	$l_3$	\\
5	&	151	&	27	&	12	&	$l_2$	\\
6 &	142	&	27	&	10	&	$l_1$	\\
\hline
\end{tabular}
\end{table}

\begin{definition}[Subgroup]
The description of a subgroup, also called pattern, is given by $d = \langle f_1, \dots, f_{|\mathcal{A}|} \rangle$ where each $f_i$ is a restriction on the value domain of the attribute $a_i \in \mathcal{A}$. 
A restriction for a nominal attribute $a_i$ is a symbol $a_i=v$ with $v \in Dom(a_i)$.
A restriction for a numerical\footnote{We consider the finite set of all intervals from the data, without greedy discretization. As shown later, better patterns can be found in that case, when using only MCTS on large datasets.} attribute $a_i$ is an interval $[l,r]$ with $l,r \in Dom(a_i)$.
The description $d$ covers a set of objects called the extent of the subgroup, denoted $ext(d) \subseteq \mathcal{O}$. The support of a subgroup is the cardinality  of its extent: $supp(d) = |ext(d)|$. %
\label{def:subgroup}
\end{definition}

The subgroup search space is structured as a lattice.
\begin{definition}[Subgroup search space]
The set of all subgroups forms a lattice, denoted as the poset $(\mathcal{S},\preceq)$. The top is the most general pattern, without restriction. Given any $s_1, s_2 \in \mathcal{S}$, we note $s_1 \prec s_2$ to denote that $s_1$ is strictly more specific, i.e. it contains more stringent restrictions.
\end{definition}
If follows that $ext(s_1) \subseteq ext(s_2)$ when $s_1 \preceq s_2$.

The ability of a subgroup to discriminate a class label is evaluated by means of a quality measure.
The weighted relative accuracy (WRAcc), intoduced by \cite{DBLP:conf/ilp/LavracFZ99}, is among the most popular measures for rule learning and subgroup discovery. Basically, WRAcc considers the precision of the subgroup w.r.t. to a class label relatively to the appearance probability of the label in the whole dataset. This difference is weighted with the support of the subgroup to avoid to consider small ones as  interesting.

\begin{definition}[WRAcc] Given a dataset $\mathcal{D} (\mathcal{O}, \mathcal{A}, \mathcal{C}, \mathit{class})$, the WRAcc of a subgroup $d$ for a label $l \in Dom(\mathcal{C})$ is given by:
\[\textit{WRAcc}(d,l) = \frac{supp(d)}{|\mathcal{O}|} \times \left(p^l_d - p^l \right)\]
where ${p^{l}_d = \frac{|\{o \in ext(d) | class(o) = l\}|}{supp(d)}}$
 and  ${p^{l} = \frac{|\{o \in \mathcal{O} | class(o) = l\}|}{|\mathcal{O}|}}$.
\end{definition}
WRAcc returns values in $[-0.25, 0,25]$, the higher and positive, the better the pattern discriminates the class label. Many quality measures other than WRAcc have been introduced in the literature of rule learning and subgroup discovery (Gini index, entropy, F score, Jaccard coefficient, etc. [\cite{DBLP:conf/pkdd/AbudawoodF09}]). Exceptional model mining (EMM) considers multiple labels (label distribution difference in \cite{DBLP:journals/datamine/LeeuwenK12}, Bayesian model difference in \cite{DBLP:conf/icdm/DuivesteijnKFL10}, etc.). The choice of a pattern quality measure, denoted $\varphi$ in what follows, is generally application dependant as explained by \cite{jf:Book-Nada}. 

\begin{example}
Consider the dataset in Table~\ref{tab:toyData} with objects in $\mathcal{O} = \{1,...,6\}$ and attributes in $\mathcal{A} = \{a,b,c\}$. Each object is labeled with a class label from $\mathcal{C} = \{ l_1, l_2, l_3 \}$. 
Consider an arbitrary subgroup with description $d = \langle [128 \leq a \leq 151], [23 \leq b \leq 29] \rangle$. 
Note that, for readability, we omit restrictions satisfied by all objects, e.g., $[9 \leq c \leq 12]$, and thus we denote that $ext(\langle \rangle) = \mathcal{O}$. 
The extent of $d$ is composed of the objects in $ext(d) = \{2,3,5,6\}$ and we have 
${WRAcc(d,l_2) = \frac{4}{6} (\frac{3}{4} - \frac{1}{2}) = \frac{1}{6}}$.
The upper part of the search space (most general subgroups) is given in Figure~\ref{Fig:searchSpace}. The direct specializations of a subgroup are given, for each attribute, by adding a restriction: Either by shrinking the interval of values to the left (take the right next value in its domain) or to the right (take the left next value). In this way, the finite set of all intervals taking borders in the attributes domain will be explored (see \cite{DBLP:conf/ijcai/KaytoueKN11}).

\begin{figure}[tb!]
		\includegraphics[width=.95\columnwidth]{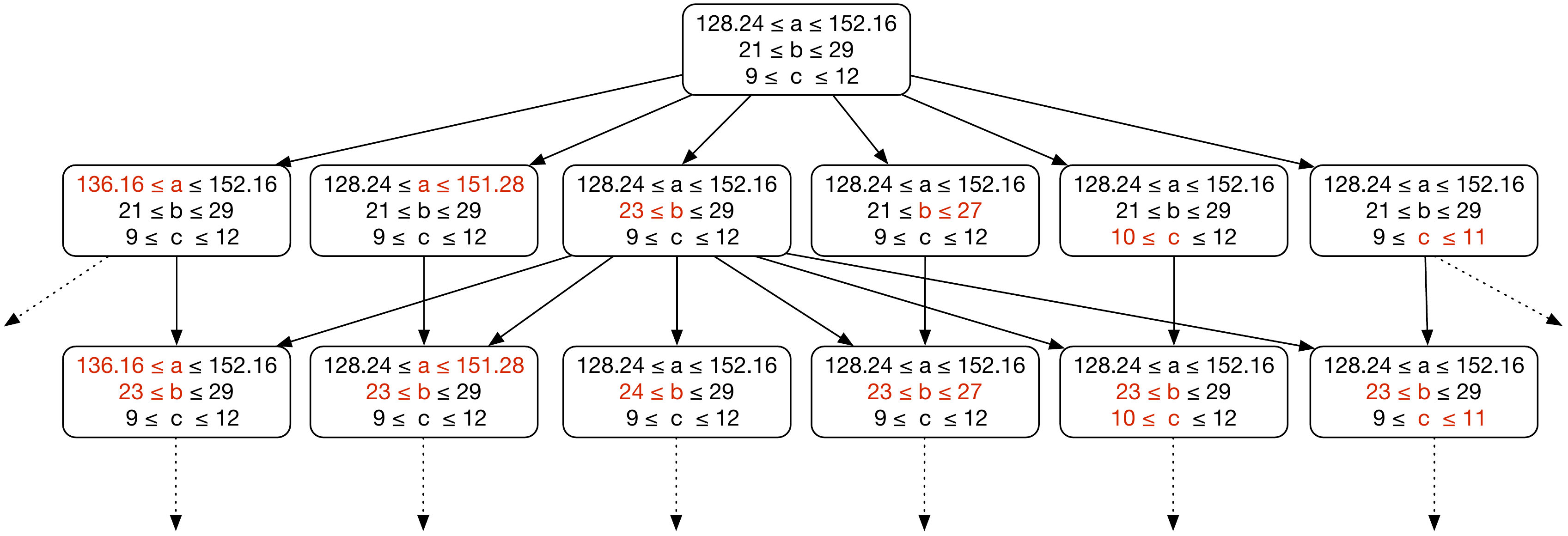}\centering	
		\caption{The upper part of the search space for Table~\ref{tab:toyData}.\label{Fig:searchSpace}}
\end{figure}
\end{example}

\textit{Pattern set discovery} consists in searching for a set of patterns $\mathcal{R} \subseteq \mathcal{S}$ of high quality on the quality measure $\varphi$ and whose patterns are not redundant. As similar patterns generally have similar values on $\varphi$, we design the pattern set discovery problem as the identification 
of the local optima w.r.t. $\varphi$. As explained below, this has two
main advantages: Redundant patterns of lower quality on $\varphi$ are pruned and the
extracted local optima are diverse and potentially interesting patterns.

\begin{definition}[Local optimum as a non redundant pattern]\label{def:redundancy}
  Let $sim : \mathcal{S}\times\mathcal{S}\rightarrow [0, 1]$ be a
  similarity measure on $\mathcal{S}$ that, given a real value
  $\Theta >0$, defines neighborhoods on $\mathcal{R}\subseteq
  \mathcal{S}: N_\mathcal{R}(x)=\{s\in\mathcal{R}\mid sim(x,s)\geq
  \Theta\}$.  $r^\star$ is a local optimum of $\mathcal{R}$ on
  $\varphi$ iff $\forall r\in N_\mathcal{R}(r^\star)$,
  $\varphi(r^\star)\geq \varphi(r)$. We denote by
  $filter(\mathcal{R})$ the set of local optima of $\mathcal{R}$ and
  by
  $redundancy(\mathcal{R})=1 - \frac{|filter(\mathcal{R})|}{|\mathcal{R}|}$
  the measure of redundancy of $\mathcal{R}$.
\end{definition}
In this paper, the similarity measure on $\mathcal{S}$ will be the Jaccard measure
defined by $$sim(r,r^\prime)=\frac{ext(r)\cap
  ext(r^\prime)}{ext(r)\cup ext(r^\prime)}$$
  
\medskip 
  
We propose to evaluate the diversity of a pattern set
$\mathcal{R}\subseteq \mathcal{S}$ by the sum of the quality of its
patterns.  Indeed, the objective is to obtain the largest set of high
quality patterns: 

\begin{definition}[Pattern set diversity]
  The diversity of a pattern set $\mathcal{R}$ is evaluated by:
  $diversity(\mathcal{R})=\sum_{r\in filter(\mathcal{R})}\varphi(r)$. \label{def:pattern-set-diversity}
\end{definition}
 
The function $filter()$ is generally defined in a greedy or heuristic way in the literature. \cite{DBLP:journals/datamine/LeeuwenK12} called  it pattern set selection and we use their implementation in this article. First all extracted patterns are sorted according to the quality measure and the best one is kept. The next patterns in the order are discarded if they are too similar with the best pattern (Jaccard similarity between pattern supports is used). Once a non similar pattern is found, it is kept for the final result and the process is reiterated: Following patterns will be compared to it. 
 
\begin{problem}[Pattern set discovery]
Compute a set of patterns $\mathcal{R}^* \subseteq \mathcal{S}$ such that $\forall r\in\mathcal{R}^*$, $r$ is a local optimum on $\varphi$ and $$\mathcal{R}^*=\mbox{argmax}_{\mathcal{R}\subseteq \mathcal{S}} diversity(\mathcal{R}).$$
\end{problem}
By construction, $\mathcal{R}^*$ maximizes diversity and it minimizes redundancy. 
Naturally, $\mathcal{R}^*$ is not unique. Existing approaches sometimes search for a pattern set of size $k$ [\cite{Lucas2017}],  with a minimum support threshold $minSupp$  [\cite{DBLP:conf/pkdd/AtzmullerP06}].

\section{Monte Carlo tree search} \label{Sec:Mcts}

MCTS is a search method used in several domains to find an optimal decision (see the survey by \cite{DBLP:journals/tciaig/BrownePWLCRTPSC12}).
It merges theoretical results from decision theory [\cite{DBLP:books/daglib/0023820}], game theory, Monte Carlo [\cite{DBLP:journals/pami/Abramson90}] and bandit-based methods [\cite{DBLP:journals/ml/AuerCF02}].
MCTS is a powerful search method because it enables the use of random simulations for characterizing a trade-off between the exploration of the search tree and the exploitation of an interesting solution, based on past observations. Considering a two-players game (e.g., Go): The goal of MCTS is to find the best action to play given a current game state.
MCTS proceeds in several (limited) iterations that build a partial game tree (called the search tree) depending on the results of previous iterations. 
The nodes represent game states. 
The root node is the current game state.
The children of a node are the game states accessible from this node by playing an available action. 
The leaves are the terminal game states (game win/loss/tie). 
Each iteration, consisting of 4 steps (see Figure~\ref{MCTS_Fig}), leads to the generation of a new node in the search tree (depending on the exploration/exploitation trade-off due to the past iterations) followed by a simulation (sequence of actions up to a terminal node).
Any node $s$ in the search tree is provided with two values: The number $N(s)$ of times it has been visited, and a value $Q(s)$ that corresponds to the aggregation of rewards of all simulations walked through $s$ so far (e.g., the proportion of wins obtained for all simulations walked through $s$). 
The aggregated reward of each node is updated through the iterations such that it becomes more and more accurate.
Once the computation budget is reached, MCTS returns the best move that leads to the child of the root node with the best aggregated reward $Q(.)$. 

In the following, we detail the 4 steps of a MCTS iteration applied to a game.
Algorithm~\ref{algoMCTS} gives the pseudo code of the most popular algorithm in the MCTS family, namely UCT (upper confidence bound for trees), as given in \cite{DBLP:conf/ecml/KocsisS06}.

\begin{algorithm}

\begin{algorithmic}[1]  
\Function{Mcts}{$\mathit{budget}$}
	\State create root node $s_0$ for current state
	\While{within computational budget $\mathit{budget}$}
		\State $s_{sel} \leftarrow$ \Call{Select}{$s_0$}
		\State $s_{exp} \leftarrow$ \Call{Expand}{$s_{sel}$}
		\State $\Delta \leftarrow$ \Call{RollOut}{$s_{exp}$}
		\State \Call{Update}{$s_{exp}, \Delta$}
	\EndWhile
	\State \Return the action that reaches the child $s$ of $s_0$ with the highest $Q(s)$
\EndFunction
\Statex
\Function{Select}{$s$}
	\While{$s$ is non-terminal}
		\If{$s$ is not fully expanded}  {~} \Return $s$
		\Else {~} $s \leftarrow$ \Call{BestChild}{$s$}
		\EndIf
	\EndWhile
	\State \Return $s$	
\EndFunction
\Statex
\Function{Expand}{$s_{sel}$}
	\State randomly choose $s_{exp}$ from non expanded children of $s_{sel}$
	\State add new child $s_{exp}$ to $s_{sel}$
	\State \Return $s_{exp}$	
\EndFunction
\Statex
\Function{RollOut}{$s$}
	\State $\Delta \leftarrow 0$
	\While{$s$ is non-terminal}
		\State choose randomly a child $s'$ of $s$
		\State $s \leftarrow s'$
	\EndWhile
	\State \Return the reward of the terminal state $s$
\EndFunction
\Statex
\Function{Update}{$s, \Delta$}
	\While{$s$ is not null}
		\State $Q(s) \leftarrow \frac{N(s) \times Q(s) + \Delta}{N(s) + 1}$
		\State $N(s) \leftarrow N(s) + 1$
		\State $s \leftarrow$ parent of $s$
	\EndWhile
\EndFunction
\Statex
\Function{BestChild}{$s$}
	\State \Return $\argmax\limits_{s' \in\ children\ of\ s} UCB(s,s')$
\EndFunction
\end{algorithmic}
\center{\textsc{UCT}:  The popular MCTS algorithm.\label{algoMCTS}}
\end{algorithm}

\begin{figure}[b!] \centering
\includegraphics[width=0.95\columnwidth]{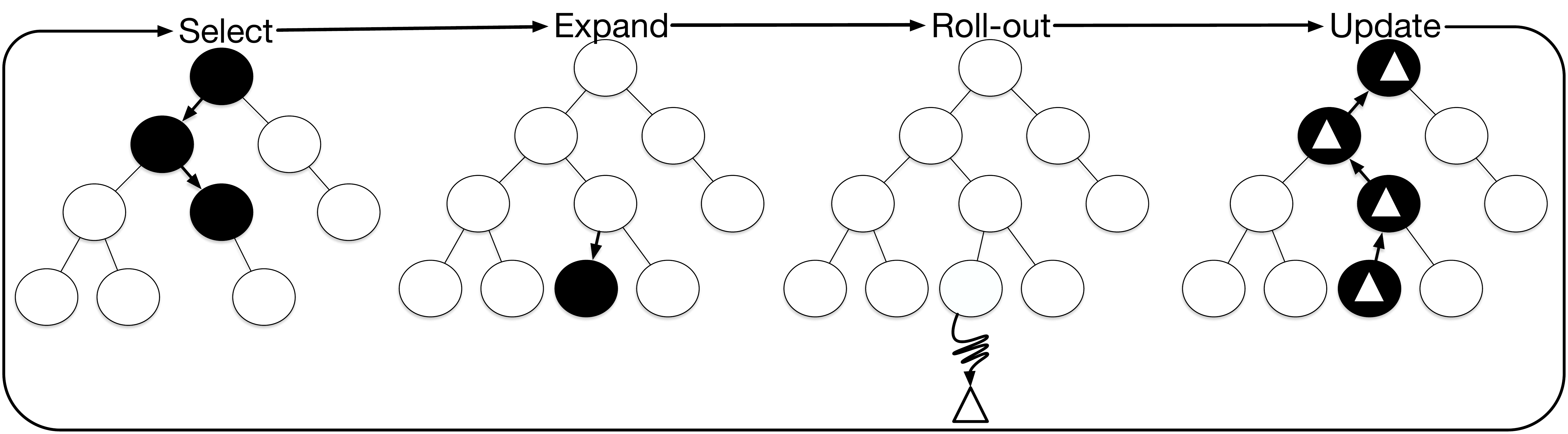}
\caption{One MCTS iteration (taken from~\cite{DBLP:journals/tciaig/BrownePWLCRTPSC12}). \label{MCTS_Fig}}
\end{figure}

\smallskip\noindent\textbf{The \textsc{Select} policy. } 
Starting from the root node, the \textsc{Select} method recursively selects an action (an edge) until the selected node is either a terminal game state or is not fully expanded (i.e., some children of this node are not yet expanded in the search tree). 
The selection of a child of a node $s$ is based on the exploration/exploitation trade-off. For that, upper confidence bounds (UCB) are used. They bound the regret of choosing a non-optimal child.
The original UCBs used in MCTS are the UCB1 from \cite{DBLP:journals/ml/AuerCF02} and the UCT from \cite{DBLP:conf/ecml/KocsisS06}:
\[UCT(s,s') = Q(s') + 2C_p\sqrt{\frac{2 \ln N(s)}{N(s')}}\]
where $s'$ is a child of a node $s$ and $C_p > 0$ is a constant (generally, $C_p = \frac{1}{\sqrt{2}}$). 
This step selects the most urgent node to be expanded, called $s_{sel}$ in the following, considering both the exploitation of interesting actions (given by the first term in UCT) and the exploration of lightly explored areas of the search space (given by the second term in UCT) based on the result of past iterations. The constant $C_p$ can be adjusted to lower or increase the exploration weight
in the exploration/exploitation trade-off . Note that when $C_p = \frac{1}{2}$, the UCT is called UCB1.

\smallskip\noindent\textbf{The \textsc{Expand} policy.} 
A new child, denoted $s_{exp}$, of the selected node $s_{sel}$ is added to the tree according to the available actions. 
The child $s_{exp}$ is randomly picked among all available children of $s_{sel}$ not yet expanded in the search tree. 

\smallskip\noindent\textbf{The \textsc{RollOut} policy.} 
From this expanded node $s_{exp}$, a simulation is played based on a specific policy. This simulation consists of exploring the search tree (playing a sequence of actions) from $s_{exp}$ until a terminal state is reached. It returns the reward $\Delta$ of this terminal state: $\Delta = 1$ if the terminal state is a win, $\Delta = 0$ otherwise.

\smallskip\noindent\textbf{The \textsc{Update} policy.} 
The reward $\Delta$ is back-propagated to the root, updating for each parent the number of visits $N(.)$ (incremented by 1) and the aggregation reward $Q(.)$ (the new proportion of wins).

\medskip\noindent\textit{Example. } Figure~\ref{MCTS_Fig} depicts a MCTS iteration. Each node has no more than 2 children. In this scenario, the search tree is already expanded:  We consider the $9^{th}$ iteration since $8$ nodes of the tree have been already added. The first step consists in running the \textsc{Select} method starting from the root node. Based on a UCB, the \emph{selection policy} chooses the left child of the root. As this node is fully expanded, the algorithm randomly selects a new node among the children of this node: Its right child. This selected node $s_{sel}$ is not fully expanded since its left hand side child is not in the search tree yet. From this not fully expanded node $s_{sel}$, the \textsc{Expand} method adds the left hand side child $s_{exp}$ of the selected node $s_{sel}$ to expand the search tree. From this added node $s_{exp}$, a random simulation is rolled out until reaching a terminal state. The reward $\Delta$ of the terminal node is  back-propagated with \textsc{Update}.

\section{Pattern set discovery with MCTS} 
\label{Sec:Method}

Designing a MCTS approach for a pattern mining problem is different than for a combinatorial game:
The goal is not to decide, at each turn, what is the best action to play, but to explore the search space: The pattern mining problem can thus be considered as a \textit{single-turn single-player game}.
Most importantly, MCTS offers a natural way to explore the search space of patterns with the benefit of the exploitation/exploration trade-off to improve diversity while limiting redundancy. For example, an exhaustive search will maximize diversity, but it will return a very large and redundant collection (but an exhaustive search is usually impossible). In contrast, a beam search can extract a limited number of patterns but it will certainly lack diversity (empirical evidences are given later in Section \ref{sec:xp-comparisons}).

\begin{figure}
\includegraphics[width=\textwidth]{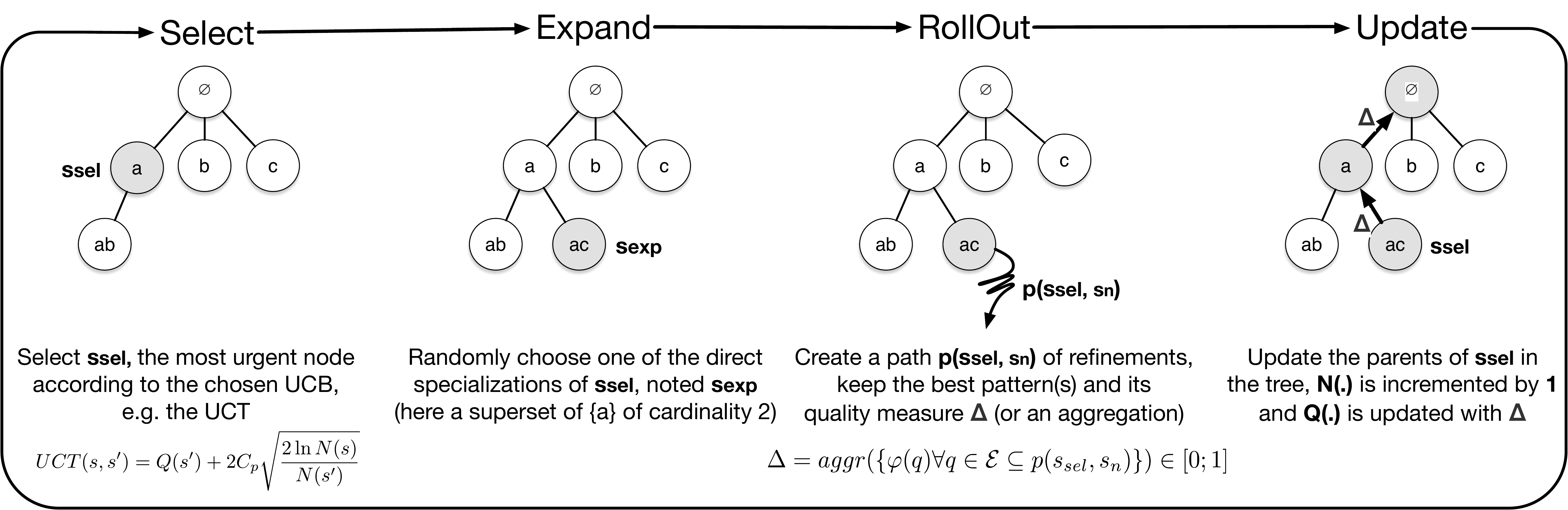}
\caption{A simple instanciation of MCTS for pattern mining.  \label{fig:mcts-itemset}}
\end{figure}

Before going into the formalization, let us illustrate how MCTS is applied to the pattern set discovery problem with Figure \ref{fig:mcts-itemset}.
We consider here itemset patterns for the sake of simplicity, that is, subgroups whose descriptions are sets of items.
We present an iteration of a MCTS for a transaction database with items $\mathcal{I} = \{a, b, c\}$.
The pattern search space is given by the lattice $\mathcal{S} = (2^\mathcal{I}, \subseteq)$.
The MCTS tree is built in a top-down fashion on this theoretical search space: 
The initial pattern, or root of the tree, is the empty set $\emptyset$.
Assume that pattern $s_{sel} = \{a\}$ has been chosen by the \textit{select} policy. 
During the \textit{expand}, one of its direct specializations in $\{ \{a,b\}, \{a,c\}\}$
is randomly chosen and added to the tree,e.g., $s_{exp} = \{a,c\}$.
During the roll out, a simulation is run from this node: 
it generates a chain  of specializations of $s_{exp}$ called a path $p(s_{sel}, s_n)$
 ({a chain is a set of comparable patterns w.r.t. $\subseteq$, or $\preceq$ in the general case}).
The quality measure $\varphi$ is computed for each pattern
of the path, and an aggregated value (max, mean, etc.) is returned and called $\Delta$.
Finally, all parents of $s_{exp}$ are updated: Their visit count $N(.)$ is incremented by one, 
while their quality estimation $Q(.)$ is recomputed with $\Delta$ (back propagation).
The new values of $N(.)$ and $Q(.)$ will directly impact the selection of the next iteration
when computing the chosen UCB, and thus the desired exploration/exploitation trade off.
When the budget is exceeded (or if the tree is fully expanded), all patterns are filtered with a chosen pattern set selection strategy ($filter(.)$).

The expected shape of the MCTS tree after a high number of iterations is illustrated in Figure~\ref{fig:exploration-mcts}. It suggests  a high diversity of the final pattern set if given enough budget (i.e., enough iterations). 
However, how to properly define each policy (select, expand, roll out and update), is not obvious. Table \ref{tab:allpolicies} sums up the different policies that we use or develop specifically for a pattern mining problem. 

\begin{table} \centering  \scriptsize
\begin{tabular}{c}
\hline {\textsc{\textbf{Select}}} \\
\hline
Choose one of the following UCB:\\
{  \textbf{UCB1}	or \textbf{UCB1-Tuned} or \textbf{SP-MCTS} or \textbf{UCT} } \\
\hline{\textsc{\textbf{Expand}}} \\
\hline
\textbf{direct-expand}: Randomly choose the next direct expansion  \\
\textbf{gen-expand}:   Randomly choose the next direct expansion until it changes the extent \\
\textbf{label-expand}: Randomly choose the next direct expansion until it changes the true positives \\
\textbf{Activate LO}: Generate each pattern only once (lectic enumeration)\\
\textbf{Activate PU}: Patterns with the same support/true positive set point to the same node\\
\hline  {\textsc{{\textbf{RollOut}}}} \\
\hline
\textbf{naive-roll-out}: Generate a random path of direct specializations of random length.\\
\textbf{direct-freq-roll-out}: Generate a random path of frequent direct specializations.\\
\textbf{large-freq-roll-out}: Generate a random paths of undirect specializations (random jumps).\\
\hline{\textsc{{\textbf{Memory}}}} \\
\hline
\textbf{no-memory}: No pattern found during the simulation is kept for the final result.\\
\textbf{top-k-memory}: Top-k patterns of a simulation are considered in memory.\\
\textbf{all-memory}: All patterns generated during the simulation are kept.\\
\hline{\textsc{{\textbf{Update}}}} \\
\hline
\textbf{max-update}: Only the maximum $\varphi$ found in a simulation is back propagated \\
\textbf{mean-update}: The average of all $\varphi$ is back-propagated\\
\textbf{top-k-mean-update}: The average of the best $k$ $\varphi$ is back-propagated\\
\hline
\end{tabular}
\caption{The different policies\label{tab:allpolicies}}
\end{table}

\subsection{The \textsc{Select} method}
The \textsc{Select} method has to select the most promising node $s_{sel}$ in terms of the exploration vs. exploitation trade-off.
For that, the well-known bounds like UCT or UCB1 can be used.
However, more sophisticated bounds have been designed for single player games.
The single-player MCTS (SP-MCTS), introduced by \cite{DBLP:conf/cg/SchaddWHCU08}, 
adds a third term to the UCB to take into account the variance $\sigma^2$ of the rewards obtained by the child so far.
SP-MCTS of a child $s'$ of a node $s$ is:
\[ \textit{SP-MCTS}(s,s') = Q(s') + C\sqrt{\frac{2 \ln N(s)}{N(s')}} + \sqrt{\sigma^2(s') + \frac{D}{N(s')}} \]
where the constant $C$ is used to weight the exploration term (it is fixed to $0.5$ in its original definition) and the term $\frac{D}{N(s')}$ inflates the standard deviation for infrequently visited children ($D$ is also a constant). 
In this way, the reward of a node rarely visited is considered as less certain: 
It is still required to explore it to get a more precise estimate of its variance. 
If the variance is still high, it means that the subspace from this node is not homogeneous w.r.t. the quality measure and further exploration is needed. 

Also, \cite{DBLP:journals/ml/AuerCF02}  designed  \emph{UCB1-Tuned} to reduce the impact of the exploration term of the original UCB1 by weighting it with either an approximation of the variance of the rewards obtained so far or the factor $1/4$. 
UCB1-Tuned of a child $s'$ of $s$ is:
\[ \textit{UCB1-Tuned}(s,s') = Q(s') + \sqrt{\frac{\ln N(s)}{N(s')} \min{(\frac{1}{4},\sigma^2(s') + \sqrt{\frac{2 \ln{N(s)}}{N(s')}})}}\]

The only requirement the pattern quality measure $\varphi$ must satisfy is, in case of UCT only, 
to take values in $[0,1]$: $\varphi$ can be normalized in this case.

\subsection{The \textsc{Expand} method}
The \textsc{Expand} step consists in adding a pattern specialization as a new node in the search tree.
In the following, we present different refinement operators, and how to avoid duplicate nodes in the search tree.

\subsubsection{The refinement operators}
A simple way to expand the selected node $s_{sel}$ is to choose uniformly an available attribute w.r.t. $s_{sel}$, that is to specialize $s_{sel}$ into $s_{exp}$ such that $s_{exp} \prec s_{sel}$: $s_{exp}$ is a refinement of $s_{sel}$. It follows that $ext(s_{exp}) \subseteq ext(s_{sel})$, and obviously  $supp(s_{exp}) \leq supp(s_{sel})$, known as the monotonocity property of the support.

\begin{definition}[Refinement operator]
A refinement operator is a function $\mathit{ref} : \mathcal{S} \rightarrow 2^\mathcal{S}$ that derives from a pattern $s$ a set of more specific patterns $\mathit{ref}(s)$ such that:
\begin{itemize}
	\item[](i) $\forall s' \in \mathit{ref}(s), s' \prec s$
	\item[](ii)$\forall s'_i, s'_j \in \mathit{ref}(s), i \neq j, s'_i \npreceq s'_j, s'_j \npreceq s'_i$
\end{itemize}
\label{def:undirec-spe}
\end{definition}

In other words, a refinement operator gives to any pattern $s$ a set of its specializations, that are pairwise incomparable (an anti-chain).
The \textit{refine} operation can be implemented in various ways given the kind of patterns we are dealing with. 
Most importantly, it can return all the direct specializations only to ensure that the exploration will, if given enough budget, explore the whole search space of patterns. Furthermore, it is unnecessary to generate infrequent patterns.

\begin{definition}[Direct-refinement operator]
A direct refinement operator is a refinement operator $\mathit{directRef} : \mathcal{S} \rightarrow 2^\mathcal{S}$ that derives from a pattern $s$ the set of direct more specific patterns $s'$ such that:
\begin{itemize}
	\item[](i) $\forall s' \in \mathit{directRef}(s)$, $s' \prec s$ 
	\item[](ii) $\not \exists s'' \in \mathcal{S}$ s.t. $ s' \prec s'' \prec s$
	\item[](iii) For any $s' \in \mathit{directRef}(s)$, $s'$ is frequent, that is $supp(s') \geq minSupp$
\end{itemize}
\label{def:direct-ref}
\end{definition}

The notion of direct refinement is well known in pattern mining. For instance, the only way to refine a nominal (resp. Boolean) attribute is to assign it a value of its domain (resp. the \textit{true} value). Refining an itemset consists in adding a item, while refining a numerical attribute can be done in two ways: Applying the minimal left change (resp. right change), that is, increasing the lower bound of the interval to the next higher value in its domain (resp. decreasing the upper bound to the next lower) as explained by \cite{DBLP:conf/ijcai/KaytoueKN11}. We still use the term \textit{restriction} to denote the operations that create a direct refinement of pattern.

\begin{definition}[The \textit{direct-expand} strategy]
We define the \textit{direct-expand} strategy as follows: From the selected node $s_{sel}$, we randomly pick a -- not yet expanded -- node $s_{exp}$ from $directRef(s_{sel})$ and add it in the search tree.
\end{definition}

As most quality measures $\varphi$ used in SD and EMM are solely based on the extent of the patterns,
considering only one pattern among all those having the same extent is enough.
However, with the \textit{direct-refinement operator}, a
large number of tree nodes may have the same extent  as their parent. 
This redundancy may bias the exploration and more iterations will be required. 
For that, we propose to use the notion of closed patterns and their generators.

\begin{definition}[Closed descriptions and their generators]
The equivalence class of a pattern $s$ is given by $[s] = \{ s' \in \mathcal{S}~|~ext(s)=ext(s') \}$.
Each equivalence class has a unique smallest element w.r.t. $\prec$ that is called the closed pattern: $s$ is said to be closed iff $\not \exists s'$ such that $s' \prec s$ and $ext(s) = ext(s')$. The non-closed patterns are called generators. 
\end{definition}

\begin{definition}[Generator-refinement operator]
A generator refinement operator is a refinement operator $\mathit{genRef} : \mathcal{S} \rightarrow 2^\mathcal{S}$ that derives from a pattern $s$ the set of more specific patterns $s'$ such that,
 $\forall s' \in \mathit{genRef}(s)$:
\begin{itemize}
	\item[](i) $s' \not\in [s]$  (different support)
	\item[](ii) $\not \exists s'' \in \mathcal{S}\backslash \mathit{genRef}(s)$ s.t. 
		$s'' \not\in [s]$, $s'' \not\in [s']$, $s' \prec s'' \prec s$ (direct next equivalence class)
	\item[](iii) $s'$ is frequent, that is $supp(s') \geq minSupp$ (frequent)
\end{itemize}
\label{def:gen-ref}
\end{definition}

\begin{definition}[The \textit{gen-expand} strategy]
To avoid the exploration of patterns with the same extent in a branch of the tree, we define the \textit{min-gen-expand} strategy as follows: From the selected node $s_{sel}$, we randomly pick a -- not yet expanded -- refined pattern from $\mathit{genRef}(s_{sel})$, called $s_{exp}$, and add it to the search tree.
\end{definition}

Finally, when facing a SD problem whose aim is to characterize a label $l \in \mathcal{C}$ we can adapt the previous refinement operator based on generators on the extents of both the subgroup and the label.
As many other measures, the WRAcc seeks to optimize the (weighted relative) precision or accuracy of the subgroup.
The accuracy is the ratio of true positives in the extent. We propose thus, for this kind of measures only,
the  \textbf{\textit{label-expand}} strategy: Basically, the pattern is refined until the set of true positives in the extent changes. This minor improvement performs very well in practice (see Section \ref{sec:quanti-exp}).

\subsubsection{Avoiding duplicates in the search tree}
We define several refinement operators to avoid the redundancy within a branch of the tree, i.e., do not expand $s_{sel}$ with a pattern whose extent is the same because the quality measure $\varphi$ will be equal.
However, another redundancy issue remains at the tree scale. 
Indeed, since the pattern search space is a lattice, a pattern can be generated in nodes from different branches of the Monte Carlo tree, that is, with different sequences of refinements, or simply permutations of refinements. 
As such, it will happen that a part of the search space is sampled several times in different branches of the tree. However, the visit count $N(s)$ of a node $s$ will not count visits of other nodes that denote exactly the same pattern: The UCB is clearly biased.  To tackle this aspect, we implement two methods: (i) Using a lectic order or (ii) detecting and unifying the duplicates within the tree. These two solutions can be used for any refinement operator. Note that enabling both these solutions at the same tame is useless since each of them ensures to avoid duplicates within the tree.

\medskip\noindent\textbf{Avoiding duplicates in the tree using a lectic order (LO). }

Pattern enumeration without duplicates is at the core of constraint-based pattern-mining [\cite{DBLP:reference/dmkdh/BoulicautJ10}]. Avoiding to generate patterns with the same extent is usually based on a total order on the set of attribute restrictions. This poset is written by $(R,\lessdot)$. 

\begin{example}
For instance, considering itemset patterns, $R = \mathcal{I}$ and a lectic order, usually the lexicographic order, is chosen on $\mathcal{I}$: 
 $a \lessdot b \lessdot c \lessdot d$ for $I = \{ a, b, c, d\}$ and $bc \lessdot ad$.
Consider that a node $s$ has been generated with a restriction  $r_i$:
we can expand the node only with restrictions $r_j$ such that   $r_i \lessdot r_j $. This total order also holds for numerical attributes by considering the minimal changes (see the work of \cite{DBLP:conf/ijcai/KaytoueKN11} for further details).
\end{example}

We can use this technique to enumerate the lattice with a depth-first search (DFS), which ensures that each element of the search space is visited exactly once.
An example is given in Figure \ref{fig:enum-expand}. 
However, it induces a strong bias: An MCTS algorithm would sample this tree instead of sampling the pattern search space. 
In other words, a small restriction w.r.t. $\lessdot$ has much less chances to be picked than a largest one.
Going back to the example in Figure \ref{fig:enum-expand} (middle), the item $a$ can be drawn only once through a complete DFS; $b$ twice; while $c$ four times (in bold). It follows that patterns on the left hand side of the tree have less chances to be generated, e.g., $prob(\{a,b\})=1/6$ while $prob(\{b,c\})=1/3$. 
These two itemsets should however have the same chance to be picked as they have the same size.
This variability is corrected by weighting the visit counts in the UCT with the normalized exploration rate (see Figure~\ref{fig:enum-expand} (right)).

\begin{definition}[Normalized exploration rate]
Let $\mathcal{S}$ be the set of all possible patterns.
The normalized exploration rate of a pattern $s$ is,
 \[   \rho_{norm}(s)  =   \frac{V_{total}(s)}{V_{lectic}(s)} = \frac{    |\{s'  | s' \preceq s , \forall s' \in \mathcal{S}\}|    }{     |\{s' |  (s \lessdot s' \wedge s' \prec s)   \vee s=s', \forall s' \in \mathcal{S}\}|      }  \]
\label{def:exploration-rate}   
\end{definition}

Given this normalized exploration rate, we can adapt the UCBs when enabling the lectic order. For example, we can define the \textit{DFS-UCT} of a child $s'$ of a pattern $s$ derived from the \textit{UCT} as follows:
\[  \textit{DFS-UCT}(s,s') = Q(s') + 2C_p\sqrt{\frac{2  \ln \left(N(s) \cdot \rho_{norm}(s)\right)}{N(s') \cdot \rho_{norm}(s')}}\]

\begin{proposition}[Normalized exploration rate for itemsets]
For itemsets, let $s_i$ be the child of $s$ obtained by playing action $r_i$ and $i$ is the rank of $r_i$ in $(R,\lessdot)$: $\rho_{norm}(s_i) = \frac{  2^{(|\mathcal{I}| - |s_i|)} }{ 2^{(|\mathcal{I}| - i - 1)}}$. 
\label{prop:rate-itemset}
\end{proposition}

\begin{proof}
Let $V_{lectic}(s_i)$ be the size of the search space sampled under $s_i$ using a lectic enumeration, and $V_{total}(s_i)$ be the size of the search space without using a lectic enumeration. 
Noting $V_{total}(s_i) = 2^{(|\mathcal{I}| - |s_i|)}$ and $V_{lectic}(s_i) = 2^{(|\mathcal{I}| - i - 1)}$ for itemsets,
we have $\rho_{norm}(s_i) = \frac{V_{total}(s_i)}{V_{lectic}(s_i)} = \frac{  2^{(|\mathcal{I}| - |s_i|)} }{ 2^{(|\mathcal{I}| - i - 1)}}$.
\qed
\end{proof}

\begin{proposition}[Normalized exploration rate for a numerical attribute]
For a single numerical attribute $a$, $\rho_{norm}(.)$ is defined as follows :
\begin{itemize}
\item Let $s' = \langle \alpha_i \leq a \leq \alpha_j \rangle$ obtained after a left change: $\rho_{norm}(s')=  1$.
\item Let $s' = \langle \alpha_i \leq a \leq \alpha_j \rangle$ obtained after a right change.
Let $n$ be the number of values from $Dom(a)$ in $[\alpha_i,\alpha_j]$: ${\rho_{norm}(s')=  \frac{n+1}{2}}$.
\end{itemize}
\label{prop:rate-num}
\end{proposition}

\begin{proof}
As explained in the proof of (Proposition \ref{prop:rate-itemset}),
$\rho_{norm}(s) = \frac{V_{total}(s)}{V_{lectic}(s)}$. For a numerical attribute, 
$V_{total}(s) = n(n+1)/2$, i.e. the number of all sub intervals. 
If $s$ was obtained after a left change, $V_{lectic}(s) = n(n+1)/2$ as both left and right changes can be applied.
If $s$ was obtained after a right change, $V_{lectic}(s)= n$, as only $n$ right changes can be applied.
It follows that $\rho_{norm}(s)  = \frac{n(n+1)/2}{n(n+1)/2} = 1$ if $s$ was obtained from a left change and $\rho_{norm}(s)  = \frac{n(n+1)/2}{n} = \frac{n+1}{2}$ otherwise.
\qed
\end{proof}

\begin{figure}\centering
\begin{tabular}{ccc} 
\includegraphics[width=.22\textwidth]{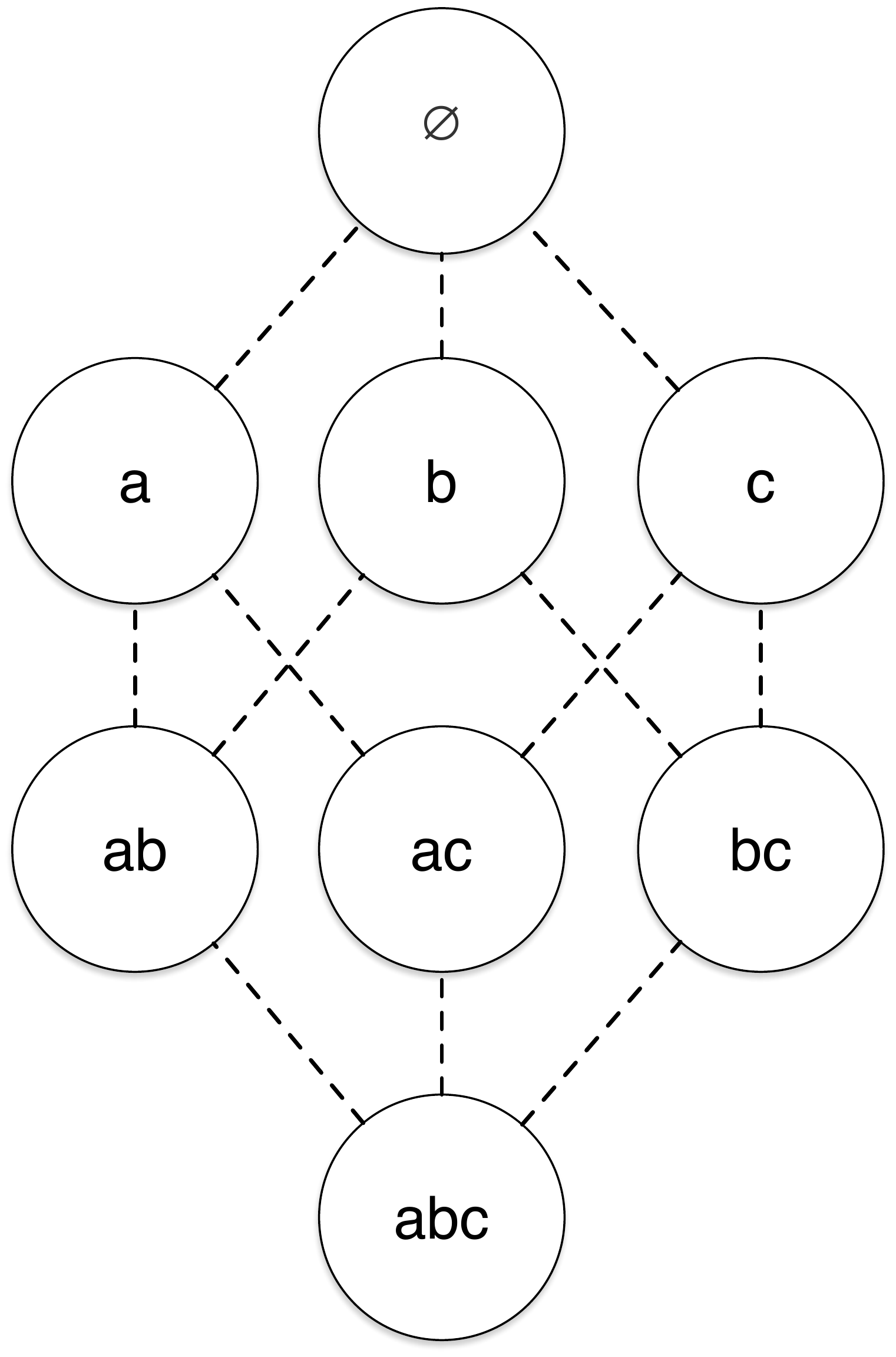} & 
\includegraphics[width=.22\textwidth]{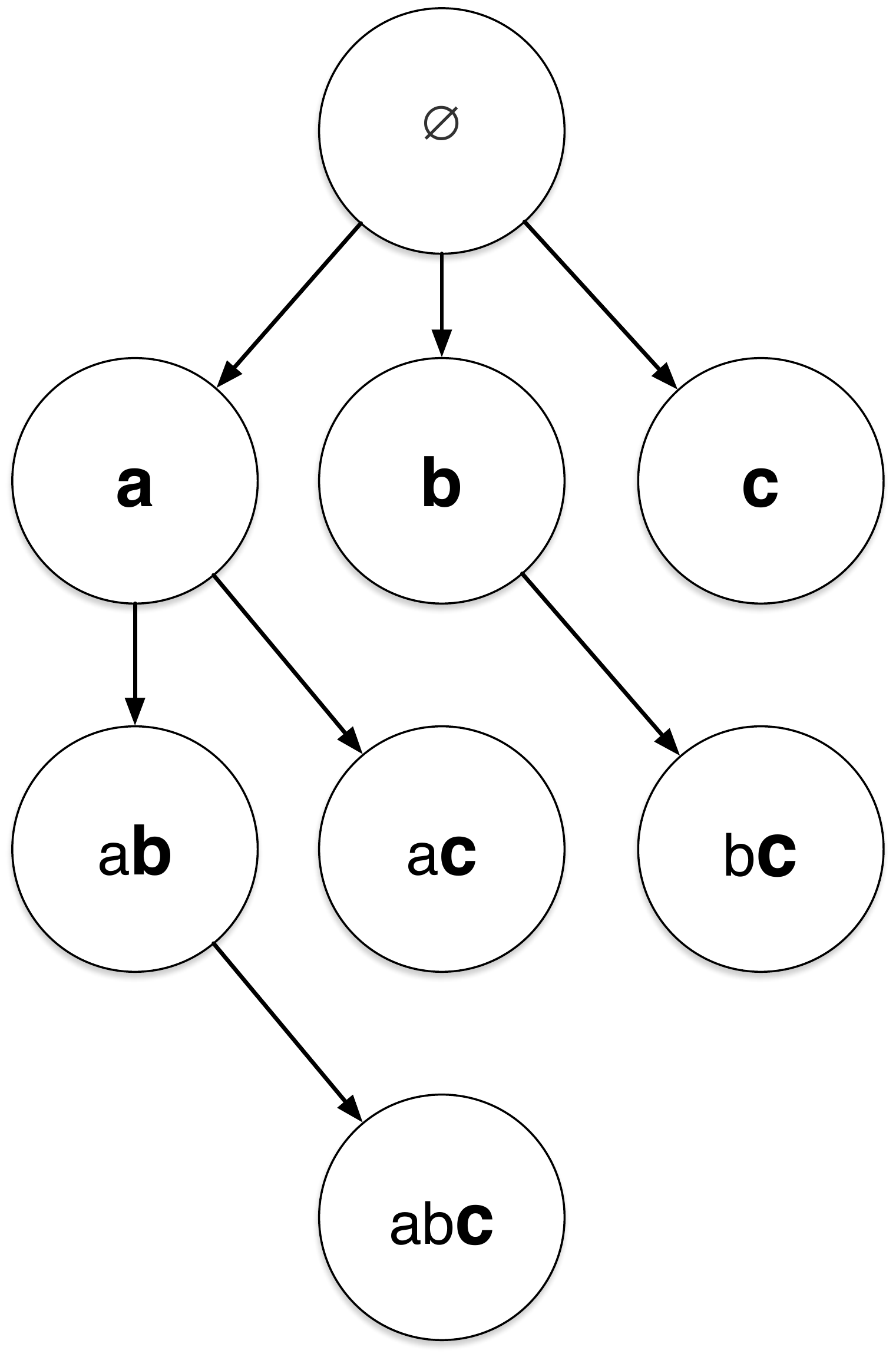} &
\includegraphics[width=0.45\textwidth]{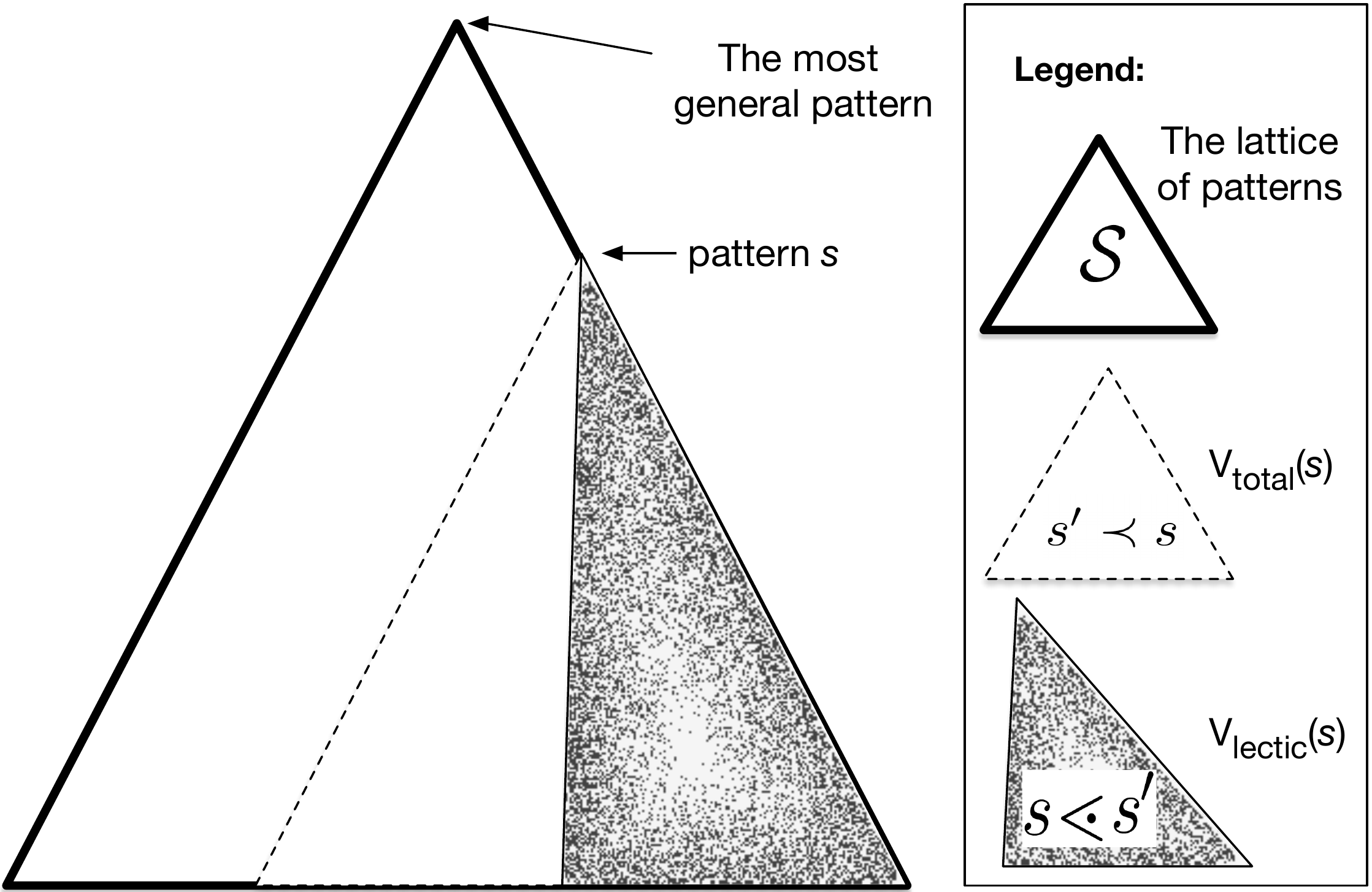}
\end{tabular}
\caption{Search space as a lattice (left), DFS of the search space (middle), and the principles of the normalized exploration rate.\label{fig:enum-expand}}
\end{figure}

\medskip\noindent\textbf{Avoiding duplicates in the tree using permutation unification (PU). }

The permutation unification is a solution that enables to keep a unique node for all duplicates of a pattern that can be expanded within several branches of the tree. 
This is inspired from \textit{Permutation AMAF} of \cite{DBLP:conf/icai/HelmboldP09},  a method used in traditional MCTS algorithms to update all the nodes that can be concerned by a play-out.
A unified node no longer has a single parent but a list of all duplicates' parent. 
This list will be used when back-propagating a reward.

This method is detailed in Algorithm \ref{algo:permuAMAF}.
Consider that the node $s_{exp}$ has been chosen as
an expansion of the selected node $s_{sel}$.
The tree generated so far is explored for finding $s_{exp}$ elsewhere in the tree: If $s_{exp}$ is not found, we proceed as usual; otherwise $s_{exp}$ becomes a pointer to the duplicate node in the tree. In our MCTS implementation, we will simply use a hash map to store each pattern and the node in which is has been firstly encountered. 

\begin{algorithm} 

\begin{algorithmic}[1] 
\State $H \leftarrow $ new Hashmap()
\Function{Expand}{$s_{sel}$}
		\State randomly choose $s_{exp}$ from non expanded children of $s_{sel}$

		\If{$(node \leftarrow H.get(s_{exp})) \not = null$}
			\State $node.parents.add(s_{sel})$
			\State $s_{exp} \leftarrow node$
		\Else{
				\State $s_{exp}.parents \leftarrow $ new List()
				\State $s_{exp}.parents.add(s_{sel})$
  				\State$H.put(s_{exp},s_{exp})$     \Comment{A pointer on the unique occurrence of $s_{exp}$}
		}
    	\EndIf
		\State add new child $s_{exp}$ to $s_{sel}$ in the tree \Comment{Expand $s_{sel}$ with $s_{exp}$}
    	\State return $s_{exp}$
\EndFunction{}
\end{algorithmic}
\center{The permutation unification principle.\label{algo:permuAMAF}}
\end{algorithm}

\subsection{The \textsc{RollOut} method}
From the expanded node $s_{exp}$ a simulation is run (\textsc{RollOut}).
With standard MCTS, a simulation is a random sequence
of actions that leads to a terminal node: A game state from which a reward can be computed (win/loss).
In our settings, it is not only the leaves that can be evaluated, but any pattern $s$ encountered during the simulation. 
Thus, we propose to define the notion of path (the simulation) and reward computation (which nodes are evaluated and how these different rewards are aggregated) separately.

\begin{definition}[Path policy]
Let $s_1$ the node from which a simulation has to be run (i.e., $s_1 = s_{exp}$). 
Let $n \geq 1 \in \mathbb{N}$, we define a path $p(s_1, s_n) = \{s_1, \dots, s_n\}$ as a chain in the lattice $(\mathcal{S},\prec)$, i.e., an ordered list of patterns starting from $s_1$ and ending with $s_n$ such that $\forall i \in \{1,\dots,n-1\}, s_{i+1}$ is a (not necessarily direct) refined pattern of $s_i$. 
\begin{itemize}
\item \textit{naive-roll-out}: 
a path of direct refinements is randomly created with length $pathLength \in \mathbb{N}^+$ a user-defined parameter.
\item \textit{direct-freq-roll-out}: The path is extended with a randomly chosen restriction until it meets an infrequent pattern $s_{n+1}$ using the direct refinement operator. Pattern $s_n$ is a leaf of the tree in our settings.
\item \textit{large-freq-roll-out} overrides the \textit{direct-freq-roll-out} policy by using specializations that are not necessarily direct. 
Several actions are added instead of one to create a new element of the path.
The number of added actions is randomly picked in $(1,...,jumpLength)$
where $jumpLength$ is given by the user ($jumpLength = 1$ gives the previous policy).
This techniques allows to visit deep parts of the search space with shorter paths.
\end{itemize}
\end{definition}

\begin{definition}[Reward aggregation policy]
Let $s_1$ be the node from which a simulation has been run
and $p(s_1, s_n)$ the associated random path.
Let $\mathcal{E} \subseteq p(s_1,s_n)$ be the subset of nodes to be evaluated.
The aggregated reward of the simulation is given by:
$\Delta = aggr( \{ \varphi(s) \forall s \in \mathcal{E} \}) \in [0;1]$ where $aggr$ is an aggregation function. We define several reward aggregation policies:
\begin{itemize}
\item \textit{terminal-reward}: $\mathcal{E} = \{ s_n \}$ and $aggr$ is the identity function.
\item \textit{random-reward}: $\mathcal{E} = \{ s_i \}$ with a random $1\leq i \leq n$ and $aggr$ is the identity function.
\item \textit{max-reward}: $\mathcal{E} = p(s_1, s_n)$ and $aggr$ is the $max(.)$ function 
\item \textit{mean-reward}: $\mathcal{E} = p(s_1, s_n)$ and $aggr$ is the $mean(.)$ function.
\item \textit{top-k-mean-reward}: $\mathcal{E} = $ \textit{top-k}$( p(s_1, s_n) )$, $aggr$ is the $mean(.)$ function and \textit{top-k(.)} returns the $k$ elements with the highest $\varphi$.
\end{itemize}
\end{definition}

A basic MCTS forgets any state encountered during a simulation. 
This is not optimal for single player games as relate \cite{DBLP:journals/tciaig/BjornssonF09}:
A pattern with a high $\varphi$ should not be forgotten as we might not expand the tree enough
to reach it. We propose to consider several memory strategies.
\begin{definition}[Roll-out memory policy]
A roll-out memory policy specifies which of the nodes of the path $p = (s_1, s_n)$ shall be kept in an auxiliary data structure $M$.
\begin{itemize}
\item \textit{no-memory}: Any pattern in $\mathcal{E}$ is forgotten.
\item \textit{all-memory}: All evaluated patterns in $\mathcal{E}$ are kept.
\item \textit{top-k-memory}: A list $M$ stores the best $k$ patterns in $\mathcal{E}$ w.r.t. $\varphi(.)$.
\end{itemize}
\end{definition}
This structure $M$ will be used to produce the final pattern set.

 \subsection{The \textsc{Update} method}
The backpropagation method updates the tree according to a simulation.
Let $s_{sel}$ be the selected node and $s_{exp}$ its expansion from which the simulation is run:
This step aims at updating the estimation $Q(.)$ and the number of visits $N(.)$ of each parent of $s_{exp}$ recursively.
Note that $s_{exp}$ may have several parents when we enable permutation unification (PU).
The number of visits is always incremented by one.
We consider three ways of updating $Q(.)$:
\begin{itemize}
\item \textit{mean-update}: $Q(.)$ is the average of the rewards $\Delta$ back-propagated through the node so far (basic MCTS).
\item \textit{max-update}: $Q(.)$ is the maximum reward $\Delta$ back-propagated through the node so far. This strategy enables to identify a local optimum within a part of the search space that contains mostly of uninteresting patterns. Thus, it gives more chance for this area to be exploited in the next iterations.
\item \textit{top-k-mean-update}: $Q(.)$ average of the $k$ best rewards $\Delta$ back-propagated through the node so far. It gives a stronger impact for the parts of the search space containing several local optima.
\end{itemize}

\textit{mean-update} is a standard in MCTS techniques. 
We introduce the \textit{max-update} and \textit{top-k-mean-update} policies 
as it may often happen that high-quality patterns are rare and scattered in the search space. 
The mean value of rewards from simulations would converge towards 0 (there are too many low quality subgroups), whereas the maximum value (and top-k average) of rewards enables to identify the promising parts of the search space.

\subsection{Search end and result output}
There are two ways a MCTS ends: Either the  computational budget is reached (number of iterations) or the tree is fully expanded (an exhaustive search has been possible, basically when the size of the search space is smaller than the number of iterations). Indeed, the number of tree nodes equals the number of iterations that have been performed.
It remains now to explore this tree and the data structure $M$ built by the memory policy to output the list of diverse and non-redundant patterns.

Let $\mathcal{P} = T \cup M$ be a pool of patterns, where $T$ is the set of patterns stored in the nodes of the tree.
The set $\mathcal{P}$ is totally sorted w.r.t. $\varphi$ in a list $\Lambda$. 
Thus, we have to pick the $k$-best diverse and non-redundant subgroups within this large pool of nodes $\Lambda$ to return the result set of subgroups $\mathcal{R} \subseteq \mathcal{P}$. 
For that, we choose to implement $filter(.)$ in a greedy manner as done by\cite{DBLP:journals/datamine/LeeuwenK12,DBLP:conf/dis/BoscGBRPBK16}.
 $\mathcal{R} = filter(\mathcal{P})$ as follows: A post-processing that filters out redundant subgroups from the diverse pool of patterns $\Lambda$ based on the similarity measure $sim$ and the maximum similarity threshold $\Theta$.  Recursively, we poll (and remove) the best subgroup $s^*$ from $\Lambda$, and we add $s^*$ to $\mathcal{R}$ if it is not redundant with any subgroup in $\mathcal{R}$. It can be shown easily that $redundancy(\mathcal{R}) = 0$.

Applying $filter(.)$ at the end of the search requires however that the pool of patterns $\mathcal{P}$ has a reasonable cardinality which may be problematic with MCTS in term of memory. The allowed budget always enables such post-processing in our experiments (up to one million iterations). 

\section{Related work} \label{sec:RW}
SD aims at extracting subgroups of individuals for which the distribution on the target variable is statistically different from the whole (or the rest of the) population~[\cite{DBLP:books/mit/fayyadPSU96/Klosgen96,DBLP:conf/pkdd/Wrobel97}]. 
Two similar notions have been formalized independently and then unified by \cite{DBLP:journals/jmlr/NovakLW09}: Contrast set mining and emerging patterns. 
Close to SD, redescription mining aims to discover redescriptions of the same group of objects in different views~\cite{DBLP:journals/tkde/LeeuwenG15}. 
Exceptional model mining (EMM) was first introduced by \cite{DBLP:conf/pkdd/LemanFK08} (see a comprehensive survey  by \cite{DBLP:journals/datamine/DuivesteijnFK16}). 
EMM generalizes SD dealing with more complex target concepts: There are not necessarily one but several target variables to discriminate. 
EMM seeks to elicit patterns whose extents induce a model that substantially deviates from the one induced by the whole dataset. 

First exploration methods that have been proposed for SD/EMM are exhaustive search ensuring that the best subgroups are found, e.g. \cite{DBLP:books/mit/fayyadPSU96/Klosgen96,DBLP:conf/pkdd/Wrobel97,DBLP:conf/ida/KavsekLJ03,DBLP:conf/ismis/AtzmullerL09}. Several pruning strategies have been used to avoid the exploration of uninteresting parts of the search space. 
These pruning strategies are usually based on the monotonic (or anti-monotonic) property of the support or upper bounds on the quality measure [\cite{DBLP:conf/pkdd/GrosskreutzRW08,KaytoueEtAl/Mach16}].
To the best of our knowledge, the most efficient algorithms are (i) \emph{SD-MAP*} from \cite{DBLP:conf/ismis/AtzmullerL09} which is based on the FP-growth paradigm [\cite{DBLP:conf/sigmod/HanPY00}]
and (ii) an exhaustive exploration with optimistic estimates on different quality measures [\cite{DBLP:journals/datamine/LemmerichAP16}].
When an exhaustive search is not possible, heuristic search can be used. The most widely used techniques in SD and EMM are \textit{beam search}, \textit{evolutionary} algorithms and \textit{sampling} methods.
Beam search performs a level-wise exploration of the search space: 
A beam of a given size (or dynamic size for recent work) is built from the root of the search space.
This beam only keeps the most promising subgroups to extend at each level [\cite{DBLP:journals/ml/LavracCGF04,DBLP:conf/ida/MuellerRSKRK09,DBLP:journals/datamine/LeeuwenK12}].
The redundancy issue due to the beam search is tackled with the pattern skyline paradigm by \cite{DBLP:conf/pkdd/LeeuwenU13}, and with a ROC-based beam search variant for SD by \cite{DBLP:conf/sdm/MeengDK14}. 
Another family of SD algorithms relies on evolutionary approaches. They use a fitness function to select which individuals to keep at the next generated population.
\emph{SDIGA}, from \cite{DBLP:journals/tfs/JesusGHM07}, is based on a fuzzy rule induction system where a rule is a pattern in disjunctive normal form (DNF).
Other approaches have been then proposed, generally ad-hoc solutions suited for specific pattern languages and selected quality measures [\cite{DBLP:journals/isci/RodriguezRRA12,DBLP:conf/hais/PachonVDL11,DBLP:journals/tfs/CarmonaGJH10}].

Finally, pattern sampling techniques are gaining interest. 
 \cite{DBLP:conf/ida/MoensB14} employ controlled direct pattern sampling (CDPS). 
It enables to create random patterns with the help of a procedure based on a controlled distribution as did \cite{DBLP:conf/kdd/BoleyLPG11}. 
This idea was extended by \cite{BendimeradEtAl/ICDM2016} for a particular EMM problem to discover exceptional models induced by attributed graphs. 
Pattern sampling is attractive as it supports direct interactions with the user for using his/her preferences to drive the search. Besides, with sampling methods, a result is available anytime. 
However, traditional sampling methods used in pattern mining need a given probability distribution over the pattern space: This distribution depends on both the data and the measures [\cite{DBLP:conf/kdd/BoleyLPG11,DBLP:conf/ida/MoensB14}]. 
Each iteration is independent and consists of drawing a pattern given this probability distribution. 
Moreover, these probability distributions exhibit the problem of the long tail: There are many more uninteresting patterns than interesting ones. 
Thus, the probability to draw an uninteresting pattern is still high, and not all local optima may be drawn: There are no guaranties on the diversity of the result set. 
Recently, the sampling algorithm \textsc{Misere} has been  proposed by \cite{DBLP:conf/pkdd/GayB12,DBLP:conf/icdm/EghoGBVC15,DBLP:journals/kais/EghoGBVC17}.
Contrary to the sampling method of Moens and Boley, \textsc{Misere} does not require any probability distribution. It is agnostic of the quality measure but it still employs a discretization of numerical attribute in a pre-processing task. To draw a pattern, \textsc{Misere} randomly picks an object in the data, and thus it randomly generalizes it into a pattern that is evaluated with the quality measure. Each draw is independent and thus the same pattern can be drawn several times.
Finally, MCTS samples the search space without any assumption about the data and the measure. Contrary to sampling methods, it stores the result of the simulations of the previous iterations and it uses this knowledge for the next iterations: The probability distribution is learned incrementally. 
If given enough computation budget, the exploration/exploitation trade-off guides the exploration to all local optima (an exhaustive search).
To the best of our knowledge, MCTS has never been used in pattern mining, however, \cite{DBLP:conf/icml/GaudelS10} designed the algorithm FUSE (Feature UCT Selection) which extends MCTS to a feature selection problem. 
This work aims at selecting the features from a feature space that are the more relevant w.r.t. the classification problem. 
For that, Gaudel and Sebag explore the powerset of the features (i.e., itemsets where the items are the features) with a MCTS method to find the sets of features that minimize the generalization error. 
Each node of the tree is a subset of feature, and each action consists of adding a new feature in the subset of features.
The authors focus on reducing the high branching factor by using $\textit{UCB1-Tuned}$ and \emph{RAVE} introduced by \cite{DBLP:conf/icml/GellyS07}. 
The latter enables to select a node even if it remains children to expand.
The aim of FUSE is thus to return the best subset of features (the most visited path of the tree), or to rank the features with the RAVE score.

\section{Empirical evaluation on how to parameterize \algo{}\label{sec:quanti-exp}}

Our MCTS implementation for pattern mining, called \algo{} is publicly available\footnote{\url{https://github.com/guillaume-bosc/MCTS4DM}}.  
As there are many ways to configure \algo{}, we propose first to study the influence of the parameters on runtime, pattern quality and diversity. We both consider benchmark and artificial data. 
The experiments were carried out on an Intel Core i7 CPU 4 Ghz machine with 16 GB RAM running under Windows 10. 

\subsection{Data}

Firstly, we gathered benchmark datasets used in the recent literature 
of SD and EMM, that is, from \cite{DBLP:journals/datamine/LeeuwenK12,DBLP:conf/icdm/DownarD15,DBLP:journals/tkde/LeeuwenG15,DBLP:conf/pkdd/LeeuwenK11,DBLP:conf/icdm/DuivesteijnK11}.
Table \ref{tab:benchmark} lists them, mainly taken from the UCI repository, and we provide some of their properties.

Secondly, we used a real world dataset from neuroscience. It concerns olfaction (see Table \ref{tab:benchmark}).This data provides a very large search space of numerical attributes (more details on the application are presented by \cite{DBLP:conf/dis/BoscGBRPBK16}).

Finally, to be able to specifically evaluate diversity, a ground-truth is required. Therefore, we create an artificial data generator  to produce datasets where patterns with a controlled WRAcc are hidden.
 The generator takes the parameters given in Table \ref{tab:gen-params} and it works as follow. 
A data table with nominal attributes is generated with a binary target.  The number of objects, attributes and attributes values are controlled with the parameters $nb\_obj$, $nb\_attr$  and $domain\_size$. 
Our goal is to hide $nb\_patterns$ patterns in noise: We generate random descriptions of random lengths $Ground = \{d_i ~|~i \in [1, nb\_patterns]\}$.
For each pattern, we generate $pattern\_sup$  objects positively labeled with a probability of $1 - noise\_rate$ to be covered by the description $d_i$, and $noise\_rate$  for not being covered. 
We also add $pattern\_sup \times out\_factor$  negative examples for the pattern $d_i$: It will allow patterns with different WRAcc. Finally, we add random objects until we reach a maximum number of transactions $nb\_obj$. 

\begin{table}\centering 
\caption{Benchmark datasets experimented on in the SD and EMM literature. \label{tab:benchmark}}
\begin{tabular}{ccccc}
\hline
Name            	& \# Objects & \# Attributes & Type of attributes & Target attribute \\
\hline
\textit{Bibtex}	& 7,395	& 1,836 & Binary	& TAG\_statphys23 \\
\textit{BreastCancer}	& 699			& 9		& Numeric	& Benign	 \\		
\textit{Cal500}				& 502			& 68	& Numeric	& Angry-Agressive	 \\		
\textit{Emotions}			& 594			& 72	& Numeric	& Amazed-suprised	 \\	
\textit{Ionosphere}		& 352			& 35	& Numeric	& Good	 \\	
\textit{Iris}      				& 150			& 4		& Numeric	& Iris-setosa	 \\	
\textit{Mushroom} 		& 8,124		& 22	& Nominal	& Poisonous	 \\	
\textit{Nursery}     			& 12,961	& 8		& Nominal	& class=priority	 \\	
\textit{Olfaction}     			& 1,689	& 82		& Numeric	& Musk	 \\	
\textit{TicTacToe   }     	& 958			& 9		& Nominal	& Positive	 \\	
\textit{Yeast} 					& 2,417		& 103	& Numeric	& Class1	 \\	
\hline
\end{tabular}
\end{table}

\begin{table} \centering 
\caption{Parameters of the artificial data generator\label{tab:gen-params}.}
\begin{tabular}{clcccc}
\hline
Name & Description  & $\mathcal{P}_{small}$ & $\mathcal{P}_{medium}$   & $\mathcal{P}_{large}$  \\ 
\hline\hline
$nb\_obj$            &  Number of objects    &  2,000 & 20,000 & 50,000 \\
$nb\_attr$            &  Number of attributes &  5     & 5  & 25 \\
$domain\_size$   &  Domain size  per attribute & 10   & 20 & 50\\
$nb\_patterns$     & Number of hidden patterns   & 3   & 5 & 25\\
$pattern\_sup$     & Support of each hidden pattern   & 100 & 100 & 100\\
$out\_factor$        &  Proba. of a pattern labeled $-$ & 0.1  & 0.1& 0.1\\
$noise\_rate$       &  Proba. of a object to be noisy  & 0.1 & 0.1  & 0.1 \\
\hline
\end{tabular}
\end{table}

\subsection{Experimental framework}
We perform a large pool of experiments to assess this new exploration method for pattern mining. For that, we have designed an experimental framework that enables to test the different combinations of factors for all the strategies we introduced in previous sections. Each experiment are run on the benchmark datasets. An experiment consists in varying a unique strategy parameter while the others are fixed. Since \algo{} uses random choices, each experiment is run five times and only the mean of the results is discussed.  

\paragraph{Default parameters.}
For each benchmark dataset, we provide a set of default parameters (Table~\ref{tab:defaultParamsBench}). Indeed, due to the specific characteristics of each dataset, a common set of default parameters is unsuitable.  Nevertheless, all datasets share a subset of common parameters:
\begin{itemize}
	\item The maximum size of the result set is set to $maxOutput = 50$.
	\item The maximum redundancy threshold is set to $\Theta = 0.5$. 
	\item The maximum description length is set to $maxLength = 5$. This is a widely used constraint in SD that enables to restrict the length of the description, i.e., it limits the number of \emph{effective} restrictions in the description.
	\item The quality measure used is $\varphi = WRAcc$ for the first label only.
	\item The SP-MCTS is used as the default UCB.
	\item The permutation unification (PU) strategy is used by default.
	\item The refinement operator for \textsc{Expand} is set to \textit{tuned-min-gen-expand}.
	\item The \textit{direct-freq-roll-out} strategy is used for the \textsc{Roll-Out}
	\item The reward aggregation policy is set to \textit{max-reward}.
	\item The memory policy is set to \textit{top-1-memory}.
	\item The update policy is set to \textit{max-udpate}.
\end{itemize}

\begin{table}[t!] \centering
\caption{The default parameters for each dataset\label{tab:defaultParamsBench}.}
\begin{tabular}{cccl}
\hline
Dataset & minSupp	& \# iterations	& Path Policy	\\ 
\hline
\textit{Bibtex}	& 50 & 50k & \textit{direct-freq-roll-out}\\
\textit{BreastCancer}	& 10 & 50k & \textit{large-freq-roll-out} ($jumpLength = 30$)\\
\textit{Cal500}	& 10 & 100k & \textit{large-freq-roll-out} ($jumpLength = 30$)\\
\textit{Emotions}	& 10 & 100k & \textit{large-freq-roll-out} ($jumpLength = 30$)\\
\textit{Ionosphere}	& 10 & 50k & \textit{large-freq-roll-out} ($jumpLength = 30$)\\
\textit{Iris}	& 10 & 50k & \textit{large-freq-roll-out} ($jumpLength = 30$) \\
\textit{Mushroom}	& 30 & 50k & \textit{direct-freq-roll-out}\\
\textit{Nursery}	& 50 & 100k & \textit{direct-freq-roll-out}  \\
\textit{Olfaction}	& 10 & 100k & \textit{large-freq-roll-out} ($jumpLength = 30$)\\
\textit{TicTacToe}	& 10 & 100k & \textit{direct-freq-roll-out}  \\
\textit{Yeast}	& 20 & 100k & \textit{large-freq-roll-out} ($jumpLength = 30$) \\
\hline 
\end{tabular}
\end{table}

\paragraph{List of experiments. }
Evaluating \algo{} is performed with six different batches of experiments:
\begin{itemize}
	\item Section~\ref{xp:UCB} is about the choice of the UCB.
	\item Section~\ref{xp:Expand} deals with the several strategies for the \textsc{Expand} method.
	\item Section~\ref{xp:RollOut} presents the leverage of all the possibilities for the \textsc{RollOut}.
	\item Section~\ref{xp:Memory} shows out the impact of the \textsc{Memory} strategy.
	\item Section~\ref{xp:Update} compares the behaviors of all the strategies for the \textsc{Update}.	
	\item Section~\ref{xp:NbIterations} performs the experiments by varying the computational budget.
	\item Section~\ref{xp:Completeness} studies if \algo{} is able to retrieved a diverse pattern set.
\end{itemize}
For simplicity and convenience, for each experiment we display the same batch of figures. For each dataset we show (i) the boxplots of the quality measure $\varphi$ of the subgroups in the result set, (ii) the histograms of the runtime and (iii) the boxplots of the description length of the subgroups in the result set depending on the strategies that are used. In this way, the impacts of the strategies are easy to understand. 

We do not evaluate memory consumption in this section, as it increases linearly with the number of iterations (to which should be added the number of patterns kept by the memory policy).

\subsection{The \textsc{Select} method}\label{xp:UCB}
The choice of the UCB is decisive, because it is the base of the exploration / exploitation trade off. Indeed, the UCB chooses which part of the search tree will be expanded and explored. We presented four existing UCBs and an adaptation with a \emph{normalized exploration rate} to take into account an enumeration based on a lectic order (LO).
As such, we need to consider also the \textit{expand} methods (standard, \textit{lectic order}  LO and \textit{permutation unification} PU) at the same time.

Figure~\ref{fig:UCB-runtime} presents the results. Comparing the runtime for all the strategies leads to conclude that there is no difference in the computation of the several UCBs (see Figure~\ref{fig:UCB-runtime}(a)). Indeed, the impact of the UCBs lies in its computation, and there is no UCB that is more time-consuming than others.
The difference we can notice, is that when LO is used, the runtime is lower. This result is expected because with LO, the search space is less large since each subgroup is unique in the search space (this is not due to the chosen UCB).
PU has also a smaller search space, but it requires call to updates pointers towards subgroups with the same extent, and requires thus more time. 

Figure~\ref{fig:UCB-runtime}(b) depicts the boxplots of the quality measure of the result set when varying the UCB. The results suggest that the \emph{UCB1-Tuned} and \emph{DFS-UCT} lead to weaker quality result for several datasets: On the \emph{Cal550}, \emph{Emotions} and \emph{Yeast} datasets, the quality measures of the result set are worse than the results of other UCBs (see, e.g., Figure~\ref{fig:UCB-runtime}(b)). This is due to the fact that the search space of these datasets is larger than the other with many local optima, and the \emph{UCB1-Tuned} is designed to explore less, thus less local optima are found. Besides, the \emph{SP-MCTS} seems to be more suitable for SD problems: The quality is slightly better than other UCBs for the \emph{BreastCancer} and \emph{Emotions} datasets. LO leads to a worse quality in the result set, whereas PU seems to be more efficient. 

The use of these different UCBs also do not impact the description length of the subgroups within the result set. For some datasets, the \emph{permutation unification} leads to longer descriptions (see for instance Figure~\ref{fig:UCB-runtime}(c)).

\begin{figure}\centering
\setlength{\tabcolsep}{0pt}
\renewcommand{\arraystretch}{1.5}\renewcommand{\arraystretch}{1.5}
\begin{tabular}{ccc}
\includegraphics[width=.33\columnwidth]{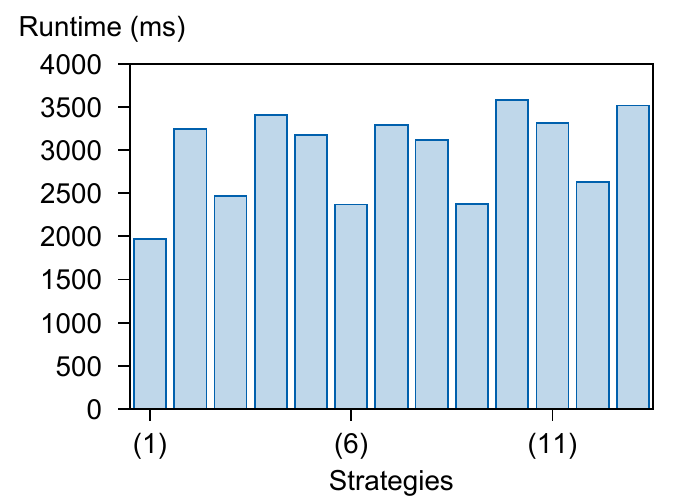} &
\includegraphics[width=.33\columnwidth]{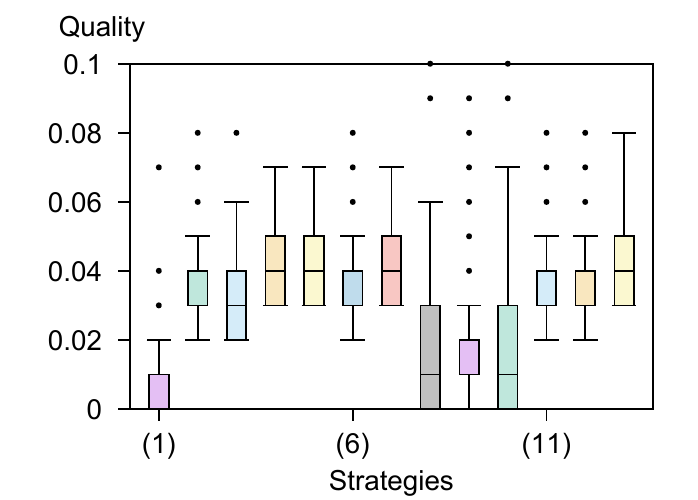}&
\includegraphics[width=.33\columnwidth]{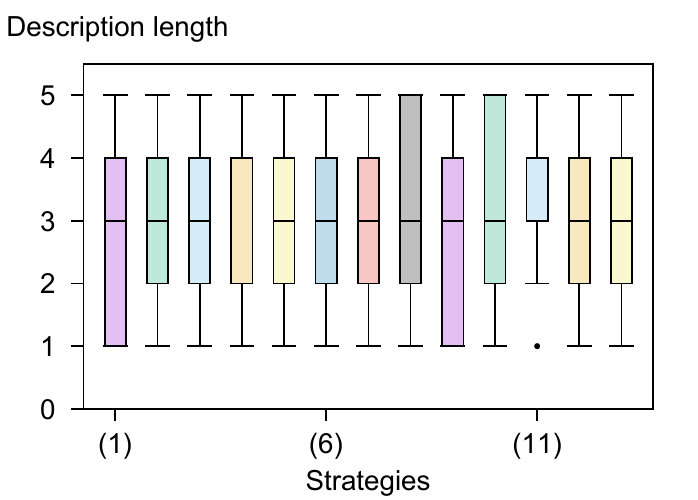}\\
(a) Runtime: \textit{BreastCancer} & (b) Avg. quality: \textit{Emotions} & (c) Descr length: \textit{Mushroom}\\
\hline
\multicolumn{3}{c}{
Select Policies: 	 
(1)	 DFS-UCT with LO 
 }\\
\multicolumn{3}{c}{
(2)	 UCB1	
(3)	 UCB1	with LO
(4)	 UCB1	with PU }\\
\multicolumn{3}{c}{
(5)	 SP-MCTS
(6)	 SP-MCTS	with LO
(7)	 SP-MCTS	with PU 
 }\\
\multicolumn{3}{c}{
(8)	 UCB1-Tuned
(9)	 UCB1-Tuned	with LO
(10)	 UCB1-Tuned	with  PU 
}\\
\multicolumn{3}{c}{
(11)	 UCT	
(12)	 UCT	with LO 
(13)	 UCT	with PU
}\\
\hline
\end{tabular} 
	\caption{Impact of the \textsc{Select} strategy\label{fig:UCB-runtime}.}
\end{figure}

\subsection{The \textsc{Expand} method}\label{xp:Expand}
Considering the \textsc{Expand} policy, we introduced three different refinement operators, namely \textit{direct-expand}, \textit{gen-expand} and \textit{label-expand}, and we presented two methods, namely LO and PU, to take into account that several nodes in the search tree are exactly the same. 
The several strategies we experiment with are given in Figure~\ref{fig:Expand-runtime}(bottom).
Let us consider the leverage on the runtime of these strategies in Figure~\ref{fig:Expand-runtime}(a). Once again, using LO implies a decrease of the runtime. Conversely, PU requires more time to run. There is very little difference in the runtime when varing the refinement operator: \emph{direct-expand} is the faster one, and \emph{label-expand} is more time consuming.

Considering the quality of the result set varying the expand strategies, we can assume that the impact differs w.r.t. the dataset (see Figure~\ref{fig:Expand-runtime}(b)). Surprisingly, LO improves the quality of the result set for some datasets (e.g. the Iris dataset in Figure~\ref{fig:Expand-runtime}(b)). This contradicts what we observe in the Emotions dataset of the previous experiment in Section~\ref{xp:UCB}. 
Most importantly, the results using \textit{label-expand} are better than other ones in most of the datasets.
Actually, this is due that this expand favors pattern with a better accuracy which is part of the WRAcc.

The description length of the extracted subgroups are quite constant when varying the expand strategies (see Figure~\ref{fig:Expand-runtime}(c)). With LO, the description lengths are slightly smaller than with other strategies.

\begin{figure}\centering
\setlength{\tabcolsep}{0pt}
\renewcommand{\arraystretch}{1.5}
\begin{tabular}{ccc}
\includegraphics[width=.33\columnwidth]{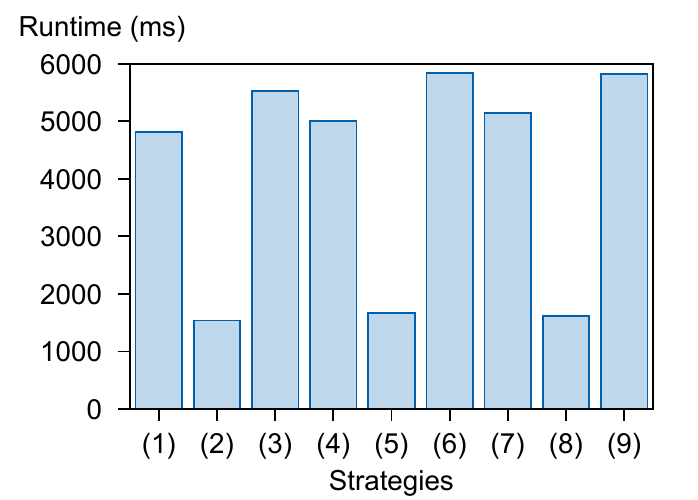}&
\includegraphics[width=.33\columnwidth]{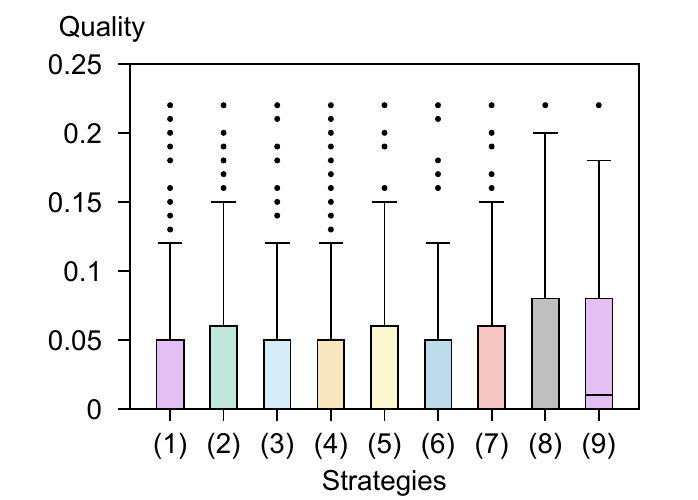} &
\includegraphics[width=.33\columnwidth]{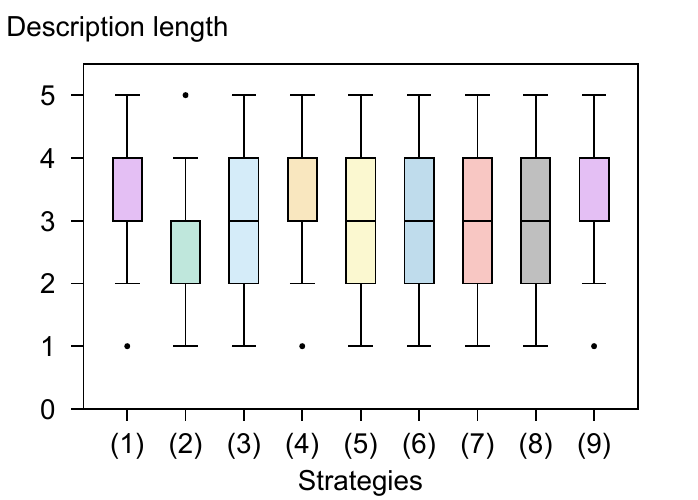} \\
(a) Runtime: \textit{Nursery}\\ & (b) Avg. quality: \textit{Iris} & (c) Descr length: \textit{Mushroom}\\
\hline
\multicolumn{3}{c}{
(1)    \textit{direct-expand}
(2)	 \textit{direct-expand}	with  LO
(3)	 \textit{direct-expand}	with  PU }\\
\multicolumn{3}{c}{
(4)	 \textit{gen-expand}	
(5)	 \textit{gen-expand} with LO
(6)	 \textit{gen-expand}	with PU }\\
\multicolumn{3}{c}{
(7)	 \textit{label-expand}	
(8)	 \textit{label-expand} with LO  
(9)	 \textit{label-expand} with PU 
}\\
\hline
\end{tabular} 
	\caption{Impact of the \textsc{Expand} strategy\label{fig:Expand-runtime}.}
\end{figure}

\subsection{The \textsc{RollOut} method}\label{xp:RollOut}
For the \textsc{RollOut} step we derived several strategies that combine both the path policy and the reward aggregation policy in Table~\ref{tab:strategies-rollout}. Clearly, the experiments show that the runs using the direct refinement operator (\textit{naive-roll-out} and \textit{direct-freq-roll-out}) are time consuming (see Figure~\ref{fig:RollOut-runtime}(a)).  In the \textit{BreastCancer} data, the runtime are twice longer with the direct refinement operator than with the \textit{large-freq-roll-out} path policy. In other datasets (e.g., \textit{Ionosphere} or \textit{Yeast}), the runtime is even more than 3 minutes (if the run lasts more than 3 minutes to perform the number of iterations, the run is ignored). Besides, it is clear that the \textit{random-reward} aggregation policy is less time consuming than other strategies. Indeed, with \textit{random-reward}, the measure of only one subgroup within the path is computed, thus it is faster.

Figure~\ref{fig:RollOut-runtime}(b) is about the quality of the result set. The \textit{naive-roll-out} and \textit{direct-freq-roll-out} path policies lead to the worst results. Besides, the quality of the result set decreases with the \textit{random-reward} reward aggregation policy in other datasets (e.g., \textit{Emotions}). Basically, these strategies evaluate only random nodes and thus they are not able to identify the promising parts of the search space. Finally, there are not large differences between other strategies.

As can be seen in Figure~\ref{fig:RollOut-runtime}(c), the description length of the subgroups is not very impacted by the strategies of the \textsc{Roll-Out} step. The results of the \textit{random-reward} reward aggregation policy are still different from other strategies: The description length is smaller for the \textit{Mushroom} dataset. Using \textit{large-freq-roll-out} with $jumpLength = 100$ leads to smaller descriptions for the \textit{Mushroom} dataset. Finally, the description length is not or almost not influenced by the \textsc{Roll-Out} strategies.

\begin{table}[t!] \centering \scriptsize
\caption{The list of strategies used to experiment with the \textsc{RollOut}  method\label{tab:strategies-rollout}.}
\begin{tabular}{cll} 
\hline
Strategy & Path Policy  & Reward Aggregation Policy \\ 
\hline
(1)	& \textit{naive-roll-out} ($pathLength = 20$)	& \textit{terminal-reward} \\
\hline
(2)	& \textit{direct-freq-roll-out}	& \textit{max-reward} \\
(3)	& \textit{direct-freq-roll-out}	& \textit{mean-reward} \\
(4)	& \textit{direct-freq-roll-out}	& \textit{top-2-mean-reward} \\
(5)	& \textit{direct-freq-roll-out}	& \textit{top-5-mean-reward} \\
(6)	& \textit{direct-freq-roll-out}	& \textit{top-10-mean-reward} \\
(7)	& \textit{direct-freq-roll-out}	& \textit{random-reward} \\
\hline
(8)	& \textit{large-freq-roll-out}	($jumpLength = 10)$& \textit{max-reward} \\
(9)	& \textit{large-freq-roll-out}	($jumpLength = 10)$& \textit{mean-reward} \\
(10)	& \textit{large-freq-roll-out}	($jumpLength = 10)$& \textit{top-2-mean-reward} \\
(11)	& \textit{large-freq-roll-out}	($jumpLength = 10)$& \textit{top-5-mean-reward} \\
(12)	& \textit{large-freq-roll-out}	($jumpLength = 10)$& \textit{top-10-mean-reward} \\
(13)	& \textit{large-freq-roll-out}	($jumpLength = 10)$& \textit{random-reward} \\
\hline
(14)	& \textit{large-freq-roll-out}	($jumpLength = 20)$& \textit{max-reward} \\
(15)	& \textit{large-freq-roll-out}	($jumpLength = 20)$& \textit{mean-reward} \\
(16)	& \textit{large-freq-roll-out}	($jumpLength = 20)$& \textit{top-2-mean-reward} \\
(17)	& \textit{large-freq-roll-out}	($jumpLength = 20)$& \textit{top-5-mean-reward} \\
(18)	& \textit{large-freq-roll-out}	($jumpLength = 20)$& \textit{top-10-mean-reward} \\
(19)	& \textit{large-freq-roll-out}	($jumpLength = 20)$& \textit{random-reward} \\
\hline
(20)	& \textit{large-freq-roll-out}	($jumpLength = 50)$& \textit{max-reward} \\
(21)	& \textit{large-freq-roll-out}	($jumpLength = 50)$& \textit{mean-reward} \\
(22)	& \textit{large-freq-roll-out}	($jumpLength = 50)$& \textit{top-2-mean-reward} \\
(23)	& \textit{large-freq-roll-out}	($jumpLength = 50)$& \textit{top-5-mean-reward} \\
(24)	& \textit{large-freq-roll-out}	($jumpLength = 50)$& \textit{top-10-mean-reward} \\
(25)	& \textit{large-freq-roll-out}	($jumpLength = 50)$& \textit{random-reward} \\
\hline
(26)	& \textit{large-freq-roll-out}	($jumpLength = 100)$& \textit{max-reward} \\
(27)	& \textit{large-freq-roll-out}	($jumpLength = 100)$& \textit{mean-reward} \\
(28)	& \textit{large-freq-roll-out}	($jumpLength = 100)$& \textit{top-2-mean-reward} \\
(29)	& \textit{large-freq-roll-out}	($jumpLength = 100)$& \textit{top-5-mean-reward} \\
(30)	& \textit{large-freq-roll-out}	($jumpLength = 100)$& \textit{top-10-mean-reward} \\
(31)	& \textit{large-freq-roll-out}	($jumpLength = 100)$& \textit{random-reward} \\
\hline 
\end{tabular}
\end{table}

\begin{figure}\centering
\setlength{\tabcolsep}{0pt}
\renewcommand{\arraystretch}{1.5}
\begin{tabular}{ccc}
\includegraphics[width=.33\columnwidth]{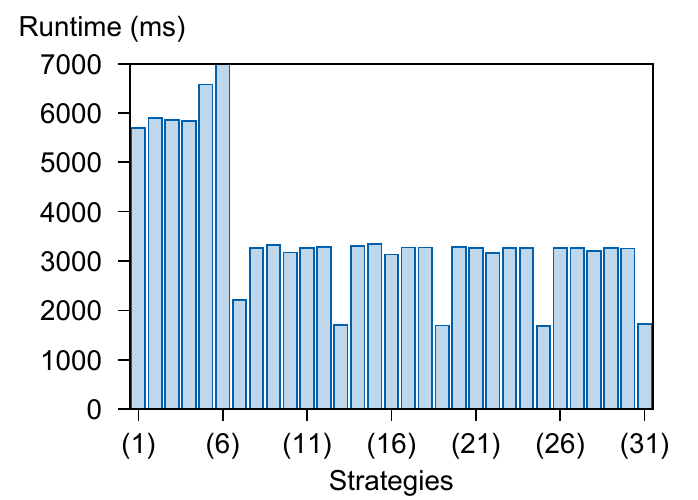} &
\includegraphics[width=.33\columnwidth]{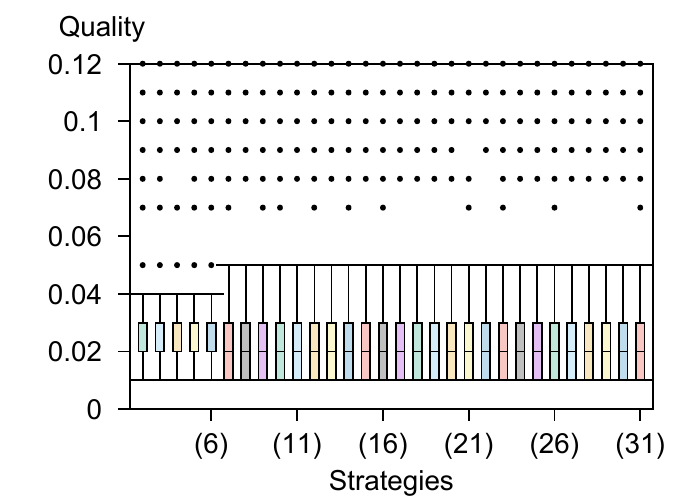}&
\includegraphics[width=.33\columnwidth]{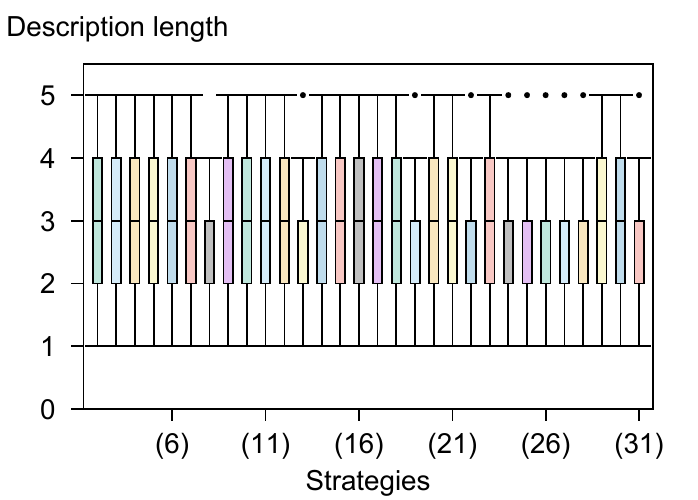}\\
(a) Runtime: \textit{BreastCancer} & (b) Avg. quality: \textit{Mushroom} & (c) Descr length: \textit{Mushroom}
\end{tabular} 
	\caption{Impact of the \textsc{Roll-Out} strategy\label{fig:RollOut-runtime}.}
\end{figure}

\subsection{The \textsc{Memory} method}\label{xp:Memory}
We derived six strategies for the \textsc{Memory} step  given in Figure~\ref{fig:Memory-runtime}(bottom). Obviously, the \textit{all-memory} policy is slower than other strategies because all the nodes within the path of the simulation have to be stored (see Figure~\ref{fig:Memory-runtime}(a)). Conversely, the \textit{no-memory} policy is the fastest strategy. The runtimes of the \textit{top-k-memory} policies is comparable.

Figure~\ref{fig:Memory-runtime}(b) shows that the quality of the result set is impacted by the choice of the memory policies. We can observe that the \textit{no-memory} is clearly worse than other strategies. Indeed, in the \textit{Emotion} dataset, the best subgroups are located more deeper in the search space, thus, if the solutions encountered during the simulation are not stored it would be difficult to find them just be considering the subgroups that are expanded in the search tree. Surprisingly, the \textit{all-memory} policy does not lead to better result. In fact the path generated during a simulation contains a lot of redundant subgroups: Storing all these nodes is not required to improve the quality of the result set. Only few subgroups within the path are related to different local optima.

As expected in Figure~\ref{fig:Memory-runtime}(c), the descriptions of the subgroups obtained with the \textit{no-memory} policy are smaller than those of other strategies. Indeed, with the \textit{no-memory} policy, the result sets contains only subgroups that are expanded in the search tree, in other words, the subgroups obtained with the \textsc{Expand} step.

\begin{figure}\centering
\setlength{\tabcolsep}{0pt}
\renewcommand{\arraystretch}{1.5}\scriptsize
\begin{tabular}{ccc}
\includegraphics[width=.33\columnwidth]{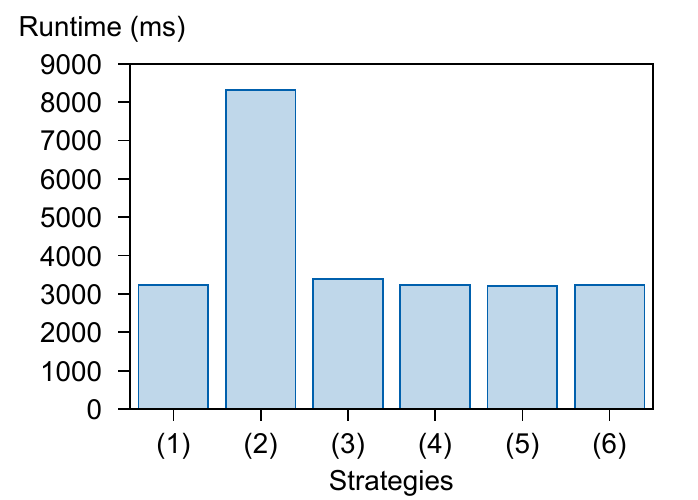} &
\includegraphics[width=.33\columnwidth]{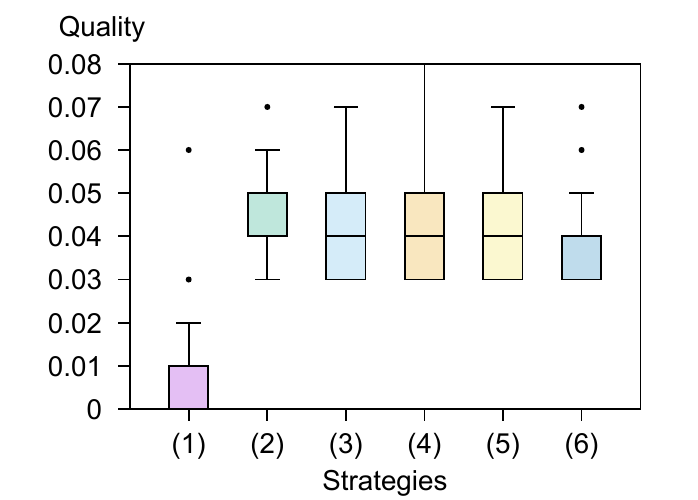}&
\includegraphics[width=.33\columnwidth]{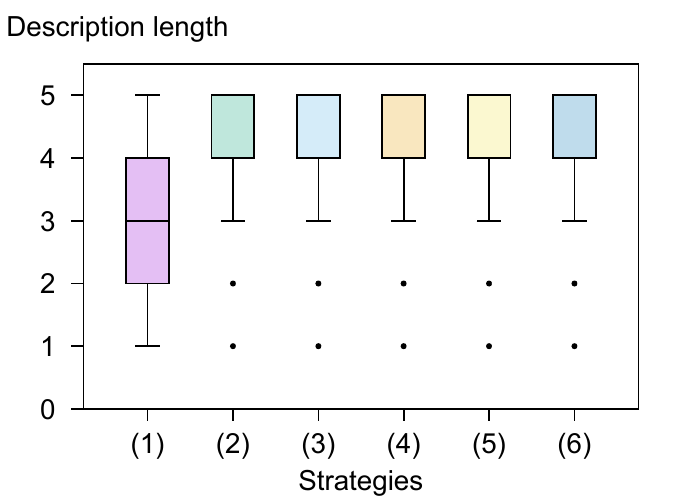}\\
(a) Runtime: \textit{Ionosphere} & (b) Avg. quality: \textit{Emotions} & (c) Descr. length: \textit{BreastCancer}\\
\hline
\multicolumn{3}{c}{
(1)	 \textit{no-memory}	
(2)	 \textit{all-memory}	
(3)	 \textit{top-1-memory}	 
}\\
\multicolumn{3}{c}{
(4)	 \textit{top-2-memory}	
(5)	 \textit{top-5-memory}	
(6)	 \textit{top-10-memory}	
}\\
\hline
\end{tabular} 
	\caption{Impact of the \textsc{Memory} strategy\label{fig:Memory-runtime}.}
\end{figure}

\subsection{The \textsc{Update} method}\label{xp:Update}
Figure~\ref{fig:Update-runtime}(bottom) presents the different strategies we use to implement the \textsc{Update} step. The goal of this step is to back-propagate the reward obtained by the simulation to the parent nodes. The runtime of these strategies are comparable (see Figure~\ref{fig:Update-runtime}(a)). However, we notice that the \textit{top-k-mean-update} policy is a little more time consuming. Indeed, we have to maintain a list for each node within the built tree that stores the top-k best rewards obtained so far.

Figure~\ref{fig:Update-runtime}(b) shows the quality of the result set when varying the \textsc{Update} policies. For most of the datasets, since the proportion of local optima is very low within the search space, the \textit{max-update} is more efficient than the \textit{mean-update}. Indeed, using the \textit{max-update} enables to keep in mind that there is an interesting pattern that is reachable from a node. However, Figure~\ref{fig:Update-runtime}(b) presents the opposite phenomena: The \textit{mean-update} policy leads to a better result. In fact, since there are a lot of local optima in the \textit{Ionosphere} dataset, the \textit{mean-update} can find the areas with lots of interesting solutions. Moreover, using the \textit{top-k-mean-update} leads to the \textit{mean-update} when $k$ increases.

The description of the subgroups in the result set is comparable when varying the policies of the \textsc{Update} method (see Figure~\ref{fig:Update-runtime}(c)). Indeed, the aim of the \textsc{Update} step is just to back-propagate the reward obtained during the simulation to the nodes of the built tree to guide the exploration for the following iterations. This step does not have a large influence on the length of the description of the subgroups.

\begin{figure}\centering
\setlength{\tabcolsep}{0pt}
\renewcommand{\arraystretch}{1.5}\scriptsize\scriptsize
\begin{tabular}{ccc}
\includegraphics[width=.33\columnwidth]{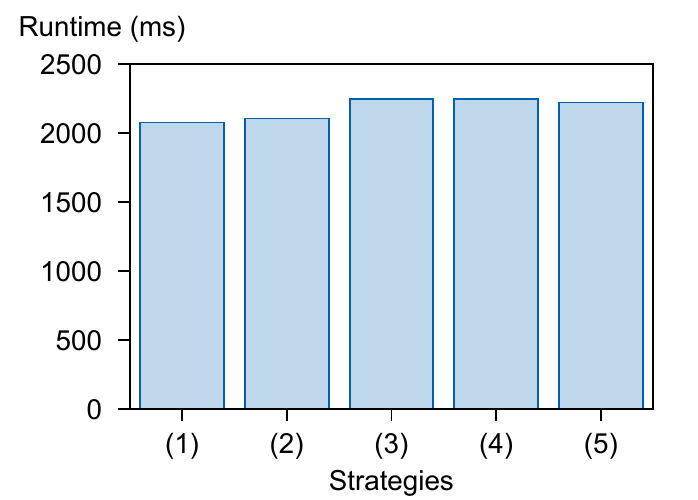}&
\includegraphics[width=.33\columnwidth]{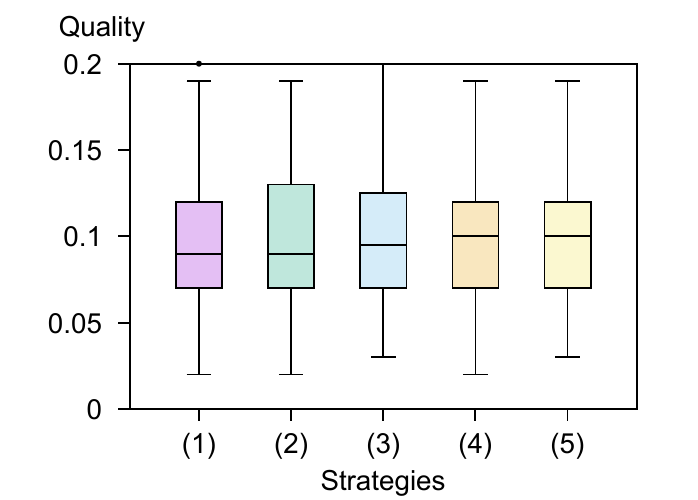}&
\includegraphics[width=.33\columnwidth]{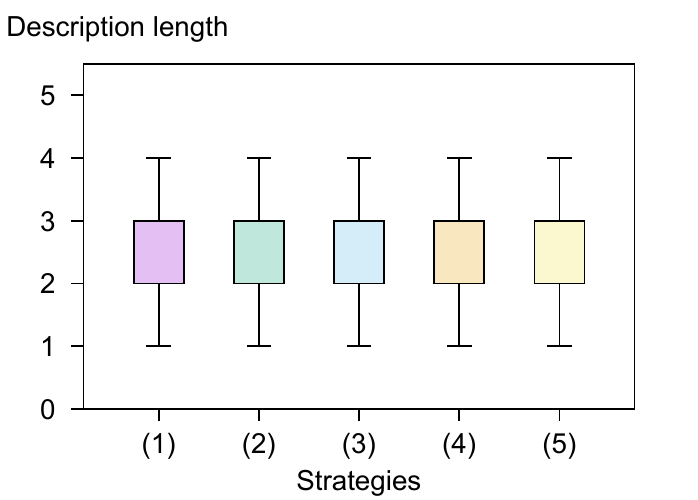}\\
(a) Runtime: \textit{TicTacToe} & (b) Avg. quality: \textit{Ionosphere} & (c) Descr length: \textit{Iris}\\
\hline
\multicolumn{3}{c}{
(1)	 \textit{max-update}	
(2)	 \textit{mean-update}	
(3)	 \textit{top-2-mean-update}	 
}\\
\multicolumn{3}{c}{
(4)	 \textit{top-5-mean-update}	
(5)	 \textit{top-10-mean-update}	
}\\
\hline
\end{tabular} 
	\caption{Impact of the \textsc{Update} strategy\label{fig:Update-runtime}.}
\end{figure}

\subsection{The number of iterations}\label{xp:NbIterations}
We study the impact of different computational budgets allocated to \algo{},
that is, the maximum number of iterations the algorithm can perform.
As depicted in Figure~\ref{fig:NbIterations-runtime}(a),
 the runtime is linear with the number of iterations. 
 The x-axis is not linear w.r.t. the number of iterations, 
 please refer to the bottom of Figure~\ref{fig:NbIterations-runtime} 
 to know the different values of the number of iterations. 

Moreover, as expected, the more iterations, the better the quality of the result set. Figure~\ref{fig:NbIterations-runtime}(b) shows that a larger computational budget leads to a better quality of the result set, but, obviously, it requires more time. Thus, with this exploration method, the user can have some results anytime. For the \emph{BreastCancer} dataset, the quality decreases from 10 to 100 iterations: This is due to the fact that with 10 iterations there are less subgroups extracted (12 subgroups) than with 100 iterations (40 subgroups), and the mean quality of the result set with 100 iterations contains also subgroups with lower quality measures.

\begin{figure}[b]\centering
\setlength{\tabcolsep}{0pt}
\renewcommand{\arraystretch}{1.5}\renewcommand{\arraystretch}{1.5}
\begin{tabular}{ccc}
\includegraphics[width=.33\columnwidth]{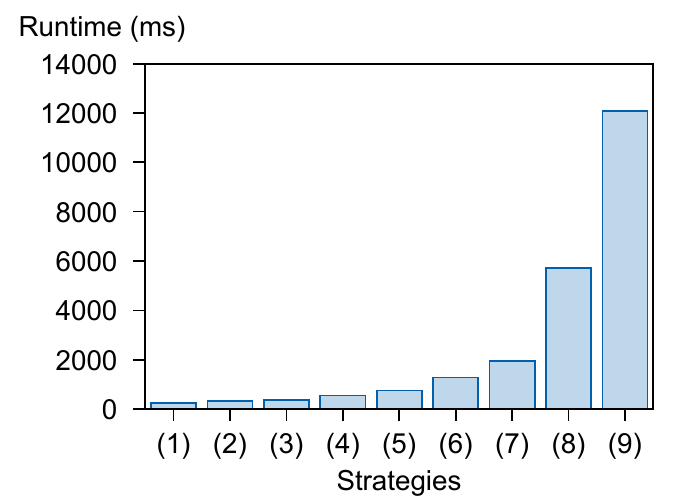} &
\includegraphics[width=.33\columnwidth]{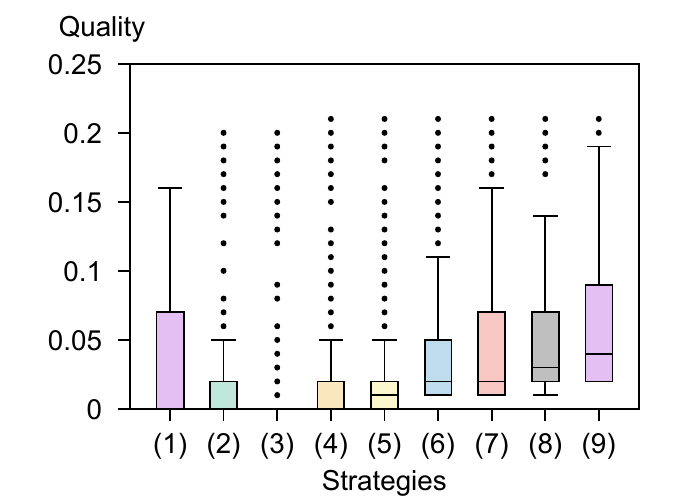}&
\includegraphics[width=.33\columnwidth]{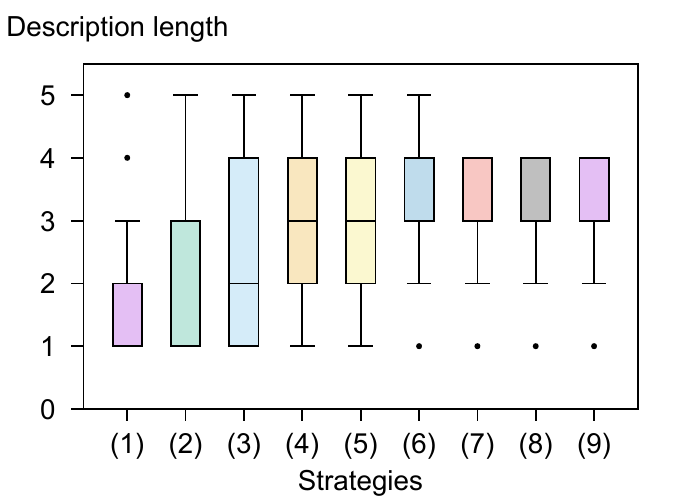}\\
(a) Runtime: Cal500 & (b) Avg. quality: BreastCancer & (c) Descr length: Nursery\\
\hline
\multicolumn{3}{c}{
(strategy)\#iterations: (1)10	 
(2)50	
(3)100	
(4)500	
(5)1K
(6)5K 
(7)10K	 
(8)50K	
(9)100K
}\\
\hline
\end{tabular} 
	\caption{Impact of the maximal budget (number of iterations)\label{fig:NbIterations-runtime}.}
\end{figure}

\subsection{Evaluating pattern set diversity when a ground truth is known\label{xp:Completeness}}

Artificial datasets are generated according to default parameters given in Table \ref{tab:gen-params}.
 Then, we study separately the impact of each parameter on the ability to retrieve the hidden patterns with our MCTS algorithm.
After a few trials, 
we use the following default MCTS parameters: The \textit{single player UCB} (SP-MCTS) for the select policy; 
the \textit{label-expand} policy with PU activated;  the \textit{direct-freq-roll-out} policy for the simulations,
the \textit{max-update} policy as aggregation function of the rewards of a simulation,
the \textit{top-10-memory}  policy and finally the \textit{max-update} policy for the back-propagation.    

The ability to retrieve hidden patterns is measured with a Jaccard coefficient between the support of the hidden patterns and the one discovered by the algorithm:
\begin{definition}[Evaluation measure]
Let $\mathcal{H}$ be the set of hidden patterns, and $\mathcal{F}$ the set of patterns found by an MCTS mining algorithm, the quality, of the found collection is given by:
\[ qual(\mathcal{H}, \mathcal{F}) = avg_{\forall h \in \mathcal{H}} ( max_{\forall f \in \mathcal{F}}(Jaccard(ext(h),ext(f)))  ) \]
that is, the average of the quality of each hidden pattern, which is the best Jaccard coefficient with a found pattern. We thus measure here the \textit{diversity}. It is a pessimistic measure in the sense that
it takes its maximum value $1$ iff all patterns have been completely retrieved.\label{def:diversity-measure}
\end{definition}
It can be noticed that we do not use the Definition \ref{def:pattern-set-diversity} for diversity: As a ground truth is available, we opt for a measure that quantifies its recovering.

\smallskip

\noindent\textit{Varying the noise parameter.} We start with the set of parameters $\mathcal{P}_{small}$.
The results are given in Figure \ref{fig:artificial1} with different minimal 
support values (used during the expand step and the simulations).
Recall that a hidden pattern is random set of symbols $attribute = value$ when dealing with nominal attributes,
repeated in $pattern\_sup$ object descriptions: The noise 
makes that each generated object description may not support the pattern.
Thus, the noise directly reduces the support of a hidden pattern:
increasing the noise requires to decrease the minimal support of the algorithm. 
This is clearly observable on the different figures. When the minimum
support is set to the same value as the support of the hidden patterns ($minSupp=100$),
the noise has a strong impact and it is difficult to retrieve the hidden patterns,
even when the whole tree (of frequent patterns) has been expanded. 
Reducing the minimal support to $1$  makes the search very resistant to noise.
Note that when two lines exactly overlaps, it means that the search space
of frequent patterns was fully explored: MCTS performed an exhaustive search.

\begin{figure}\centering
\setlength{\tabcolsep}{0pt}
\renewcommand{\arraystretch}{1.5}
\begin{tabular}{cc}
\includegraphics[width=0.5\textwidth]{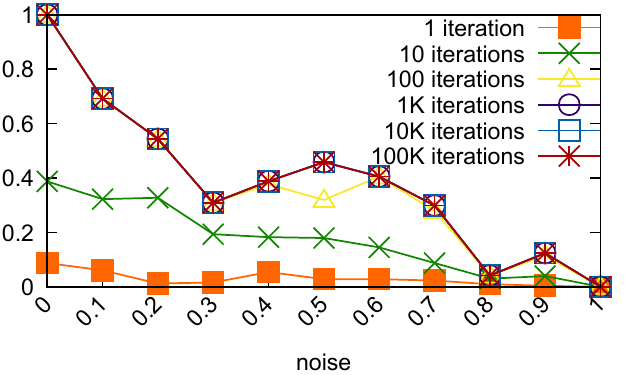} &
\includegraphics[width=0.5\textwidth]{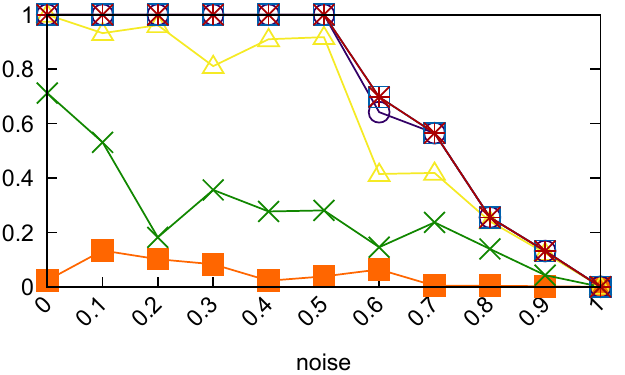} \\
(i) $minSup=100$ & (ii) $minSupp = 50$ \\
\multicolumn{2}{c}{\includegraphics[width=0.5\textwidth]{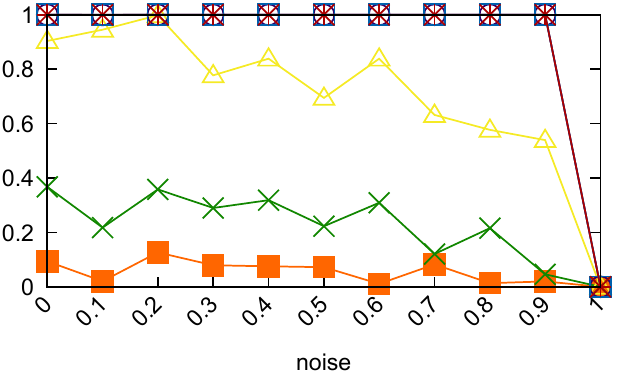}} \\
\multicolumn{2}{c}{ (iii) $minSup=10$}
\end{tabular} 
\caption{Ability to retrieve hidden patterns ($qual(\mathcal{H}, \mathcal{F})$ in Y-axis) when introducing noise and mining with different minimum supports $minSup$.\label{fig:artificial1}}
\end{figure}

\smallskip

\noindent\textit{Varying the out factor parameter.} Each pattern is inserted in $pattern\_sup$ transactions
 (or less when there is noise) as positive examples (class label $+$). 
We also add  $pattern\_sup \times out\_factor$ negative examples (class label $-$). 
When $out\_factor=1$, each pattern appears as much  in positive and negative examples.
This allows to hide patterns with different quality measure, 
and especially different $WRAcc$ measures.
The Table \ref{fig:artificial2} (row (1)) shows that this parameter has no impact:
patterns of small quality are retrieved easily in a small number of iterations. 
The UCB hence drives the search towards promising parts that have the best rewards.

\smallskip

\noindent\textit{Varying the number of hidden patterns.} We claim that the UCB will guide the search towards interesting parts (exploitation) but also unvisited parts (exploration) of the search space. It follows that 
all hidden patterns should be retrieved and well retrieved. We thus vary the number of random patterns  
between $1$ and $20$ and observe that they are all retrieved (Table \ref{fig:artificial2} (row (2))).
When increasing the number of hidden more patterns, retrieving all of them requires more iterations in the general case. 

\smallskip

\noindent\textit{Varying the support size of the hidden patterns.} Patterns with a high support (relative to the total number of objects)
should be easier to be retrieved as a simulation has more chance to discover them, even partially. We observe that patterns
with small support can still be retrieved but it requires more iterations to retrieve them in larger datasets
(Table \ref{fig:artificial2} (row (3))). 

\smallskip

\noindent\textit{Varying the number of objects.} The number of objects directly influences the computation of 
the support of each node: Each node stores a \textit{projected database} that lists which objects belong 
to the current pattern. The memory required for our MCTS implementation follows a linear complexity
w.r.t. the number of iterations. This complexity can be higher depending on the chosen memory policy
(e.g. in these experiments, the top-10 memory policy was chosen). The time needed to compute the support
of a pattern is higher for larger dataset, but it does not change the number of iterations required to find
a good result.  This is reported in (Table \ref{fig:artificial2} (row (4))). Run times will be discussed later.

\smallskip

\noindent\textit{Varying the number of attributes and the size of attributes domains.} These two parameters
directly determine the branching factor of the exploration tree. It takes thus more iterations to 
fully expand a node and to discover all local optima. Here again, all patterns are well discovered but
larger datasets require more iterations (Table \ref{fig:artificial2} (row (5) and (6))).

\begin{longtable}{ccc}  
\toprule
     & $\mathcal{P}_{medium}$ & $\mathcal{P}_{large}$ \\
\midrule
\endhead
\rotatebox{90}{~~~(1) $out\_factor$} &
\includegraphics[width=0.43\textwidth]{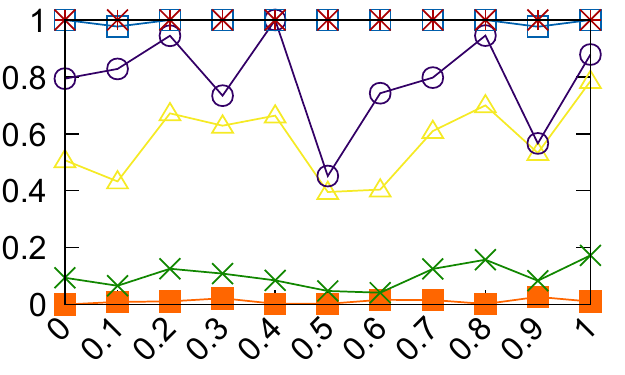} &
\includegraphics[width=0.43\textwidth]{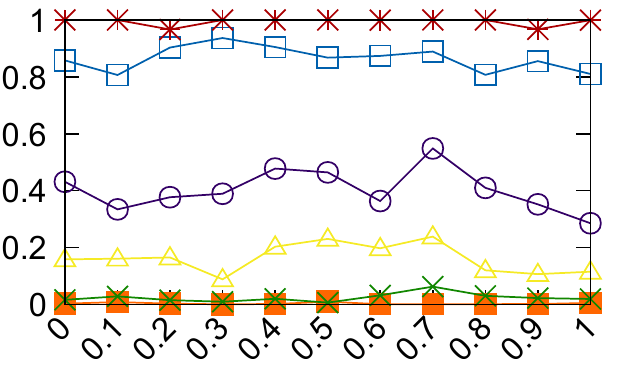} \\[6pt]
\rotatebox{90}{~~(2) $nb\_patterns$} &
\includegraphics[width=0.43\textwidth]{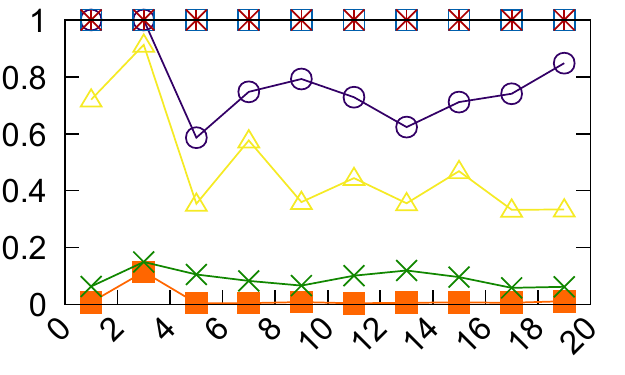} &
\includegraphics[width=0.43\textwidth]{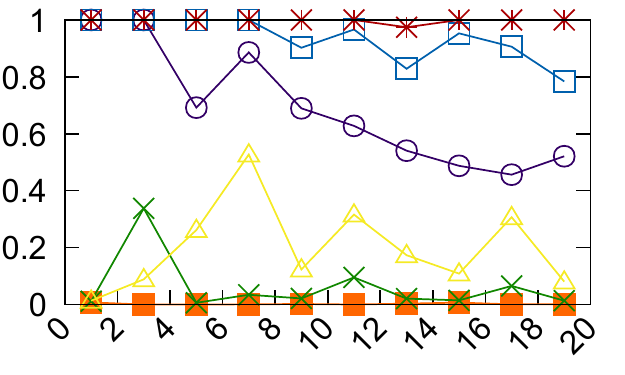} \\[6pt]
\rotatebox{90}{~~(3) $pattern\_sup$} &
\includegraphics[width=0.43\textwidth]{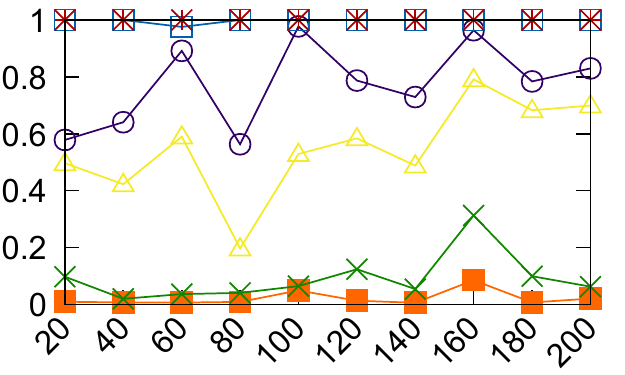} &
\includegraphics[width=0.43\textwidth]{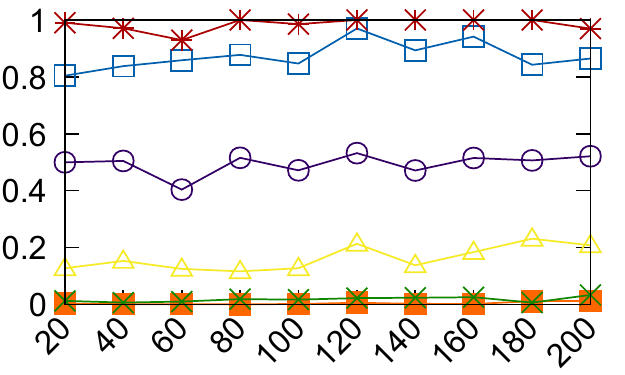} \\[6pt]
\rotatebox{90}{~~~~~(4) $nb\_obj$} &
\includegraphics[width=0.43\textwidth]{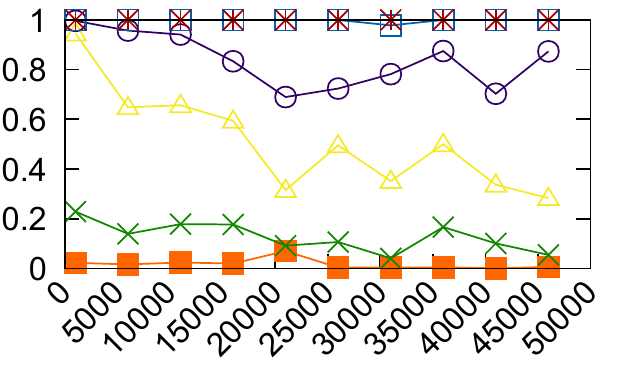} &
\includegraphics[width=0.43\textwidth]{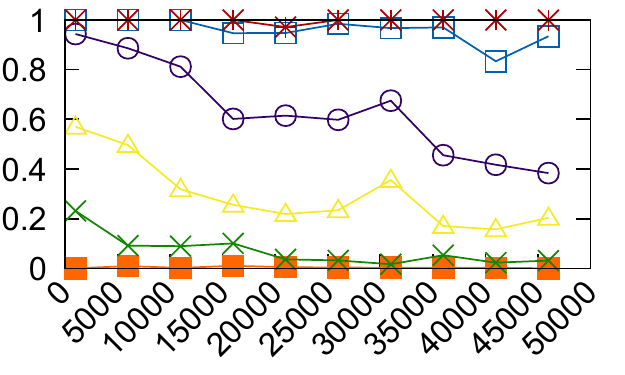}  \\[6pt]
\rotatebox{90}{~~~~~~(5) $nb\_attr$} &
\includegraphics[width=0.43\textwidth]{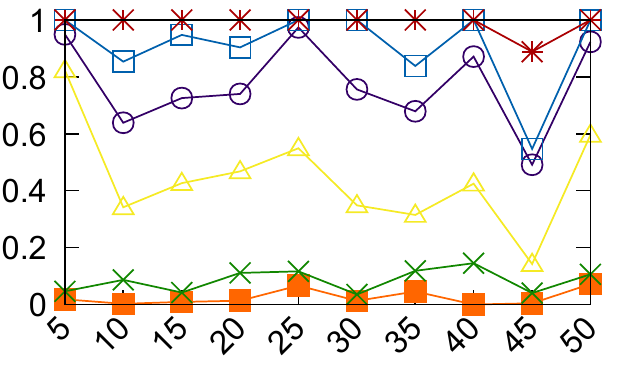} &
\includegraphics[width=0.43\textwidth]{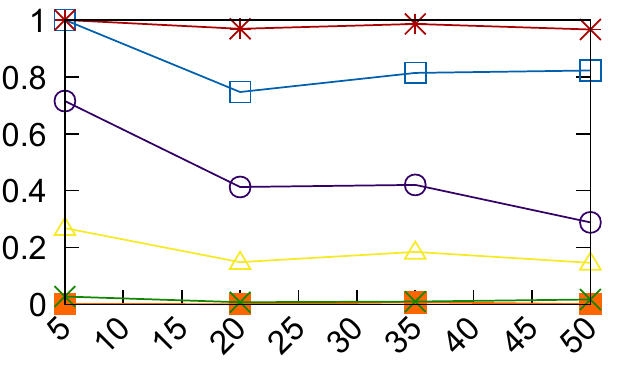}  \\[6pt]
\rotatebox{90}{~(6) $domain\_size$} &
\includegraphics[width=0.43\textwidth]{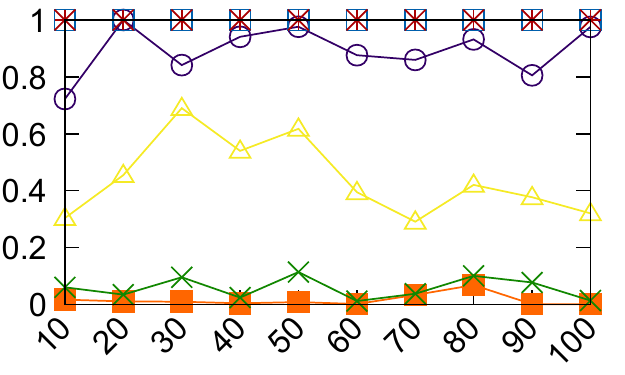} &
\includegraphics[width=0.43\textwidth]{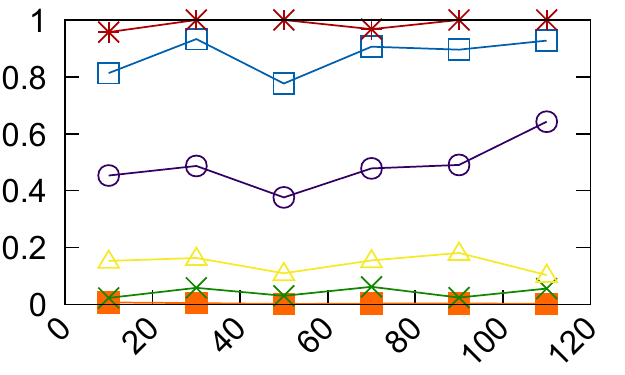}  \\[6pt]
\hline
 &  $\mathcal{P}_{medium}$ & $\mathcal{P}_{large}$  \\
\bottomrule
\multicolumn{3}{c}{Legend: \includegraphics[width=6cm]{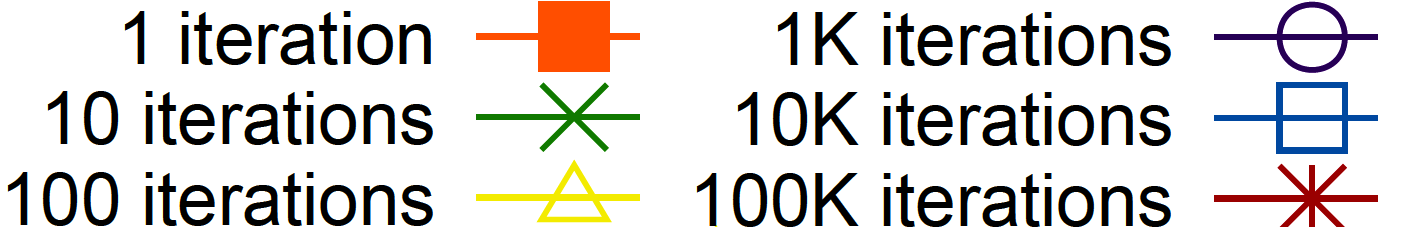}}\\
\caption{Evaluation of the ability to retrieve hidden patterns in artificial data generated according to different parameters (average of 5 runs for each point).\\ $qual(\mathcal{H}, \mathcal{F})$ in Y-axis, parameters given in the first columns as X-axis. \label{fig:artificial2}}
\end{longtable}

\section{Comparisons with existing algorithms}\label{sec:xp-comparisons}

We compare \algo{} to other SD approaches (exhaustive search, beam search, genetic algorithms and sampling) in terms of computational time, diversity and redundancy of the pattern set, and memory usage. In addition to the benchmark data we used in the previous section, we generate 5 new artificial datasets for which parameters are given in Table~\ref{tab:artificialData}. In this empirical study, we consider a timeout of 5 minutes that is enough to capture the behavior of the algorithms that are not based on a computational budget, such as \textsc{SD-Map} or beam search approaches. Indeed, \algo{} and sampling methods use a computational budget.

\begin{table}[t!] \centering 
\setlength{\tabcolsep}{4pt}
\renewcommand{\arraystretch}{1.3}
\caption{Parameters of the artificial data generator\label{tab:artificialData}. The format of name of the data is given by their parameters by $[nb\_obj]$\_$[nb\_attr]$\_$[domain\_size]$. }
\begin{tabular}{cccccc}
\hline
Name& 5000\_10\_200 & 5000\_50\_50   & 5000\_50\_200 & 20000\_10\_200 & 20000\_50\_200  \\ 
\hline\hline
$nb\_obj$            &  5,000 & 5,000 & 5,000 & 20,000 & 20,000\\
$nb\_attr$            &  10     & 50  & 50 & 10 & 50 \\
$domain\_size$   & 200   & 50 & 200 & 200 & 200\\
$nb\_patterns$     & 5   & 5 & 5 & 5 & 5\\
$pattern\_sup$     & 200 & 200 & 200 & 200 & 200\\
$out\_factor$       & 0.05  & 0.05  & 0.05  & 0.05  & 0.05  \\
$noise\_rate$       & 0.05  & 0.05  & 0.05  & 0.05  & 0.05   \\
\hline
\end{tabular}
\end{table}

\begin{table}
\centering
\caption{Diversity in the result set for artificial data. The value is the $qual(\mathcal{H}, \mathcal{F})$ where $\mathcal{H}$ is the set of hidden patterns in the artificial data and $\mathcal{F}$ is the set of found patterns by the algorithm The character '-' means that the algorithm exceeds the time limit of 5 minutes.\label{tab:diversity_artificial}}
\setlength{\tabcolsep}{2.5pt}
\renewcommand{\arraystretch}{2.5}
\rotatebox{90}{
\scalebox{0.75}{
\begin{tabular}{|c|c|c|ccccccc|ccc|ccccccc|ccccc|}
\hline
\multirow{2}{*}{Data} & \multirow{2}{*}{$minSupp$} & \multirow{2}{*}{\textsc{SD-Map}} & \multicolumn{7}{c|}{\algo{}} & \multicolumn{3}{c|}{\textsc{BeamSearch}} & \multicolumn{7}{c|}{\textsc{Misere}} & \multicolumn{5}{c|}{\textsc{SSDP}} \\
 & &  & 1K & 5K & 10K & 50K & 100K & 500K & 1,000K & 10 & 100 & 500 & 1K & 5K & 10K & 50K & 100K & 500K & 1,000K & 100 & 500 & 1K & 5K & 10k \\
 \hline
5000\_10\_200 & 50 &1 & 1 & 1 & 1 & 1 & 1 & 1 & 1 & 1 & 0.38 & 1 & 1 & 1 & 1 & 1 & 1 & 1 & 1 & 0.72 & 0.72 & 0.72 & 0.72 & 0.72  \\
5000\_50\_50 & 50 & - & 0.82 & 1 & 1 & 1 & 1 & - & - & 0.99 & 1 & - & 0.73 & 0.94 & 0.93 & 0.99 & 0.99 & 0.99 & 1 & 0.48 & 0.72 & 0.72 & 0.72 & 0.72 \\
5000\_50\_200 & 50 & - & 1 & 1 & 1 & 1 & 1 & 1 & 1 & 1 & 1 & - & 0.7 & 0.89 & 0.97 & 0.98 & 0.98 & 0.99 & 1 & 0.68 & 0.72 & 0.72 & 0.72 & 0.72 \\
20000\_10\_200 & 100 & 1 & 1 & 1 & 1 & 1 & 1 & 1 & 1 & 1 & 1 & - & 0.77 & 0.89 & 0.96 & 0.97 & 1 & 1 & - & 0.73 & 0.73 & 0.73 & 0.72 & -  \\
20000\_50\_200 & 100 & - & 0.69 & 1 & 1 & 1 & - & - & - & - & - & - & 0.6 & 0.86 & 0.87 & 0.95 & 0.98 & - & - & 0.61 & 0.72 & 0.72 & 0.72 & - \\
\hline
\end{tabular}
}
}
\end{table}

\subsection{\textsc{SD-Map}}
\textsc{SD-Map*}, an improvement of \textsc{SD-Map}, is considered as the most efficient exhaustive method for subgroup discovery [\cite{DBLP:conf/pkdd/AtzmullerP06,DBLP:conf/ismis/AtzmullerL09}]. It employs the FP-Growth principle to enumerate the search space [\cite{DBLP:journals/datamine/HanPYM04}]. It operates a greedy discretization as a pre-processing step to handle numerical data. It can consider several quality measures to evaluate the interestingness of a subgroup ($WRAcc$, $F1$ score, etc.). The source code is available at \url{http://www.vikamine.org}. 

\medskip\noindent\textbf{Runtime.}
\textsc{SD-Map*} is very efficient when dealing with dataset of reasonable search space size. We empirically study the scalability of this algorithm compared to those of \algo{} for several numbers of iterations. Figure~\ref{fig:runtime-SDMAP}~(a) displays the runtime of \textsc{SD-Map*} on the \textit{Mushroom} dataset when varying the minimum support threshold. Clearly, for high minimum support thresholds, \textsc{SD-Map*} is able to provide the results quickly. However, the runtime is exponential w.r.t. this threshold, and thus this algorithm cannot be applied to extract small subgroups. Conversely, \algo{} is tractable for very low minimum support thresholds: Many iterations can be performed. Figure~\ref{fig:runtime-SDMAP}~(b) displays the runtimes for the \textit{Ionosphere} data and once again \algo{} is able to perform lots of iteration in a linear time w.r.t. the minimum support threshold whereas \textsc{SD-Map*} fails.

\begin{figure}[t]\centering
\setlength{\tabcolsep}{0pt}
\renewcommand{\arraystretch}{1.5}
\begin{tabular}{cc}
\includegraphics[width=.5\columnwidth]{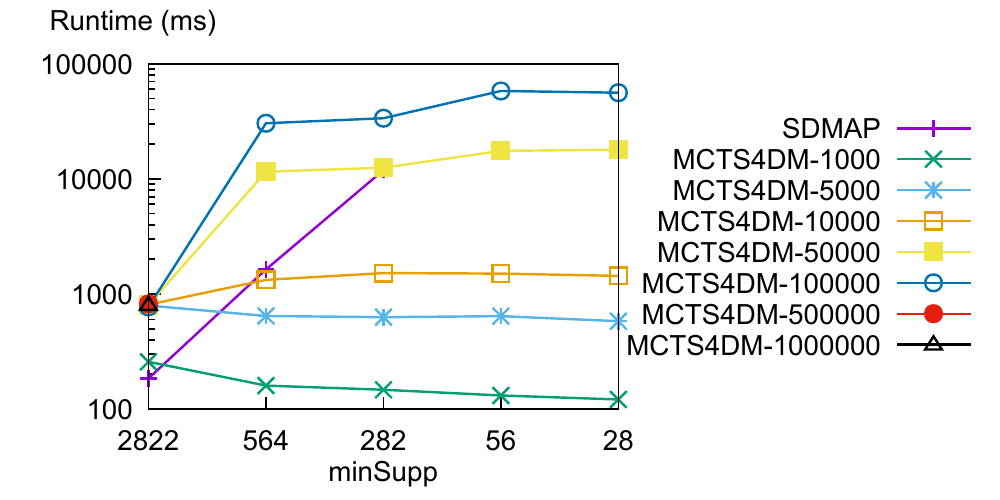}&
\includegraphics[width=.5\columnwidth]{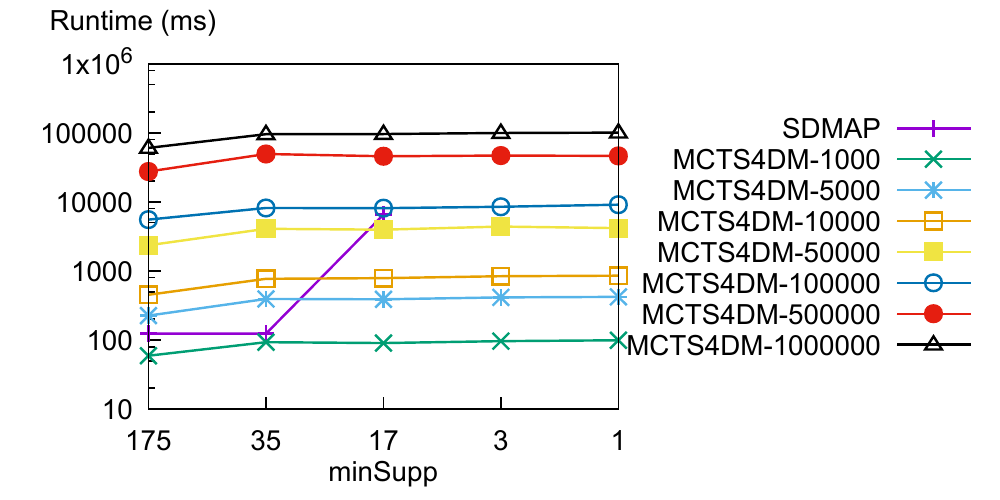}\\
(a) \textit{Mushroom} & (b) \textit{Ionosphere}
\end{tabular} 
	\caption{Runtime of \textsc{SD-Map} and \algo{} when varying \textit{minSupp}.\label{fig:runtime-SDMAP}}
\end{figure}

\medskip\noindent\textbf{Redundancy and diversity in the result set.}
\textsc{SD-Map*} is an exhaustive search, thus the diversity of the result set is either perfect if the run can finish or null: Table~\ref{tab:diversity_artificial} gives the diversity using the formula of Definition~\ref{def:diversity-measure} on artificial data since the ground truth is known. However when dealing with numerical attributes, \textsc{SD-Map*} does not ensure a perfect diversity anymore. Indeed, since it handles numerical attributes by performing a discretization in a pre-processing step, there is no guarantee to extract the best patterns. For example, in the \textit{BreastCancer} dataset, the quality measure of the best subgroup extracted by \textsc{SD-Map*} is 0.18 whereas \algo{} has found a subgroup whose quality measure is 0.21 with $50,000$ iterations in only $0.213$ms.
Figure~\ref{fig:redundancy-SDMAP}~(a) and Figure~\ref{fig:redundancy-SDMAP}~(b) show  the redundancy of the result set (computed with the formula in Definiton~\ref{def:redundancy}) extracted on the \textit{Mushroom} dataset respectively with $minSupp = 264$ and $minSupp= 282$. Obviously, the lower the maximum similarity threshold $\Theta$, the more redundant the result set. Compared to \textsc{SD-Map*}, \algo{} produces few redundancy when performing few iterations, but few iterations are not enough to provide good results:  The more  iterations, the more redundancy.
Surprisingly, the result set of \algo{} can be more redundant than those of \textsc{SD-Map*} that represents our baseline. Indeed, the set of redundant patterns for the main local optima is larger than for other small local optima, i.e., there are many more patterns that are similar with the main local optima than with small local optima. Since \algo{} generally finds at first the main local optima, the redundancy measure is higher than those of \textsc{SD-Map*} because there are, in proportion, more redundant subgroups in the result set than local optima.
When $minSupp$ decreases, \textsc{SD-Map*} becomes more redundant compared to some \algo{} runs.

\begin{figure}[t]\centering
\setlength{\tabcolsep}{0pt}
\renewcommand{\arraystretch}{1.5}
\begin{tabular}{cc}
\includegraphics[width=.5\columnwidth]{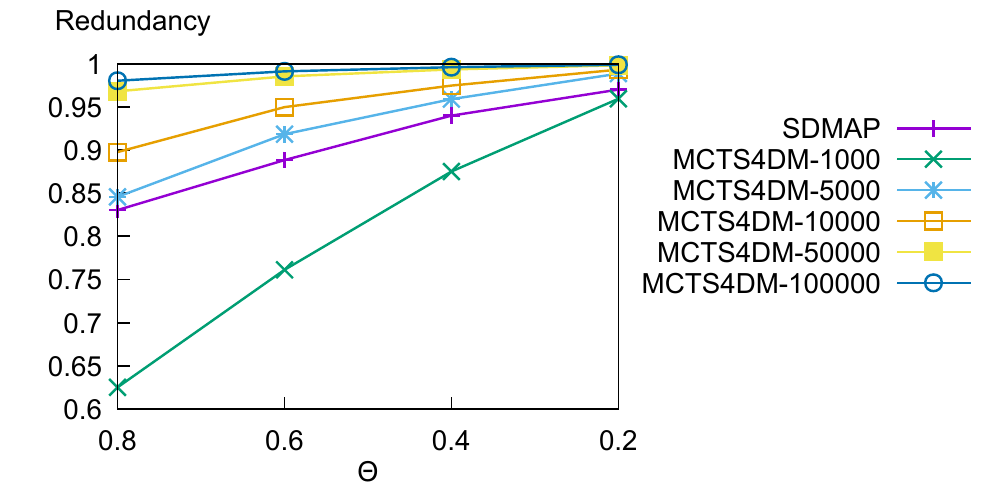}&
\includegraphics[width=.5\columnwidth]{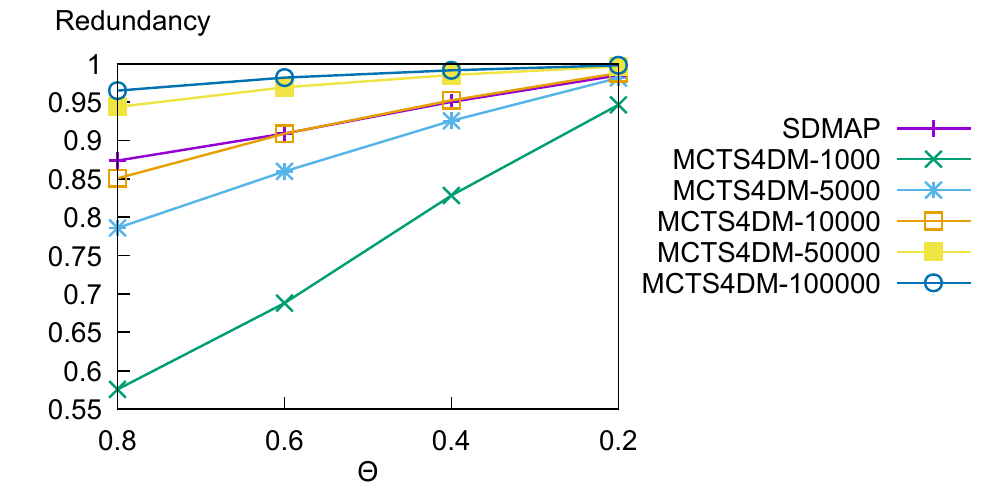}\\
(a) minSupp = 564 & (b) minSupp=282
\end{tabular} 
	\caption{The redundancy in the result set for the \textit{Mushroom} data varying $\Theta$.\label{fig:redundancy-SDMAP}}
\end{figure}

\medskip\noindent\textbf{Memory usage.}
Figure~\ref{fig:memory}~(a) displays the memory usage of our algorithm \algo{} on the \textit{Mushroom} data with different numbers of iterations. As expected, the more iterations, the higher the memory usage. It grows  linearly with the number of iterations (the creation of the storage structures avoids to see the linear growth of the memory during the first iterations). Figure~\ref{fig:memory}~(b) displays the memory usage when varying the minimum support threshold in the \textit{Mushroom} dataset for all the considered algorithms. Here, we only discuss the case of \textsc{SD-Map*} compared to \algo{}. Although \textsc{SD-Map*} is an exhaustive search, its memory usage is similar to (but slightly lower than) those of \algo{} with $100k$ iterations. It confirms that the implementation of \textsc{SD-Map*} is efficient.

\begin{figure}[t]\centering
\setlength{\tabcolsep}{0pt}
\renewcommand{\arraystretch}{1.5}
\begin{tabular}{cc}
\includegraphics[width=.41\columnwidth]{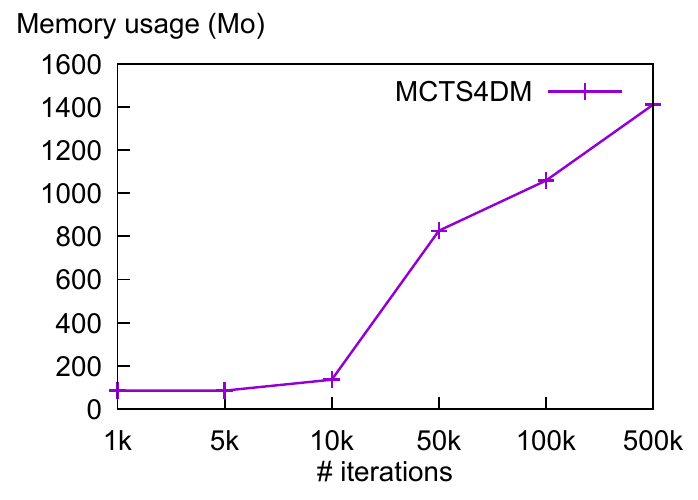}&
\includegraphics[width=.59\columnwidth]{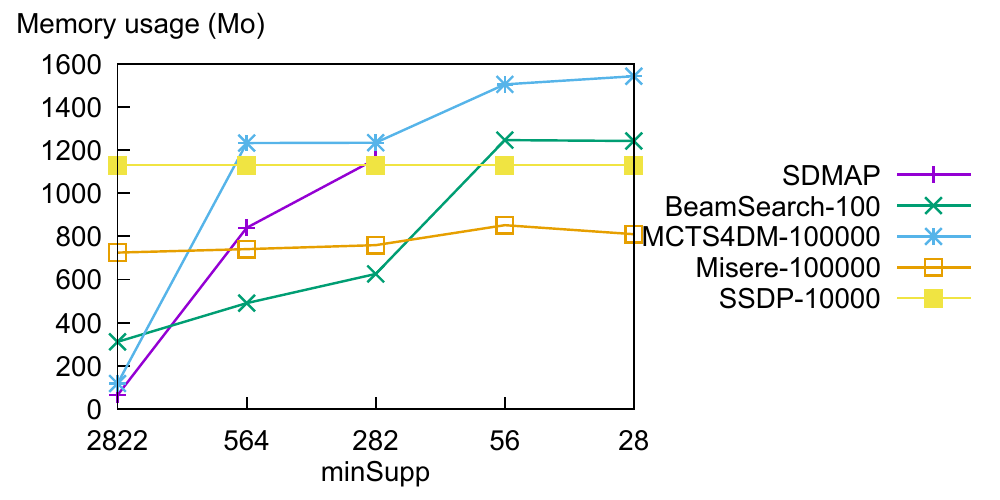}\\
(a) \algo{} & (b) All algorithms
\end{tabular} 
	\caption{The memory usage in the \textit{Mushroom} dataset.\label{fig:memory}}
\end{figure}

\subsection{Beam search}
The beam search strategy is the most popular heuristic method in subgroup discovery. Cortana is a tool that enables to run SD tasks with beam search approaches and its source code is available at \url{http://datamining.liacs.nl/cortana.html}. Beam search, originally introduced in \cite{lowerre1976harpy}, is a greedy method that partially explores the search space with several hill climbings run in parallel. It proceeds in a level-wise approach considering at each level the best subgroups to extend at the next level. The number of subgroups that are kept to be extended at the next level is called the beam width. 

\medskip\noindent\textbf{Runtime.}
By definition, beam search can only find, yet very quickly, local optima reachable from the most general pattern with a hill climbing. Figure~\ref{fig:runtime-BeamSearch} shows the runtimes with different beam widths. Obviously, the larger the beam width, the longer the runtimes. However, even with a beam width of 500, the runtime is lower than those of \algo{} with $100k$ iterations. This is due to the greedy nature of beam search that expands subgroups only if the quality measure increases. However, the local optima that are located deeper in the search space are often missed since the quality measure is not monotone. For large data, such as \textit{Bibtex}, the beam search is not tractable since it is required to expand all the first level of the search tree to build up the beam. Thus, in our settings, the timeout of 5 minutes is reached with beam search whereas  \algo{} can proceed to $100k$ iterations in 4 minutes.

\begin{figure}[t]\centering
\includegraphics[width=.7\columnwidth]{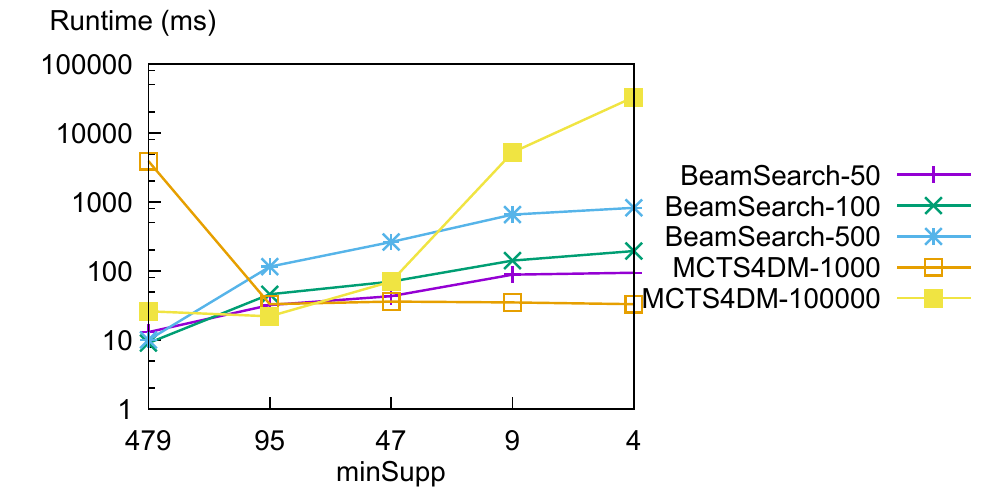}
	\caption{Runtime of the beam search exploration and \algo{} when varying \textit{minSupp} in the \textit{tictactoe} dataset.\label{fig:runtime-BeamSearch}}
\end{figure}

\medskip\noindent\textbf{Redundancy and diversity in the result set.}
Due to the greedy approach of the beam search, the redundancy in the result set is the main problem. Figure~\ref{fig:redundancy-BeamSearch}~(a) compares the redundancy in the \textit{TicTacToe} data obtained with several beam searches and with \algo{}. Clearly, the beam search leads to a more redundant result set than \algo{}. This remark holds for all data we experimented with. For example, in the \textit{Olfaction} dataset, there is a high difference in the redundancy in the result set obtained with a beam search as well (Figure~\ref{fig:redundancy-BeamSearch}~(b)).
Besides this high redundancy in the result set with a beam search, the diversity is not as good as with \algo{}. Even if in Table~\ref{tab:diversity_artificial}, the beam search extracts the local optima for some datasets, it may require large beam widths that are time consuming. Figure~\ref{fig:diversity-BeamSearch} illustrates the diversity for the \textit{BreastCancer} data with $\Theta = 0.2$. With $100k$ iterations, \algo{} leads to a much more diverse result set than a beam search.
\begin{figure}[t]\centering
\setlength{\tabcolsep}{0pt}
\renewcommand{\arraystretch}{1.5}
\begin{tabular}{cc}
\includegraphics[width=.5\columnwidth]{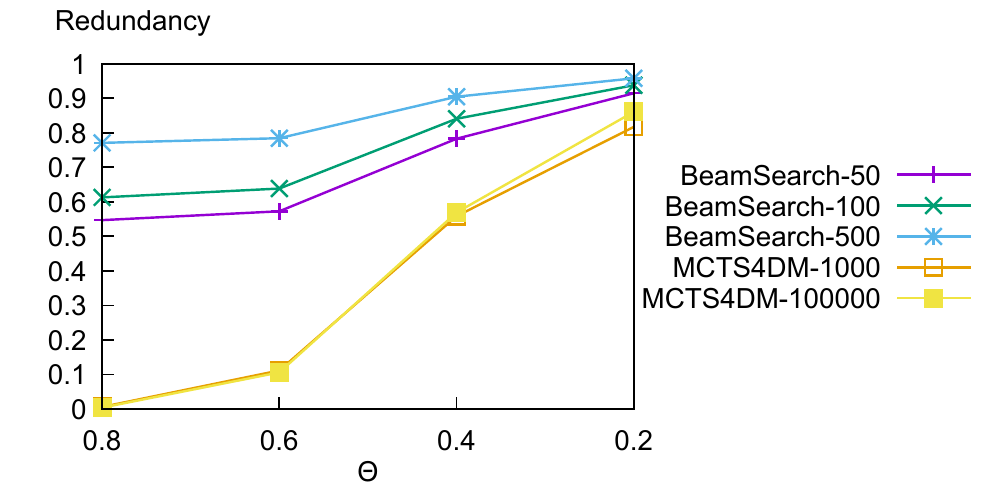}&
\includegraphics[width=.5\columnwidth]{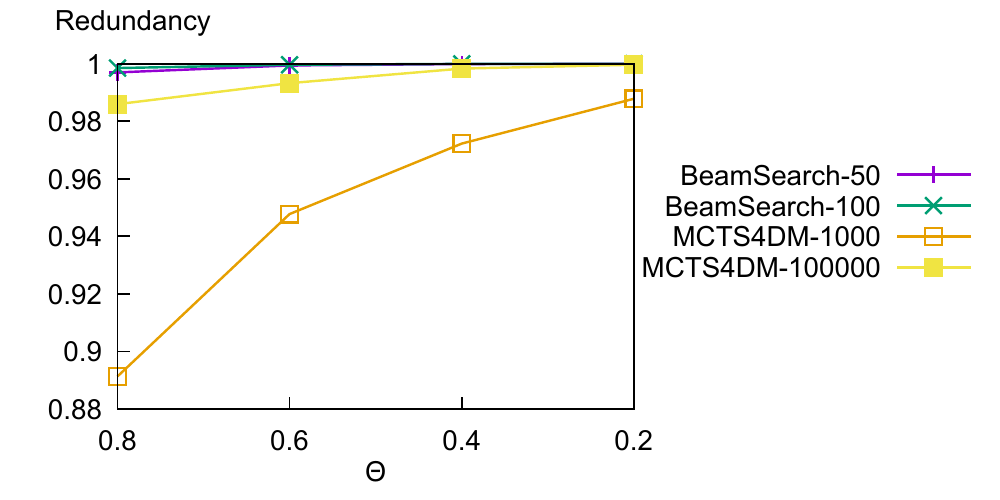}\\
(a) \textit{TicTacToe} ($minSupp = 47$) & (b) \textit{Olfaction} ($minSupp = 84$)
\end{tabular} 
	\caption{The redundancy in the result set for the \textit{TicTacToe} and \textit{Olfaction} data for the beam search strategy.\label{fig:redundancy-BeamSearch}} 
\end{figure}

\begin{figure}[t]\centering
\includegraphics[width=.7\columnwidth]{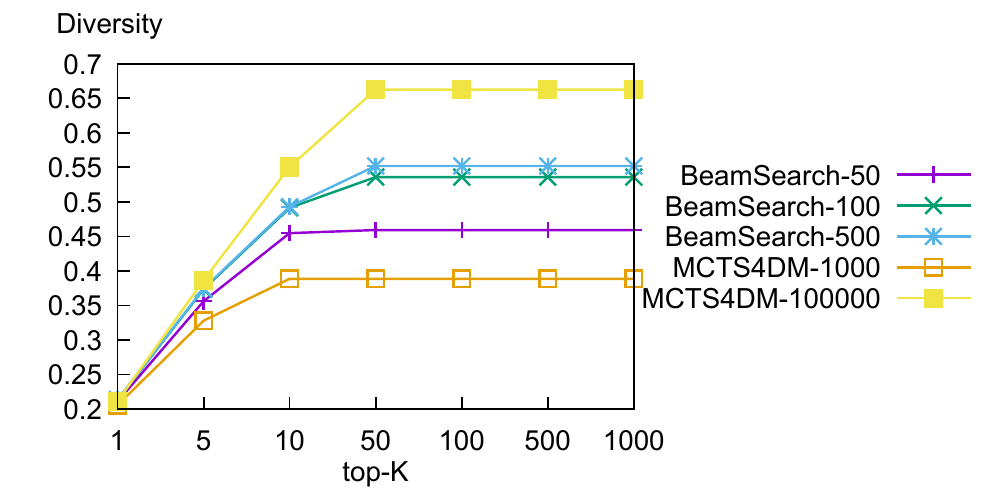}
	\caption{The diversity of the result set in \textit{BreastCancer} data when $\Theta = 0.2$.\label{fig:diversity-BeamSearch}}
\end{figure}

\medskip\noindent\textbf{Memory consumption.}
The size of the set of patterns extracted with a beam search is generally lower than the size of the set of patterns obtained with tens of thousands iterations with \algo{}. Thus, the memory usage is lower for a beam search. Figure~\ref{fig:memory} displays the memory usage of a beam search with a beam width set to 100 in the \textit{Mushroom} data. We observe that the memory usage increases similarly to (but it is lower than) those of \algo{} when varying the minimum support thresholds. 

\subsection{Evolutionary algorithms}
The evolutionary approaches aim at solving problems imitating the process of natural evolution. Genetic algorithms are a branch of the evolutionary approaches that use a fitness function to select which individuals to keep at the next generated population~[\cite{holland1975adaptation}]. In this empirical study, we evaluate the efficiency of \algo{} from \cite{Lucas2017} against the evolutionary algorithm \textsc{SSDP}. 

\medskip\noindent\textbf{Runtime.}
\textsc{SSDP} is free from the minimum support constraint: It explores the whole search space without pruning w.r.t. the support of the patterns. Therefore, the runtimes of \textsc{SSDP} are the same for all minimum support thresholds (Figure~\ref{fig:runtime-SSDP}). However, when varying the population size, it comes with large changes in the runtimes. The runtimes of \textsc{SSDP} are quite similar to those of \algo{} when varying the number of iterations. However, in general \textsc{SSDP} is not scalable when considering a large population size.

\begin{figure}[t]\centering
\includegraphics[width=.7\columnwidth]{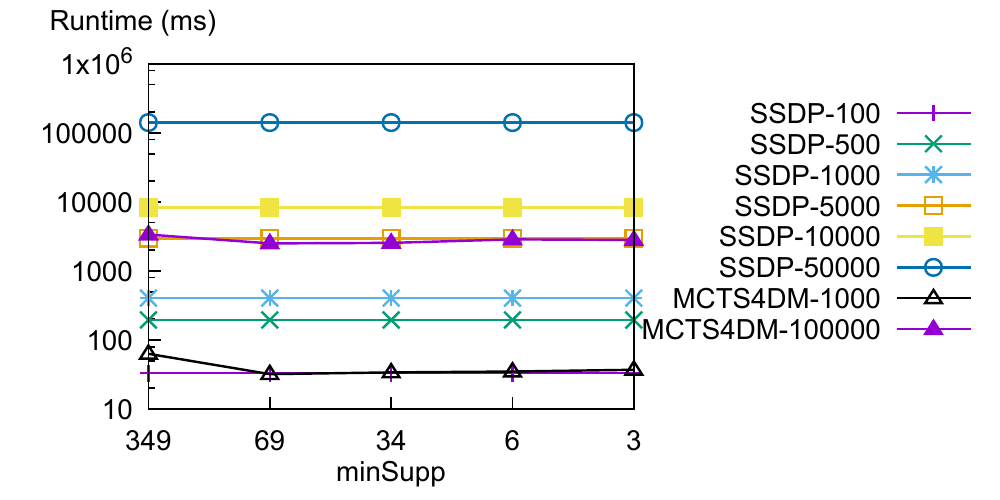}
	\caption{Runtime of \textsc{SSDP} and \algo{} when varying \textit{minSupp} in \textit{BreastCancer}.\label{fig:runtime-SSDP}}
\end{figure}

\medskip\noindent\textbf{Redundancy and diversity in the result set.} 
On one hand, \textsc{SSDP} seems to provide less redundant pattern sets, due to the mutation and cross-over operations of this evolutionary algorithm. Figure~\ref{fig:diversity-SSDP}~(a) deals with the \textit{Iris} data with $minSupp = 7$: The redundancy of \textsc{SSDP} is generally better than those of our algorithm \algo{}. This is the same conclusion in Figure~\ref{fig:redundancy-SSDP}~(b) for \textit{Mushroom} with $minSupp = 56$.  
On the other hand, the diversity in the result set of \textsc{SSDP} is lower than those of \algo{}. In Table~\ref{tab:diversity_artificial}, \textsc{SSDP} fails to extract all hidden patterns in our artificial data. Figure~\ref{fig:diversity-SSDP}~(a) and Figure~\ref{fig:diversity-SSDP}~(b) display the same result for benchmark datasets. Clearly, \algo{} is able to extract much more interesting subgroups than \textsc{SSDP}. Thus, even if the result set of \algo{} can be redundant, it provides a more diverse set of patterns compared to the result set extracted by \textsc{SSDP}. This is due to the population size that is not enough large to provide a high diversity (but \textsc{SSDP} is not tractable for large population sizes).

\begin{figure}[t]\centering
\setlength{\tabcolsep}{0pt}
\renewcommand{\arraystretch}{1.5}
\begin{tabular}{cc}
\includegraphics[width=.5\columnwidth]{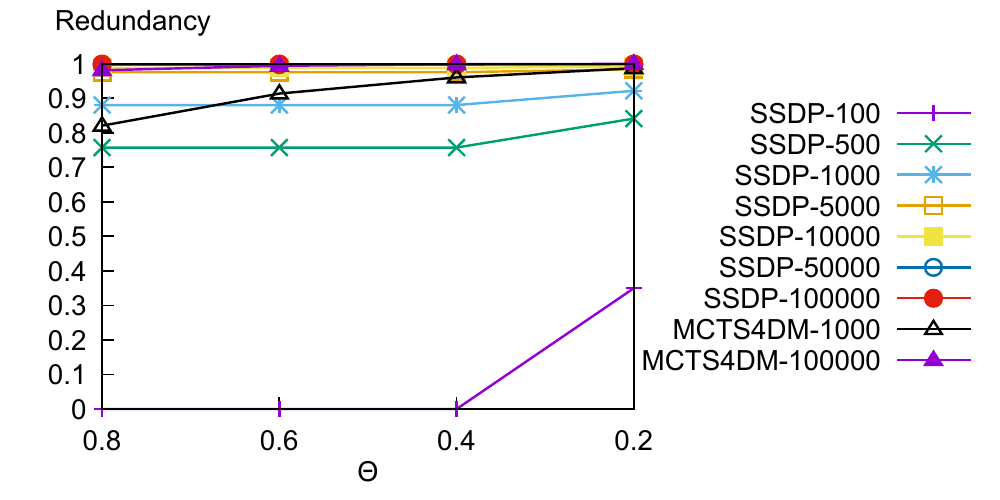}&
\includegraphics[width=.5\columnwidth]{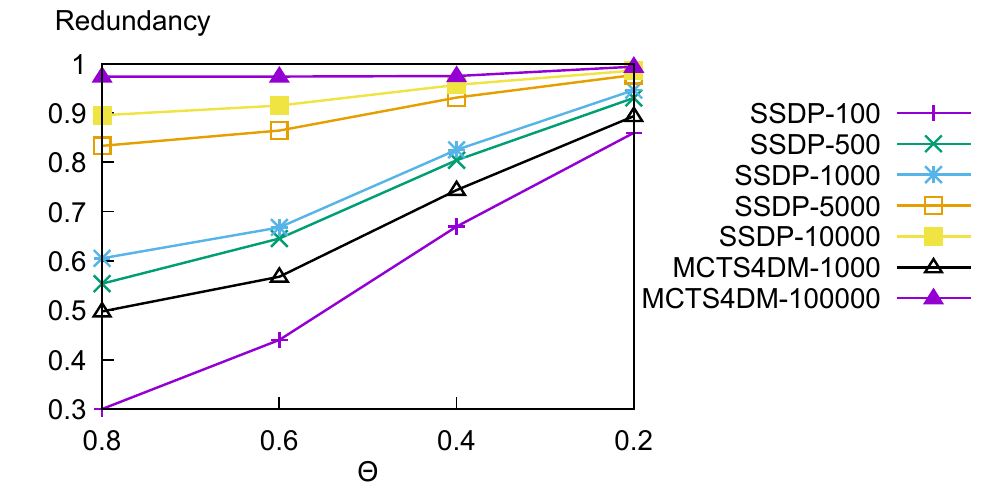}\\
(a) \textit{Iris} ($minSupp = 7$) & (b) \textit{Mushroom} ($minSupp = 56$)
\end{tabular} 
	\caption{The redundancy in the result set for the iris and mushroom data.\label{fig:redundancy-SSDP}}
\end{figure}

\begin{figure}[t]\centering
\setlength{\tabcolsep}{0pt}
\renewcommand{\arraystretch}{1.5}
\begin{tabular}{cc}
\includegraphics[width=.5\columnwidth]{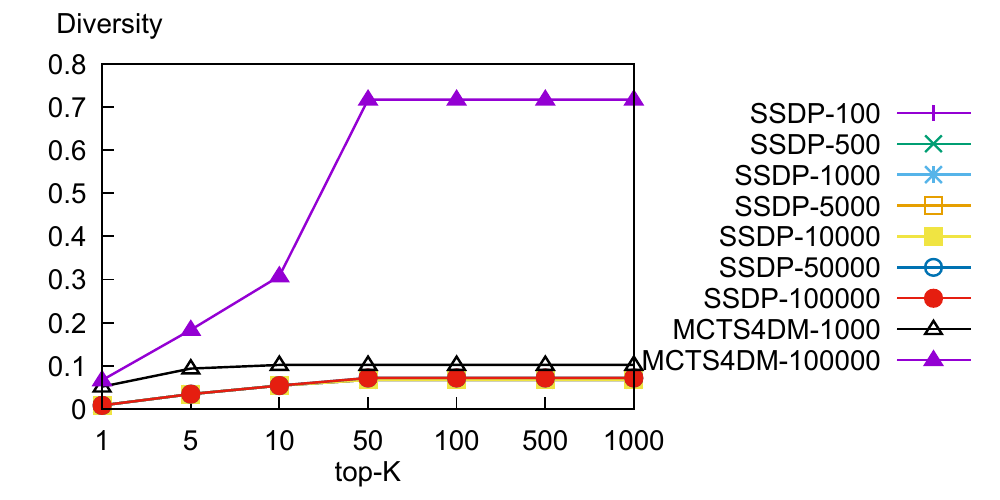}&
\includegraphics[width=.5\columnwidth]{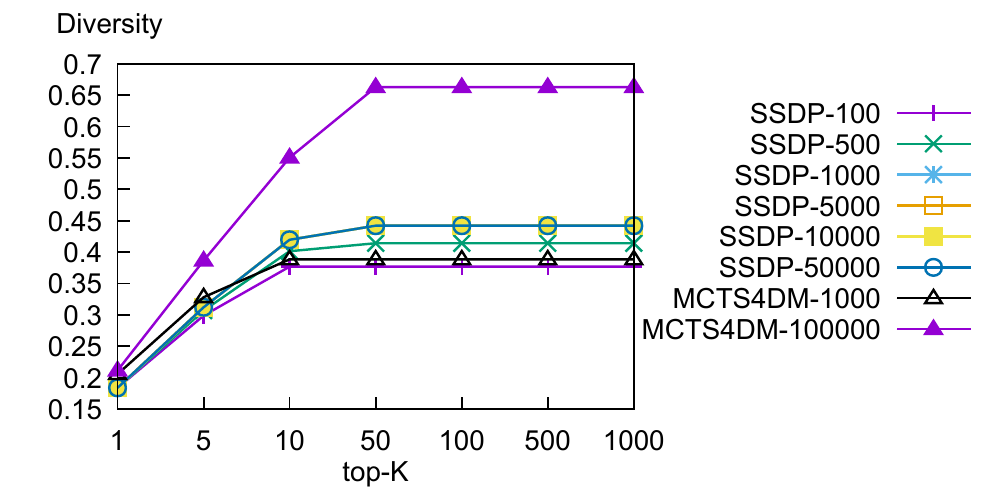}\\
(a) \textit{Emotions} ($minSupp = 5$) & (b) \textit{BreastCancer} ($minSupp = 34$)
\end{tabular} 
	\caption{The diversity of the result set when $\Theta = 0.2$.\label{fig:diversity-SSDP}}
\end{figure}

\medskip\noindent\textbf{Memory consumption.}
The memory usage of \textsc{SSDP} depends on the size of the population. In \textit{mushroom}, when considering a population of size $1,000$, the memory usage is higher than those of \algo{} with $100k$ iterations for high minimum support thresholds but it is lower for low minimum support thresholds (see Figure~\ref{fig:memory}). Indeed, since \textsc{SSDP} does not use any minimum support threshold, its memory usage is independent w.r.t. $minSupp$.

\subsection{Sampling approach}
Sampling methods are useful to provide interactive applications. Indeed, they enable a result anytime. We experiment with the sampling algorithm \textsc{Misere}~[\cite{DBLP:conf/pkdd/GayB12,DBLP:conf/icdm/EghoGBVC15,DBLP:journals/kais/EghoGBVC17}]. Its principle consists in drawing uniformly an object from the data, and then uniformly pick one of its possible generalizations.
Each sample is independent and thus a pattern can be drawn several times.
We chose \textsc{Misere} as it can consider any pattern quality measure (in contrast to other sampling approaches such as \cite{DBLP:conf/ida/MoensB14,DBLP:conf/kdd/BoleyLPG11}), and it performs very well.

\medskip\noindent\textbf{Runtime.}
Since this strategy consists in randomly drawing patterns, the runtime is linear with the number of draws. Varying the minimum support thresholds does not really impact the runtime (Figure~\ref{fig:runtime-Misere}). An iteration with \algo{} is almost only twice much longer than a draw with \textsc{Misere}. This is explained by the fact that  \textsc{Misere} only draws one pattern at once without additional memory (i.e., the Monte Carlo tree). Conversely, in one iteration, MCTS additionally performs \textsc{Select}, \textsc{Expand}, \textsc{Memory} and \textsc{Update} steps.

\begin{figure}[t]\centering
\includegraphics[width=.7\columnwidth]{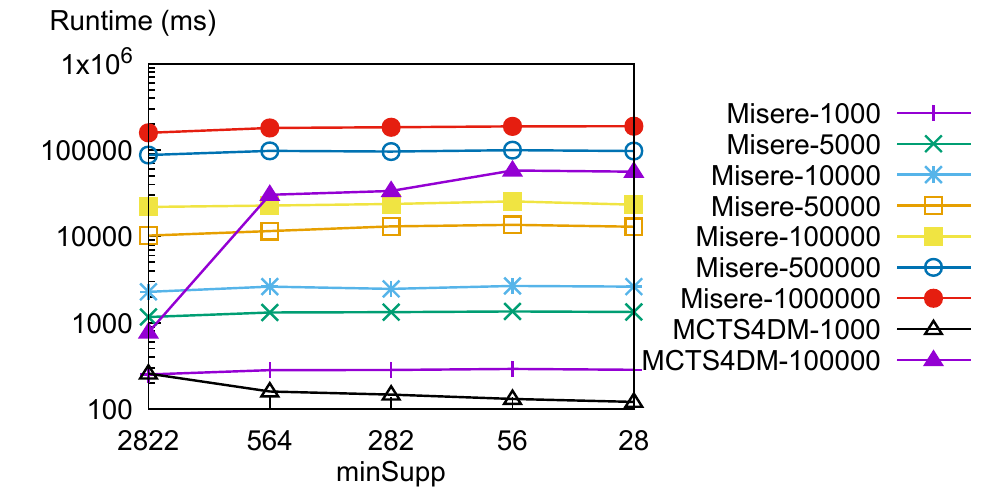}
	\caption{Runtime of \textsc{Misere} and \algo{} when varying \textit{minSupp} on the \textit{Mushroom} dataset.\label{fig:runtime-Misere}}
\end{figure}

\medskip\noindent\textbf{Redundancy and diversity in the result set.}
Since \textsc{Misere} proceeds in independent draws of patterns without exploiting the result of the previous draws of patterns, it leads to a result set that contains little redundancy. Indeed, it can draw an interesting pattern that is close to its local optimum, but it would not try to find this optimum at the next draws.
 Figure~\ref{fig:redundancy-Misere} illustrates this: The result set is much less redundant than those of \algo{}. For example, considering a result set containing $1,000$ draws of patterns and another obtained with $1,000$ iterations from \algo{}.   \algo{} returns pattern set 10 times more redundant than \textsc{Misere}. However, since \textsc{Misere} does not exploit the result of the previous draws, it leads to less diversity. In Table~\ref{tab:diversity_artificial}, \textsc{Misere} may require lots of draws to find all local optima. Figure~\ref{fig:redundancy-Misere}~(b) shows for \textit{Bibtex} that the diversity is better with \algo{} than with \textsc{Misere}. Contrary to \algo{}, there is no guarantee that \textsc{Misere} will explore the whole search space, even given a large computational budget. Nevertheless, in Table~\ref{tab:diversity_artificial}, we can notice that, in practice, \textsc{Misere} can extract all patterns hidden in artificial data, but it might require a lot of draws to find them.

\begin{figure}[t]\centering
\setlength{\tabcolsep}{0pt}
\renewcommand{\arraystretch}{1.5}
\begin{tabular}{cc}
\includegraphics[width=.5\columnwidth]{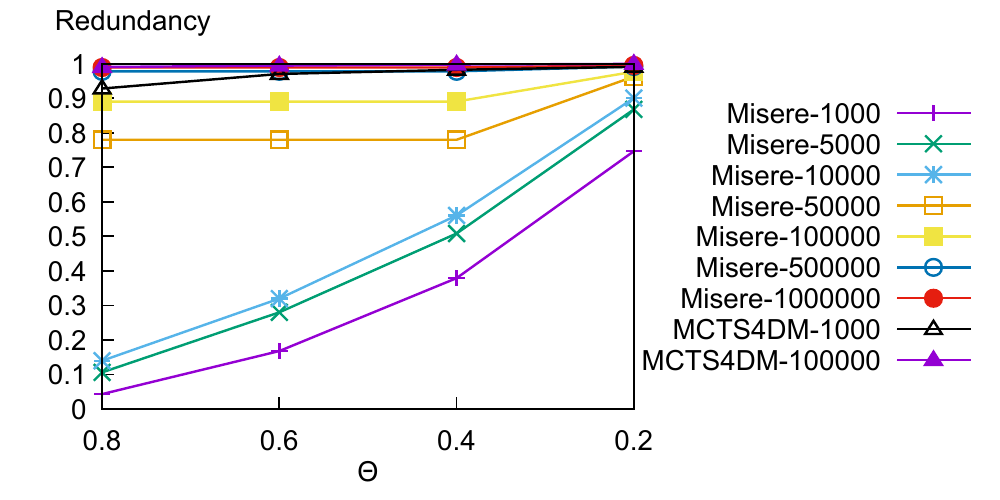}&
\includegraphics[width=.5\columnwidth]{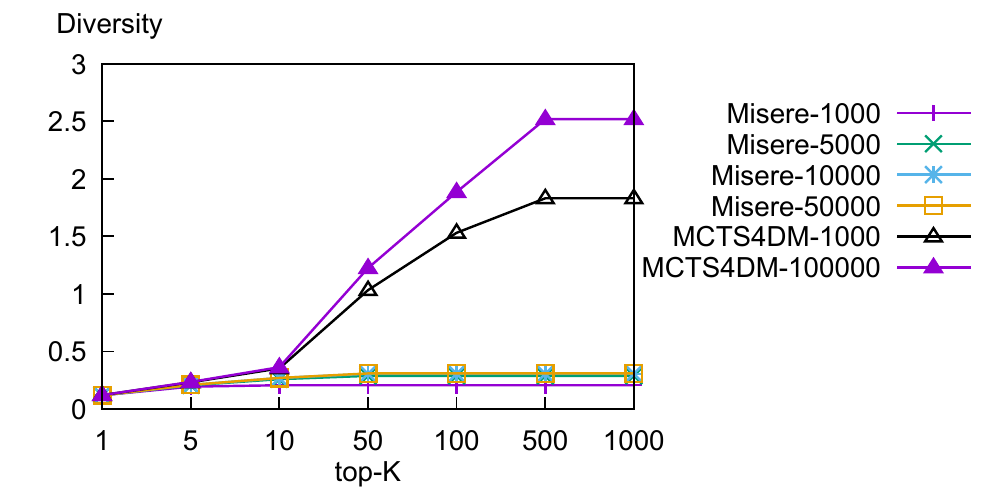}\\
(a) Redundancy in \textit{Cal500}  & (b) Diversity in \textit{Bibtex}
\end{tabular} 
	\caption{The redundancy and the diversity in the result set for the \textit{Cal500} and \textit{Bibtex} data.\label{fig:redundancy-Misere}}
\end{figure}

\medskip\noindent\textbf{Memory consumption.}
As expected, since this sampling method performs independent draws, the memory usage is low. In our settings, only the patterns that are drawn are stored. Figure~\ref{fig:memory} illustrates the memory usage of \textsc{Misere} when it has randomly picked $100k$ patterns. It is constant w.r.t. the minimum support thresholds, and this is the exploration method that requires the less memory for low minimum support thresholds.

\subsection{Considering several measures}

\algo{} can consider any pattern quality measures. Up to now, we experimented with the popular \textit{WRAcc} measure only. We empirically evaluate \algo{} with several quality measures that are also used in SD. We consider some of the quality measures available in Cortana: The entropy, the F1 score, the Jaccard coefficient and the accuracy (or precision). The measures we use are not equivalent since they do not sort the patterns in the same order : Each measure induces a specific profile on the pattern space. 

\algo{} is not measure-dependent since it does not use any prior knowledge to explore the search space. During the first iterations, \algo{} randomly samples the search space, then once it has an estimation -- that is usually rather not reliable at the beginning -- of the distribution of the quality measure on the pattern space, it biases the exploration to focus on the promising areas (\emph{exploitation}) and the areas that have been rarely visited (\emph{exploration}). The strategies we developed are useful to handle the specific profile induced by a quality measure on the pattern space, e.g., if there are lots of local optima in the search space the exploration strategy should be different than if there are few local optima. For instance, the \textit{mean-update} strategy is the most efficient strategy when we are facing a pattern space with lots of local optima since it enables to exploit the areas that are deemed to be interesting in average. Thus, \algo{} can be used with any quality measure. The choice of the strategies only impacts how fast it will find the interesting patterns.

We compare our approach with the sampling method \textsc{Misere} which is the most efficient opponent based on the previous results (see Figure~\ref{fig:measures}). We experiment on the \textit{Mushroom} dataset to reach low minimum support thresholds using the four quality measures (the entropy, the F1 score, the Jaccard coefficient and the accuracy). The results suggest that \algo{} is able to provide a good diversity regardless the quality measure that is used. \algo{} finds a result set with a better diversity for all quality measures. The results on the other datasets are similar but not reported here.

\begin{figure}[t]\centering
\setlength{\tabcolsep}{0pt}
\renewcommand{\arraystretch}{1.5}
\begin{tabular}{cc}
\includegraphics[width=.5\columnwidth]{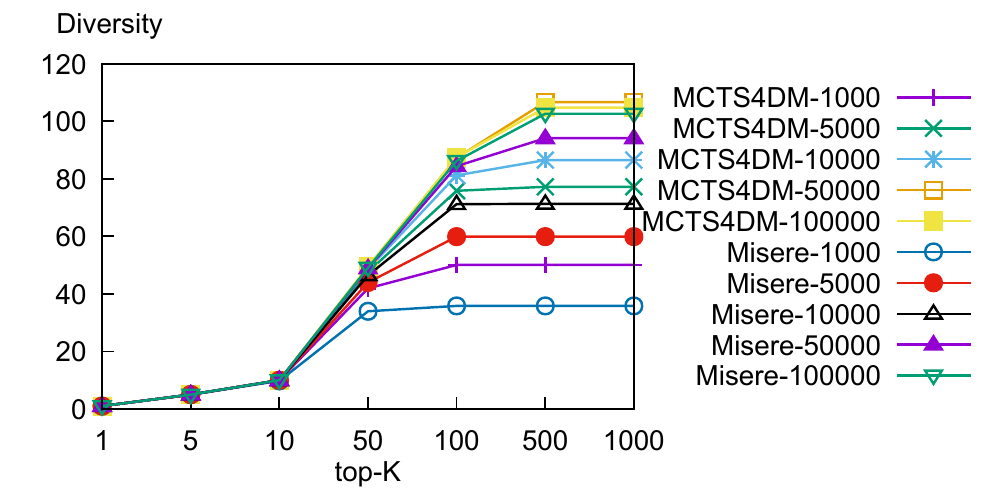}&
\includegraphics[width=.5\columnwidth]{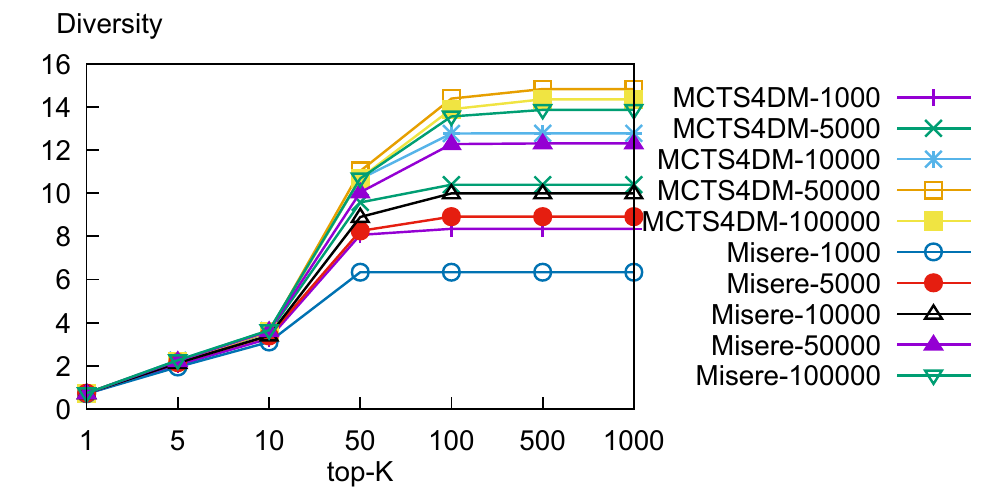}\\
(a) Entropy  & (b) F1 \\
\includegraphics[width=.5\columnwidth]{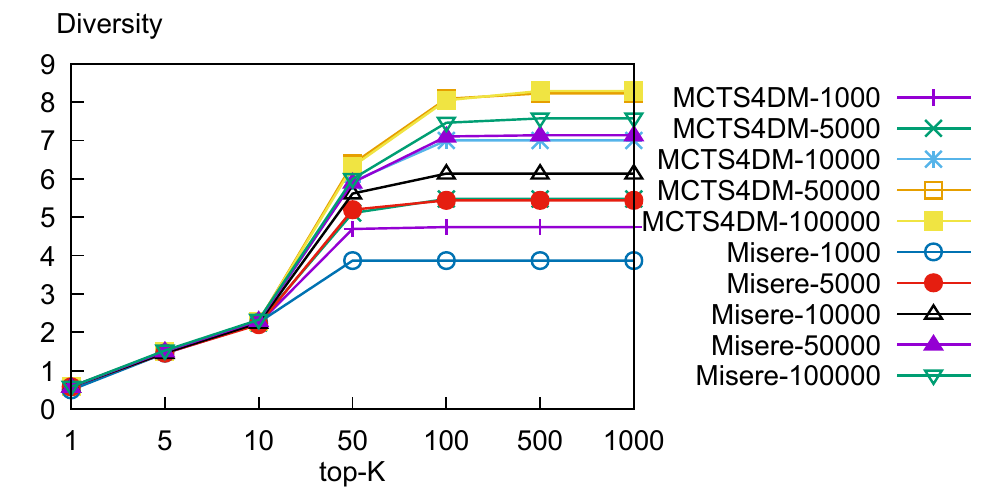}&
\includegraphics[width=.5\columnwidth]{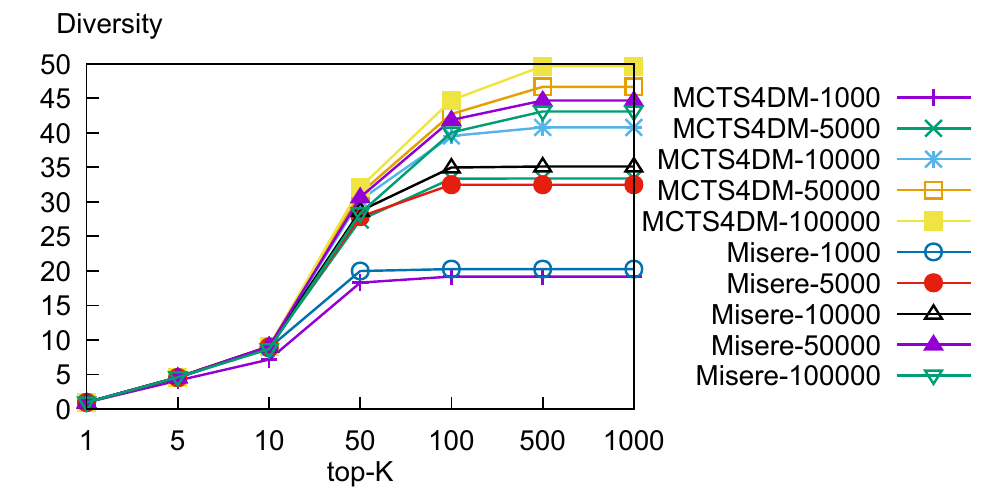}\\
(c) Jaccard coefficient  & (d) Accuracy
\end{tabular} 
	\caption{The diversity of the result set for several quality measures in the \textit{Mushroom} dataset.\label{fig:measures}}
\end{figure}

\section{Discussion}

Based on the empirical study reported in the two previous sections, we now provide a summary of the main results. First, we experimented with several strategies we defined for our algorithm \algo{}. Our conclusions are the following:
\begin{itemize}
\item \textsc{Select}: Concerning the choice of the upper confidence bound, it seems more suitable to use the \textit{SP-MCTS} for SD problems, although it has a limited impact. Activating LO leads to worse results, but with PU we are able to get more interesting patterns. 
This is a quite interesting fact as LO is a widely used technique in pattern mining (enumerate each pattern only once with a lectic order).
\item \textsc{Expand}: We advise to use the \textit{label-gen} strategy that enables to reach more quickly the best patterns, but it can require more computational time. 
\item \textsc{RollOut}: For nominal attributes, the \textit{direct-freq-roll-out} is an efficient strategy. However, when facing numerical attributes, we recommend to employ the \textit{large-freq-roll-out} since it may require a lot of time to reach the maximal frequent patterns. 
\item \textsc{Memory}: Using a memory strategy is essential since it enables to store the patterns encountered during the \textsc{RollOut} step. The \textit{top-1-memory} is enough to avoid to miss interesting patterns that are located deeper in the search space. 
\item \textsc{Update}:  When there are potentially many local optima in the search space, we recommend to set the \textit{mean-update} strategy for the \textsc{Update} step. Indeed it enables to exploit the areas that are deemed to be interesting in average. However, when there are few local optima among lots of uninteresting patterns, using \textit{mean-update} is not optimal since the mean of the rewards would converge to 0. In place, the \textit{max-update} should be used to ensure that an area containing a local optima is well identified.
\end{itemize}

Our second batch of experiments compared \algo{} with the main existing approaches for SD. For that, we experimented with one of the most efficient exhaustive search in SD, namely \textsc{SD-Map*}, a beam search, the recent evolutionary algorithm \textsc{SSDP} and a sampling method implemented in the algorithm \textsc{Misere}. The results suggest that \algo{} leads, in general, to a more diverse result set when an exhaustive search is not tractable. The greedy property of the beam search leads to a low diversity in the result set, and the lack of memory in sampling methods avoid to exploit interesting patterns to find the local optima (a pattern may be drawn several times). There is no guarantee that evolutionary algorithms and sampling approaches converge to the optimal pattern set even with an infinite computational budget. 

\smallskip

MCTS comes with several advantages but has some limits:
\begin{itemize}
\item[+] It produces a good pattern set anytime and it converges to an exhaustive search if given enough time and memory (a \textit{best-first search}). 

\item[+] It is agnostic of the pattern language and the quality measures: It handles numerical patterns without discretization in a pre-processing step and it still provides a high diversity using several quality measures.

\item[+] \algo{} is aheuristic: No hypotheses are required to run the algorithm whereas with some sampling methods, a probability distribution (based on the quality measure and the pattern space) has to be given as a parameter. 
 
\item[-] \algo{} may require a lot of memory. This memory usage becomes more and more important with the increase of the number of iterations. 
 
\item[-] Despite the use of UCB,  it is now well known that MCTS algorithms explore too much the search space.  As MCTS basically requires to expand all the children of a node before exploiting one of them, this problem is even stronger when dealing with very high branching factor (number of direct specializations of a pattern).
This problem has been in part tackled by the progressive widening approach that enables to exploit a child of a node before all of the other children of the node have been expanded [\cite{DBLP:conf/icml/GaudelS10,DBLP:journals/tciaig/BrownePWLCRTPSC12}].
\end{itemize}

\section{Conclusion}
Heuristic search of supervised patterns becomes mandatory with large datasets. However, classical heuristics lead to a weak diversity in pattern sets: Only few local optima are found. We advocate for the use of MCTS for pattern mining: An exploration strategy leading to \emph{``any-time"} pattern mining that can be adapted with different measures and policies. The experiments show that MCTS provides a much better diversity in the result set than existing heuristic approaches. For instance, interesting subgroups are found by means of a reasonable amount of iterations  and the quality of the result iteratively improves. 

MCTS is a powerful exploration strategy that can be applied to several, if not all, pattern mining problems that need to optimize a quality measure given a subset of objects. For example, \cite{DBLP:conf/ijcai/BelfodilKRK17} have already tuned MCTS4DM for mining convex polygon patterns in numerical data. 
In general, the main difficulties are to be able to deal with large branching factors, and jointly deal with several quality measures. This opens new research perspectives for mining more complex patterns such as sequences and graphs.

\section*{Acknowledgments}
The authors would like to thank the anonymous reviewers for their constructive and insightful comments. They also warmly thank  Sandy Moens, Mario Boley, Tarc\'{i}sio Lucas,  Renato Vimiero, Albrecht Zimmermann, Marc Plantevit,  Aimene Belfodil, Abdallah Saffidine, Dave Ritchie and especially C\'eline Robardet for discussions, advice or code sharing. This work has been partially supported by the European Union (GRAISearch, FP7-PEOPLE-2013-IAPP) and the \textit{Institut rh\^{o}nalpin des syst\`{e}mes complexes} (IXXI).

\bibliographystyle{plain}

\end{document}